\documentclass[a4paper,10pt]{article}

\pdfoutput=1 %

\usepackage{geometry}
\usepackage{amsthm,amsmath,amsfonts,amssymb}
\usepackage[round]{natbib}

\usepackage[colorlinks=false, urlcolor=citecolor, linkcolor=citecolor, citecolor=citecolor]{hyperref}
\usepackage[dvipsnames,svgnames,x11names]{xcolor}

\usepackage{dsfont}
\usepackage{caption}
\usepackage{graphicx}

\usepackage[utf8]{inputenc}
\usepackage[OT1]{fontenc}
 \usepackage[english]{babel}

\usepackage{algorithm,algpseudocode, algorithmicx}%
\algnewcommand{\Inputs}[1]{%
  \State \textbf{inputs:}
  \Statex \hspace*{\algorithmicindent}\parbox[t]{.8\linewidth}{\raggedright #1}
}
\algnewcommand{\Initialize}[1]{%
  \State \textbf{initialize:}
  \Statex \hspace*{\algorithmicindent}\parbox[t]{.8\linewidth}{\raggedright #1}
}

\DeclareRobustCommand{\R}{\ensuremath{\mathbb{R}}}

\usepackage{titlesec}
\titleformat{\subsection}[runin]%
       {\normalfont\bfseries}%
       {\thesubsection}%
       {0.5em}%
       {}%
       [.]%
\titleformat{\section}
  {\normalfont\fontsize{12}{15}\bfseries}{\thesection}{1em}{}

 \usepackage{tikz}
\usetikzlibrary{bayesnet} %

\usepackage{float}

\newcommand{\lb}{\left(}
\newcommand{\rb}{\right)}

\newcommand{\ls}{\left[}
\newcommand{\rs}{\right]}

\newcommand{\lc}{\left\{}
\newcommand{\rc}{\right\}}

\newcommand{\pr}{\mathrm{pr}}

\newcommand{\one}{\mathds{1}}
\newcommand{\E}{\mbox{E}}
\newcommand{\Var}{\mbox{Var}}

\newcommand{\mo}{\mu_{\omega}}
\newcommand{\vo}{\sigma^2_{\omega}}

 \newtheorem{theorem}{Theorem}[section]

\newtheorem{lemma}[theorem]{Lemma}

\usepackage{bm}

\title{\bf \Large A global-local approach for detecting hotspots in multiple-response regression}
\author{\large \hspace{-0.5cm} H\'el\`ene Ruffieux$^{a, }$\thanks{Corresponding author (\href{mailto:helene.ruffieux@mrc-bsu.cam.ac.uk}{helene.ruffieux@mrc-bsu.cam.ac.uk}).}\;,
Anthony C. Davison$^b$,
J\"{o}rg Hager$^c$, Jamie Inshaw$^{d}$, 
\\ \hspace{-0.5cm}  \large Benjamin P. Fairfax$^{e}$, Sylvia Richardson$^{a,f}$, Leonardo Bottolo$^{g,f,a}$\\
\footnotesize $^a$MRC Biostatistics Unit, University of Cambridge, Cambridge, United Kingdom\\
\footnotesize $^b$Ecole Polytechnique F\'ed\'erale de Lausanne (EPFL), Lausanne, Switzerland\\
\footnotesize $^c$Nestl\'e Research, EPFL Innovation Park, Lausanne, Switzerland\\
\footnotesize $^d$Wellcome Centre for Human Genetics, University of Oxford, Oxford, United Kingdom\\
\footnotesize \hspace{-1cm}  $^e$Department of Oncology, MRC Weatherall Institute for Molecular Medicine, University of Oxford, Oxford, United Kingdom\\
\footnotesize $^f$Alan Turing Institute, London, United Kingdom\\
\footnotesize $^g$Department of Medical Genetics, University of Cambridge, Cambridge, United Kingdom\\
}
\date{\normalsize \today}

\begin{document}
\sloppy
\maketitle

\begin{abstract} \,
We tackle modelling and inference for variable selection in regression problems with many predictors and many responses.  We focus on detecting  \emph{hotspots}, i.e., predictors associated with several responses. Such a task is critical %
in statistical genetics, as %
hotspot genetic variants shape the architecture of the genome by controlling the expression of many genes and may initiate decisive functional mechanisms underlying disease endpoints. %
Existing hierarchical regression approaches designed to model hotspots suffer from two %
limitations: their discrimination of hotspots is sensitive to the choice of top-level scale parameters for the propensity of predictors to be hotspots, and they do not scale %
to large predictor and response vectors, e.g., of dimensions $10^3-10^5$ in %
genetic applications. %
We address these shortcomings by introducing a flexible hierarchical regression framework that is tailored to the detection of hotspots and scalable to the above dimensions. %
Our proposal %
implements 
a fully Bayesian model for hotspots based on %
the horseshoe %
shrinkage prior. Its global-local formulation 
shrinks noise globally and hence accommodates the highly sparse nature of %
genetic analyses, while being robust to individual signals, thus leaving the effects of hotspots unshrunk. %
Inference is carried out using a fast variational algorithm coupled with a novel %
simulated annealing %
procedure that allows %
efficient exploration of multimodal distributions. 
\end{abstract}


\section{Introduction}\label{sec_introduction}

Understanding the genetic architecture of complex human traits is crucial for predicting health risks and developing effective therapies. %
Over the past two decades, thousands of genome-wide association studies have assessed the effects of millions of genetic variants on disease susceptibility. %
Among other important findings, these studies have revealed that most of the genetic variants involved in associations lie in non-coding regions of the genome \citep{ward2012interpreting, tak2015making}, which renders their functional interpretation difficult and suggests studying how they may affect clinical %
traits
through changes in gene regulation. %
This observation stimulated %
much of the current focus in statistical genetics on %
expression quantitative trait locus (eQTL) analyses, which assess how genetic variants control intermediate gene expression phenotypes. %
Genetic variants can act locally, affecting the expression of a nearby gene (\emph{cis}-eQTL) or they can alter expression of remote transcripts (\emph{trans}-eQTL). 
  Understanding by which mechanisms \emph{trans}-regulation can take place, via a local  \emph{cis} gene that acts on a whole network or via other means, is a subject of active debate \citep{%
  westra2013systematic, solovieff2013pleiotropy, brynedal2017large, yao2017dynamic}. 
  In particular, the detection of \emph{pleiotropic} variants, regulating the expression of tens or possibly hundreds of transcripts, is of great %
  interest: %
  such ``\emph{trans}-hotspot'' genetic variants may provide insight into the regulatory landscape of the transcriptome, and hence into the mechanisms shaping the evolution of the human genome. They may also %
 shed light on important functional processes underlying clinical traits and diseases.

  Despite these promises, %
  the locations and abundance
of master regulatory sites on the genome remain largely unknown. %
Indeed, most %
eQTL studies rely on conventional univariate screening, such as provided by MatrixEQTL \citep{shabalin2012matrix}, and have focused on detecting proximal \emph{cis} associations, %
either to limit the multiple testing burden or because the distal \emph{trans} associations uncovered would fail to replicate. Existing joint modelling approaches that directly model the response covariance \citep[e.g.,][]{yin2011sparse, bhadra2013joint} %
only provide partial solutions %
to modelling hotspots. For computational reasons, they %
are typically limited to the analysis of a few clinical phenotypes or require drastic preliminary dimension reduction that often dilutes or even discards weak but relevant signals. %

The present paper aims to provide an effective statistical tool to bridge this gap: %
it describes a joint modelling framework %
that is tailored to the detection of \emph{trans}-regulatory hotspots, %
while scaling to tens of thousands of molecular expression levels. %
The model consists of a series of sparse regressions linked in a hierarchical manner, which allows the borrowing of strength across all responses (molecular expression levels) and candidate predictors (genetic variants), a key benefit of the Bayesian hierarchical framework adopted. %
It provides information beyond pairwise associations of predictors and responses, and yields interpretable posterior measures of the propensity of predictors to be hotspots. These modelling features %
were introduced and discussed in %
\citet{richardson2010bayesian}, \citet{bottolo2011bayesian} and \citet{ruffieux2017efficient}, wherein the gain in statistical power %
over certain existing approaches %
was demonstrated.

This work focuses on realistic molecular quantitative trait locus settings, where a very large number of responses is analysed. It %
characterizes a parameter sensitivity issue, which was not highlighted in previous work and can be especially damaging %
for large response dimensions, and develops a robust solution based on %
a second-stage continuous shrinkage model that %
allows automatic discrimination of hotspots. %
The sensitivity concerns the specification of hyperparameters for a top-level variance parameter controlling hotspot propensity. %
Specifying variance components in hierarchical models is %
often difficult. %
\citet{gelman2006prior} 
discusses the relevance of %
several %
noninformative and weakly informative priors on random effect variances.  In \emph{large $n$} settings, the Bernstein--von Mises theorem suggests that the choice of prior may be unimportant in practice, but in high-dimensional settings, priors may have a strong impact on inferences. When the variance is close to zero, which is %
the case in sparse scenarios such as %
molecular QTL studies, Gelman cautions that badly chosen priors may severely distort posterior inferences. %
This observation is at the heart of %
work on scale-mixture priors such as the %
Strawderman--Berger prior \citep{strawderman1971proper, berger1980robust}, the Student\emph{-t} prior \citep{gelman2008weakly} or the horseshoe prior \citep{carvalho2010horseshoe}. These shrinkage priors differ in %
the modelling %
of 
the scale parameter, %
and all have substantial mass near zero in order to achieve %
good recovery of the overall sparsity pattern, while being sufficiently heavy-tailed to %
 capture strong signals. Fully noninformative priors (e.g., whereby the scale parameter would be assigned a Jeffreys prior) are ruled out, %
as they would fail to regularize. 
We lean on this body of work and overcome %
the sensitivity issue by introducing a fully Bayesian framework for hotspot detection based on  %
the horseshoe prior. Because it entails both global and local scale parameters, our proposal flexibly adapts to the pleiotropic level and the number of responses associated with each genetic variant, and robustly identifies large individual hotspot effects, whatever %
the overall sparsity level.

The detection of hotspots in molecular QTL studies would not be feasible without fast inference procedures, 
yet scalability should not be at the expense of accurate posterior exploration. This %
is particularly important in very high-dimensional settings, where posteriors %
are difficult to explore because they are highly multimodal. Building on previous work, we propose a computationally advantageous variational inference scheme for our global-local framework; the accuracy of such a scheme was validated and benchmarked against MCMC inference in \citet{ruffieux2017efficient}. Here we pay  particular attention to problems with strongly-correlated predictors, which further exacerbate multimodality. Such settings are %
typically encountered in genetics, as genetic variants exhibit local correlation structures along the genome. %
We %
augment the state-space of our algorithm with a simulated annealing procedure which allows it to escape more easily from local modes, and thus increases the chances of converging to the global mode \citep{rose1990deterministic, ueda1998deterministic}. %

The paper is organized as follows. Section \ref{sec_mot} presents the dataset used throughout the paper and provides a data-driven motivation for our work. Section \ref{sec_ps} states the problem in light of %
\citet{richardson2010bayesian}, \citet{bottolo2011bayesian} and \citet{ruffieux2017efficient} and formalizes its consequences for sensitivity and multiplicity control. Section \ref{sec_fram} presents our modelling framework and discusses its properties. Section \ref{sec_inference} describes our annealed variational inference procedure. 
Section~\ref{sec_simulations} assesses the performance of our approach in simulations, and Section~\ref{sec_application} applies it to real eQTL data.
Section~\ref{sec_conclusion} summarizes the results and gives some general discussion. Our approach is implemented in the publicly available R package \texttt{atlasqtl}.

\section{%
Data and motivating example} \label{sec_mot}

We introduce an eQTL study which serves both to demonstrate the need for tailored modelling of hotspots and to illustrate the merits %
of our proposal throughout the paper. This study differs from %
most molecular QTL analyses, as it involves expression from CD14$^+$ monocytes before and after immune stimulation,  %
performed by exposing the monocytes to the inflammation proxies interferon-$\gamma$ (IFN-$\gamma$) or differing durations of lipopolysaccharide (LPS 2h or LPS 24h). The genetic variants are single nucleotide polymorphisms (SNPs) determined using Illumina arrays and the samples were obtained from $432$ healthy European individuals.

Related work \citep{fairfax2014innate, kim2014characterizing, lee2014common} has suggested that gene stimulation may trigger substantial \emph{trans}-regulatory activity, %
creating favourable conditions for the manifestation of hotspot genetic variants. Indeed, while hotspots often exhibit associations with genes in their vicinity, they are evidenced by their capacity to influence (\emph{trans}-act on) many remote genes. %
In addition to monocyte expression, we consider B-cell expression data for the same samples, %
to contrast the hotspot activity for the two cell types. 

To recall the known drawbacks of the basic univariate screening approach when used for detecting of hotspots, 
we regressed each unstimulated monocyte level on each genetic variant from chromosome one. This led to the following observations (Table \ref{tb1} and Appendix~\ref{app_mot}): %
first, as expected, the estimated effect sizes of \emph{trans} associations uncovered at Benjamini--Hochberg false discovery rate of $20\%$ were substantially smaller than those of the \emph{cis} effects. Second, although this screening uncovered about $2.5$ times more \emph{cis} associations than \emph{trans} associations, about one-third of the former were essentially redundant: because of the local correlation structure on the genome (\emph{linkage disequilibrium}), a single transcript was often assessed as under control by several genetic variants at the same locus, yet these genetic variants are likely to be proxies for a single causal variant. %
Such scenarios were much less represented among the uncovered \emph{trans} associations, as they concerned only about $2\%$ of them. Hence the large number of false positive \emph{cis} associations reported by the marginal screening is likely to have hampered the detection of, weaker, \emph{trans} effects. %

{\centering
\begin{table}[t!]
\small
{\begin{center}
\begin{tabular}{lrrr}
  \hline
 & Number & Number after LD pruning  & Magnitude of estimated effects\\ %
  \hline
\emph{Cis} effects & 1,611 & 1,049 & 0.11 (0.10) \\ %
  \emph{Trans} effects & 655 & 641 & 0.04 (0.03)\\ %
   \hline
\end{tabular}
\end{center}}
\caption{\small Detection of \emph{cis} and \emph{trans} associations by univariate screening using a Benjamini--Hochberg false discovery rate threshold of $0.2$. Effects between a transcript and a SNP located less than $2$ megabases (Mb) to it were defined as in \emph{cis} effects; the remaining effects were defined as \emph{trans} effects. Left: number of detected pairwise associations. Middle: number of detected pairwise associations after grouping those between a given transcript and several SNPs in linkage disequilibrium (LD) using $r^2 > 0.5$ and window size $2$ Mb. Right: average magnitude of  regression estimates, and standard deviation in parentheses. 
}\label{tb1}
\end{table}}

There is a broad consensus about the generality of the above remarks when using marginal approaches \citep{gilad2008revealing, mackay2009genetics, nica2013expression}. 
It may be tempting to view them as consequences of the multiplicity burden entailed by molecular QTL problems; false discovery rate techniques with different corrections for \emph{cis} and \emph{trans} effects have indeed been proposed \citep{peterson2016treeqtl} and may alleviate the issue. Rather than pursue this approach, we anticipate and tackle the question upfront, at 
the modelling stage, %
by building a model for hotspots that can directly borrow information across genes.  %
Hierarchical regression models along this line exist, but none of them allow a fully Bayesian treatment of the hotspot propensities that is computationally feasible at the scale required by current eQTL studies.  %
We now show that adequate %
calibration of hotspot sizes is difficult and uncertain if not properly learnt from the data.

\section{Problem statement}\label{sec_ps}

We consider a series of hierarchically related regressions, with $q$ centered responses, $y = \lb y_1, \ldots, y_q\rb$, and $p$ centered candidate predictors, $X =\lb X_1, \ldots, X_p\rb$, for $n$ samples ($n \ll p$), 
\begin{eqnarray}\label{eq_model}
y_t &\mid& \beta_t, \tau_t \sim \mathcal{N}_n\lb X\beta_t, \tau_t^{-1} I_n\rb , \hspace{3.75cm}  t = 1, \ldots, q\,,\nonumber\\
\beta_{st} &\mid& \gamma_{st}, \tau_t, \sigma^2 \sim \gamma_{st}\,\mathcal{N}\lb 0, \sigma^2\,\tau_{t}^{-1}\rb  + (1-\gamma_{st})\,\delta_0 \,, \hspace{0.75cm} s = 1, \ldots, p\,,\\
\gamma_{st} &\mid& \omega_{st} \sim \mathrm{Bernoulli}\lb \omega_{st} \rb,\nonumber
\end{eqnarray}
where $\delta_0$ is the Dirac distribution, and where $\tau_t$ and $\sigma^{-2}$ are assigned Gamma priors. In the molecular QTL setting on which we focus, the predictors represent $p$ genetic variants, typically SNPs, and the responses are $q$ molecular expression levels, for $n$ individuals. The regression parameters, $\beta_{st}$, are specific to each pair of predictor $X_s$ and response $y_t$ and have spike-and-slab priors to induce sparsity \citep{mitchell1988bayesian, george2000variable}. 
Hence the binary latent variables $\gamma_{st}$ take value unity in case of association, and are zero otherwise. %
The global variance of effects, $\sigma^2$, allows information-sharing across responses associated with overlapping sets of predictors. This specification will be complemented with a second-level model on the probabilities of association $\omega_{st}$ in Section \ref{sec_fram}.

Model formulation $(\ref{eq_model})$ and variants thereof have been employed by authors such as \citet{jia2007mapping}, \citet{richardson2010bayesian}, \citet{bottolo2011bayesian} and \citet{ruffieux2017efficient}. Their proposals differ primarily in the prior specification for the probability of association parameter, $\omega_{st}$. \citet{richardson2010bayesian} and \citet{bottolo2011bayesian} decouple the predictor and response effects by setting $\omega_{st} = \omega_s \times \omega_t$, and place prior distributions on each of $\omega_s$ and $\omega_t$, whereas \citet{jia2007mapping} and \citet{ruffieux2017efficient} use the simpler formulation $\omega_{st} \equiv \omega_s$.
A suitable specification of the predictor-specific parameter $\omega_s$ is crucial, as $\omega_s$ controls the propensity of each predictor $X_s$ to be a hotspot, i.e., to be simultaneously associated with several responses. %
As we now explain, the discrimination of hotspots can be very sensitive to the choice of prior distribution for $\omega_{s}$ and this %
sensitivity becomes particularly %
severe in very large response settings, where the detection of hotspots is a key task. %

For the sake of discussion, we  
illustrate our point %
with the formulation of \citet{ruffieux2017efficient}, %
whereby
\begin{equation}\label{eq_beta}\omega_{st} \equiv \omega_s  \overset{\mathrm{iid}}{\sim} \mathrm{Beta}(a, b)\,, \qquad a, b > 0\,,\qquad\quad s = 1, \ldots, p\,;\end{equation}
similar considerations apply to the models of \citet{jia2007mapping}, \citet{bottolo2011bayesian}, and \citet{richardson2010bayesian}. %
We discuss the choice of the %
hyperparameters $a$ and $b$ through the prior expectation and variance for $\omega_s$. The expectation corresponds to the prior base rate of associated pairs, $\mo = \E(\omega_s) = \mathrm{pr}(\gamma_{st} = 1)$. Its value should be small to induce sparsity, typically $\mo \ll 1$ for $p, q \gg n$, and may be fixed using an estimate of the overall signal sparsity. In contrast, there is no 
prior state of knowledge about %
$\vo = \Var(\omega_s)$, and %
its choice turns out to impact the prior size of hotspots when $q$ is large. %
To formalize this, it is helpful to study prior odds ratios, as \citet{scott2010bayes} did when discussing expected model sizes in single-response sparse regression. 
For a given predictor $X_s$, write $\gamma_s(q_s)$ the $q$-variate indicator vector whose first $0 < q_s \leq q$ entries are unity and the  following $q-q_s$ are zero. The prior odds ratio
\begin{equation}\label{eq_por}
\mathrm{POR}(q_s - 1 : q_s)= \frac{\pr\left\{ \gamma_s(q_s-1)\right\}}{\pr\left\{ \gamma_s(q_s)\right\}} = \frac{b + q -q_s}{a + q_s -1}\,,
\end{equation}
quantifies the penalty induced by the prior when moving from $q_s-1$ to $q_s$ responses associated with  $X_s$. 
The penalty increases with the total number of responses in the model (for fixed $a$, $b$ and $q_s$), but it also decreases monotonically as $q_s$ increases, so that it is \emph{a priori} easier to add a response when $X_s$ is already associated with many responses. %
More insight into this phenomenon can be obtained by looking at the quantity%
\begin{equation}\label{eq_ratiopor} \frac{ \mathrm{POR}(0 : 1)}{ \mathrm{POR}(q_s - 1 : q_s)}\,, %
\end{equation}
which compares the cost of adding a further response association with $X_s$ when moving from the null model or from a model with $q_s-1$ associations already. 

In molecular QTL problems, $q_s$ is typically much smaller than $q$, as each SNP is believed to control just a few %
molecular entities. For $q_s \ll q$, $(\ref{eq_ratiopor})$ behaves roughly linearly in $q_s$ with slope $\approx a^{-1} = \vo \left\{\mo^{2} \,(1-\mo) - \mo \vo\right\}^{-1}$.  Hence, large $\vo$ favours large hotspots while small $\vo$ tends to give an association pattern that is more scattered across predictors.
In the latter case, strong shrinkage towards $\mo \ll 1$ may be induced and the resulting hotspot sizes may be underestimated, whereas, in the former case artifactual hotspots may appear when data are insufficiently informative %
to dominate the prior specification. %
Table $\ref{tab1}$ shows that the penalties $(\ref{eq_ratiopor})$ can differ drastically for different choices of $\vo$. %

\begin{table}[ht]
\centering
\small
\begin{tabular}{rrrrr}
  \hline
\hspace{0.8cm}$q_s$ & $5$ & $10$ & $50$ & $100$ \vspace{-0.2cm}\\
 $\vo$ \hspace{0.7cm} \\
  \hline
$10^{-4}$\hspace{0.6cm}\mbox{}& $1.0$ & $1.1$ & $1.5$ & $2.1$ \\ 
$10^{-3}$\hspace{0.6cm}\mbox{}& $1.4$ & $2.0$ & $6.5$ & $12.2$ \\ 
$10^{-2}$ \hspace{0.5cm}\mbox{}& $6.0$ & $12.3$ & $62.4$ & $125.4$ \\ 
   \hline
\end{tabular}
\caption{\small Ratios (\ref{eq_ratiopor}) for a grid of variances $\vo$ and numbers of associated responses $q_s$. The total number of responses is $q = 20,000$ and the base rate is $\mo = 0.1$. %
The penalty varies greatly depending on the chosen value for $\vo$ and increases roughly linearly with~$q_s$.} 
\label{tab1}
\end{table}

To evaluate the extent to which this could impact inference in flat likelihood scenarios, it is helpful to also study the case where $q_s$ is of order $q$, even though this is unlikely to be encountered in our applications. %
When $q_s \sim q$ (i.e., when $q_s/q$ tends to a strictly positive constant as $q \rightarrow \infty$), %
$(\ref{eq_ratiopor})$ is of order $O(q)$, so that, in weakly informative data settings, the sensitivity may lead to %
the manifestation of massive spurious hotspots associated with nearly all responses.
Such undesired ``pile-up'' effects %
highlight the need to adjust for the dimensionality of the response. %

\begin{figure}[t!]
\includegraphics[scale=0.65]{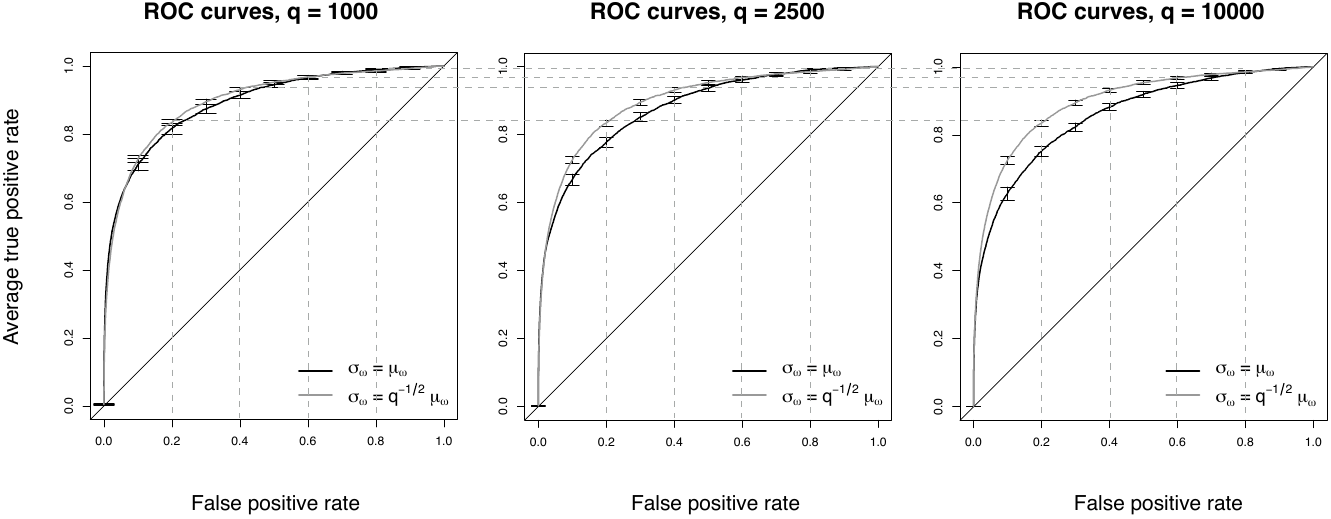}
\caption{\small Variable selection performance with and without multiplicity adjustment, measured by average receiver operating characteristic (ROC) curves with 95\% confidence intervals obtained from $100$ replicates. Three problems are simulated, with an increasing number of response variables, $q = 1,000$ (left), $q= 2,500$ (middle), $q=10,000$ (right), and $p = 100$ candidate predictors for $n = 100$ samples. The pattern of associations is the same for all three scenarios: $50$ responses are chosen randomly among the first $1,000$ responses to be associated with at least one of $10$ predictors; the rest of the responses are drawn from Gaussian noise. For a given response, the proportion of its variance explained by the predictors does not exceed $15 \%$. 
Two implementations of model (\ref{eq_model})--(\ref{eq_beta}) are compared: one uses a fixed choice of variance $\vo = \mo^2$ (black curves); its performance deteriorates as $q$ increases, from left to right. The other uses the proposed adjustment for the total number of responses $q$, i.e., %
$\vo = q^{-1}\mo^2$ (grey curves); its performance remains unchanged as $q$ increases (see grid). The base rate %
is fixed to the simulated proportion of associated predictors, i.e., $\mo = 0.1$.}\label{fig1} %
\end{figure}

The sensitivity of inferences to the hotspot propensity variance relates to the well-known issue of specifying prior distributions for variance components, as $\omega_s$ can be viewed as a random effect. %
While this sensitivity and its related response multiplicity burden are important problems that affect any hierarchically related regression model such as  %
$(\ref{eq_model})$, 
they have %
been neither formalized nor investigated in the literature. In fact, the number of responses presented in %
numerical experiments %
is usually rather small ($10$$-$$1,000$), mainly limited by the heavy computational load of MCMC sampling, %
so that this sensitivity issue typically goes unnoticed. %
Another aspect is that ``pile-up'' effects can be avoided by choosing a %
small hotspot propensity variance, %
at the risk of giving up substantial hotspot selection performance. %
The very sparse nature of molecular QTL analyses also rules out the use of simple empirical Bayes estimates, which typically collapse to the degenerate case $\hat{\sigma}^2_{\mathrm{\omega}} = 0$, see, e.g., \citet{scott2010bayes, van2018learning}. Thus, a %
tailored solution is~needed. %

Our proposal resolves the above issues, based on two considerations. %
First, we argue that ``pile-up'' effects can be prevented by suitably linking the hotspot propensity variance to the number of responses, in effect performing multiplicity adjustment. Indeed, %
choosing $\vo = O\left(q^{-1}\right)$, %
ratio $(\ref{eq_ratiopor})$ is $O(1)$ when $q_s \sim q$. For small values of $\mo$, typically chosen in sparse association problems, this adjustment %
amounts to enforcing small and similar numbers of response associations for all predictors, %
with the degree of shrinkage depending on the number of responses $q$ (in the limiting case $q \rightarrow \infty$, we obtain $\omega_s \equiv \mo$). Figure $\ref{fig1}$ illustrates the degradation of the variable selection performance in moderately informative problems with increasing $q$, and shows how the proposed penalty addresses the issue. %

Second, 
we embed and \emph{relax} this multiplicity adjustment in a fully Bayesian framework involving a second-stage model on the probability of association, $\omega_{st}$, and hence %
infer 
the hotspot propensity variances from the data in a fully automatic way, with %
 no ad-hoc choice or compromise %
 that would bias the hotspot sizes. %

  \section{Global-local modelling framework}\label{sec_fram}

  \subsection{Second-stage probit model on the probability of association} %
 
 As a first step in detailing %
  our proposal, we complement %
 $(\ref{eq_model})$ with a %
 hierarchical probit model on the probability of association, i.e., 
\begin{equation}\label{eq_ext}
\omega_{st} = \Phi(\theta_s + \zeta_t), %
\qquad %
\zeta_t \overset{\mathrm{iid}}{\sim}  \mathcal{N}(n_0, t_0^2), \qquad s = 1, \ldots, p, \; t = 1, \ldots q,
\end{equation}
where $\Phi(\cdot)$ is the standard normal cumulative distribution function, and where we assume for now that $\theta_s \overset{\mathrm{iid}}{\sim} \mathcal{N}(0, s_0^2)$. %
This second-stage model offers a flexible and interpretable representation of the association probability in multi-response settings: it involves a response-specific parameter, $\zeta_t$, which adapts to the sparsity pattern corresponding to each response, %
and a propensity parameter, $\theta_s$, which encodes predictor-specific modulations of the probability of association, as in %
\citet{richardson2010bayesian} and \citet{bottolo2011bayesian}. The hyperparameters $n_0$ and $t_0^2$ are set to match %
 a selected expectation and variance for the prior number of associated predictors per response %
 (see Appendix~\ref{app_hyper}). The variance parameter  $s_0^2$ %
 essentially plays the role of $\vo$, presented in Section \ref{sec_ps}, in influencing the prior odds ratios; in particular, an application of the delta method shows that if $s^2_0 \sim O\left(q^{-1}\right)$ as $q \rightarrow \infty$, then %
 $\Var\left\{\Phi\left(\theta_s\right)\right\} \sim O\left(q^{-1}\right)$. While no closed form can be obtained for prior odds ratios $(\ref{eq_por})$ based on model  $(\ref{eq_ext})$, numerical experiments suggest that %
 $(\ref{eq_ratiopor})$ indeed behaves independently of $q$ when $s_0^2 = q^{-1}$, for $q_s \approx q$ large. Formulation $(\ref{eq_ext})$ sets the stage for introducing our new multiplicity-adjusted hotspot model, which combines the benefits of both global and local control and adaptation.

\subsection{Horseshoe prior on hotspot propensities}\label{sec_hs}

Our proposed specification for the hotspot propensity adds %
flexibility in modelling the scale of $\theta_s$ in (\ref{eq_ext}) by letting %
\begin{equation}\label{eq_hs}
\theta_s \mid \lambda_s, \sigma_0 \sim \mathcal{N}\left(0, \sigma_0^2 \lambda_s^2 \right)\,, \qquad  \lambda_s \overset{\mathrm{iid}}{\sim} \mathrm{C}^+(0, 1)\,, \qquad s = 1, \ldots, p\,,
\end{equation} 
where $\mathrm{C}^+(\cdot, \cdot)$ is a half-Cauchy distribution. 
This corresponds to placing a horseshoe prior \citep{carvalho2010horseshoe} on the hotspot propensities, $\theta_s \mid \sigma_0 \overset{\mathrm{iid}}{\sim} \mathrm{HS}(0, \sigma_0)$. 
The global scale $\sigma_0$ adapts to the overall sparsity pattern, while the Cauchy tails of the predictor-specific scale parameters $\lambda_s$ flexibly capture the hotspot effects. %

The horseshoe prior is a popular example of absolutely continuous shrinkage priors, with newly established theoretical guarantees, %
such as near-minimaxity in estimation \citep{van2017adaptive}. It also belongs to the class of global-local shrinkage priors that have an infinite spike at the origin and regularly-varying tails \citep{polson2010shrink, bhadra2016default}.

\subsection{Multiplicity-adjusted shrinkage profile}\label{sec_multadj}
While the local scale parameters $\lambda_s$ are essential to suitably %
detect the few large signals, the choice of the global scale $\sigma_0$ is no less important, as $\sigma_0$ controls the ability of the model to discriminate signal from noise. \citet{piironen2017sparsity} propose to choose $\sigma_0$ based on specific sparsity assumptions; we extend their considerations to our multi-response setting and further highlight how the dimension of the response needs to be accounted for in order 
 to recover the beneficial shrinkage properties conferred by the horseshoe prior when used in the classical normal means model.
For a given predictor $X_s$, we reparametrize the probit link formulation, 
$$\gamma_{st}\mid \theta_s,  \zeta_t \sim \mathrm{Bernoulli}\left\{\Phi(\theta_s + \zeta_t)\right\}\,, \qquad t = 1, \ldots, q\,,$$
by introducing a $q$-variate auxiliary variable $z_{s}=(z_{s1}, \ldots, z_{sq})$, as 
\begin{equation}\label{eq_reparam_z}\gamma_{st} = \one\{z_{st} > 0\}\,,\qquad z_{st} \mid \theta_s, \zeta_t \sim \mathcal{N}\left(\theta_{s} + \zeta_t, 1\right)\,, \qquad t = 1, \ldots, q\,.\end{equation}
 In this second-stage probit model, %
$z_{st}$ can be understood as data, %
and $\theta$ as a sparse parameter. 
Given the hyperparameters $n_0$ and $t_0^2$ for $\zeta_t$,  we have
$$z_{st} \mid \theta_s \sim \mathcal{N}(n_0 + \theta_s, 1 + t_0^2),$$ so that
$$\E\left(\theta_s \mid z_s, \sigma_0, \lambda_s\right) = (1-\kappa_s)  \frac{1}{q} \sum_{t = 1}^q (z_{st} - n_0) + \kappa_s \times 0 = (1-\kappa_s)\bar{z}'_{s},$$
where $\bar{z}'_{s} =  \bar{z}_s - n_0$ and
$$\kappa_s = \frac{1}{1 + \alpha(\sigma_0) \lambda_s^2}$$
is the \emph{shrinkage factor} for hotspot propensities, with $\alpha(\sigma_0) = q (1 + t_0^2)^{-1} \sigma_0^2$  (Lemma~\ref{sm_lemma_mean} %
of Appendix~\ref{app_mult}). %

In the horseshoe prior literature, with half-Cauchy priors on the local scales %
as well as unit global scale and error variance, this factor has a $\mathrm{Beta}\left(1/2, 1/2\right)$ prior whose shape resembles a horseshoe, hence the name. As this prior density is unbounded at $0$ and $1$, one expects \emph{a priori}, either large effects, with $\kappa_s$ close to zero, or no effects, with $\kappa_s$ close to one. In our case, it can be shown that
$$p(\kappa_s \mid \sigma_0) =\pi^{-1} \alpha(\sigma_0)^{1/2}\, \kappa_s^{-1/2} (1 - \kappa_s)^{-1/2} \ls 1 + \kappa_s \lc \alpha(\sigma_0) - 1\rc\rs^{-1}, \qquad 0 < \kappa_s < 1\,,$$
using $\lambda_s \overset{\mathrm{iid}}{\sim} \mathrm{C}^+(0, 1)$; this prior density reduces to $\mathrm{Beta}\left(1/2, 1/2\right)$ 
when $\alpha(\sigma_0) = 1$, that is, when $\sigma_0^2 \approx q^{-1}$,  as $t_0^2 \ll 1$ under sparse assumptions (Lemma~\ref{sm_lemma_prior_kappa} %
and Figure~\ref{sm_fig_shr} %
of Appendix~\ref{app_mult}). %
This formulation therefore enjoys the shrinkage properties of the horseshoe prior. Critically, using a default choice of $\sigma_0^2 = O(1)$ as $q \rightarrow \infty$ would yield $\E\left(\kappa_s \mid \sigma_0\right) \approx 0$ for $q$ large, so that on average, $\theta_s$ would be %
 unregularized given $z_s$. %
These two choices can be read in light of the discussion in Section \ref{sec_ps}: the latter mirrors the absence of any correction for the dimensionality of the response, possibly creating spurious ``pile-up'' effects, whereas the former satisfies the multiplicity adjustment condition with the proposed scaling factor $q^{-1}$ for $\sigma^2_0$. %

\begin{figure}[t!]
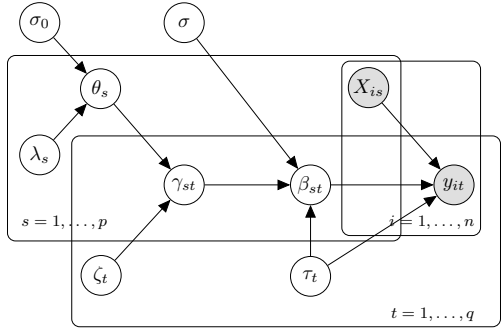

\centering
 \scalebox{0.75}{\tikz{ %
     \node[latent] (lambda) {$\lambda_{s}$} ; %
    \node[latent, right= of lambda, yshift = 1.17cm, xshift = -0.65cm] (theta) {$\theta_{s}$} ; %
         \node[latent, left=of theta, yshift = 1.17cm, xshift = 0.65cm] (sigma0) {$\sigma_{0}$} ; %
    \node[latent, right=of theta, yshift = -1.7cm, xshift = -0.25cm] (gamma) {$\gamma_{st}$} ; %
    \node[latent, right=of gamma, xshift = 0.5cm] (beta) {$\beta_{st}$} ; %
        \node[latent, left=of gamma, yshift=-1.6cm, xshift = 0.26cm] (zeta) {$\zeta_{t}$} ;
     \node[obs, right=of beta, xshift = 0.8cm] (y) {$y_{it}$} ; %
         \node[obs, above=of y, xshift = -1.5cm] (x) {$X_{is}$} ; %
                               \node[latent, below=of beta, yshift =0.1cm] (tau) {$\tau_t$} ; %
                                   \node[latent, above=of gamma, yshift=1.15cm] (sig2inv) {$\sigma$} ; %
    \plate[inner sep=0.1cm, xshift = 0cm] {plate1} {(x) (y)} {$i = 1, \ldots, n$}; %
    \plate[inner sep=0.2cm] {plate2} {(lambda) (theta) (gamma) (beta) (x)} {$s = 1, \ldots, p$\hspace{4.95cm}\mbox{}}; %
    \plate[inner sep=0.3cm, yshift =0.2cm, xshift = 0.15cm] {plate3} {(zeta) (gamma) (tau) (y)} {$t = 1, \ldots, q$}; %
        \edge {lambda} {theta} ; %
            \edge {sigma0} {theta} ; %
    \edge {theta} {gamma} ; %
       \edge {zeta} {gamma} ; %
    \edge {gamma} {beta} ; %
        \edge {sig2inv} {beta} ; %
    \edge {tau} {beta} ; %
    \edge {tau} {y} ; %
     \edge {beta} {y} ; %
     \edge {x} {y} ; %
  }}
  \caption{\small Graphical representation of model (\ref{eq_all}). The shaded nodes are observed, the others are inferred; $\beta_{st}$ is the regression coefficient for association between predictor $X_s$ (SNP) and response $y_t$ (expression level), and $\gamma_{st}$ is the latent binary indicator for the presence or absence of this effect. The probability of association is decoupled into response-specific, $\zeta_t$, and predictor-specific, $\theta_s$, contributions. The latter entails the global-local second-stage model for hotspots.}\label{fig_gm}
\end{figure}
Fixing $\sigma_0^2 = q^{-1}$ would stop the global scale from adapting to the degree of signal sparsity. %
We instead place  a hyperprior on $\sigma_0$ which embeds the penalty. Following \citet{carvalho2010horseshoe}, we choose a half-Cauchy prior,
\begin{equation}\label{eq_global}
\sigma_0 \sim  \mathrm{C}^+(0, q^{-1/2})\,.
\end{equation}
An equivalent parametrization of $(\ref{eq_hs})$ and $(\ref{eq_global})$ is
\begin{equation}\label{eq_rev2}\theta_s \mid \lambda_s, \sigma_0 \sim \mathcal{N}\left(0, q^{-1}\lambda_s^2 \sigma_0^2 \right), \quad \lambda_s \overset{\mathrm{iid}}{\sim} \mathrm{C}^+(0, 1)\,, \quad \sigma_0 \sim \mathrm{C}^+(0, 1)\,,\end{equation}
from which one clearly sees how the multiplicity factor rescales the hotspot propensity variance. For clarity, we gather the complete specification of our global-local hierarchical model; it combines (\ref{eq_model}) and the decomposition of the probability parameter (\ref{eq_ext}) with (\ref{eq_hs}) and (\ref{eq_global}):
\begin{eqnarray}\label{eq_all}
y_t &\mid& \beta_t, \tau_t \sim \mathcal{N}_n\lb X\beta_t, \tau_t^{-1} I_n\rb , \hspace{4.5cm}  t = 1, \ldots, q\,,\nonumber\\
\beta_{st} &\mid& \gamma_{st}, \tau_t, \sigma \sim \gamma_{st}\,\mathcal{N}\lb 0, \sigma^2\,\tau_{t}^{-1}\rb  + (1-\gamma_{st})\,\delta_0 \,, \hspace{1.67cm} s = 1, \ldots, p\,,\\
\gamma_{st} &\mid&\theta_s, \zeta_t \sim \mathrm{Bernoulli}\lc \Phi(\theta_s + \zeta_t) \rc,\qquad  \zeta_t \overset{\mathrm{iid}}{\sim}  \mathcal{N}(n_0, t_0^2),\nonumber\\
\theta_s &\mid& \lambda_s, \sigma_0 \sim \mathcal{N}\left(0, \lambda_s^2 \sigma_0^2 \right), \qquad \lambda_s \overset{\mathrm{iid}}{\sim} \mathrm{C}^+(0, 1)\,, \qquad \sigma_0 \sim  \mathrm{C}^+(0, q^{-1/2})\,,\nonumber
\end{eqnarray}
with Gamma prior distributions for $\tau_t$ and $\sigma^{-2}$; a graphical representation is provided in Figure~\ref{fig_gm}.

\section{Annealed variational inference}\label{sec_inference} %

Joint inference on molecular QTL models is particularly difficult, a serious complication being the high dimensionality of the predictor and the response spaces. %
In our proposal, %
as well as in those of \citet{jia2007mapping}, \citet{richardson2010bayesian}, \citet{bottolo2011bayesian} and \citet{ruffieux2017efficient}, the %
binary latent matrix $\Gamma = \{\gamma_{st}\}$ creates a discrete search space of dimension $2^{p \times q}$, $p, q \gg n$, and the quality of inferences hinges on successful exploration of this space. %
Several sampling schemes have been proposed for spike-and-slab models. Most of them involve drawing each latent component from its marginal posterior distribution, and therefore require costly evaluations of marginal likelihoods at each iteration. %
Mixing problems also arise, %
mainly caused by the difficulty that the sampler has in jumping between the states defined by the spike and the slab components. The resulting sample autocorrelations are high, so %
many iterations are usually needed to collect enough independent samples. 

There are two paths towards scaling up Bayesian inference. The first is to design more efficient Markov Chain Monte Carlo algorithms. While there is a growing literature on the \emph{large $n$} case, with proposals involving partition-based parallelization \citep[e.g.,][]{wang2013parallelizing} or data subsampling \citep[e.g.,][]{bardenet2014towards}, %
research is limited for the high-dimensional case, apart from work on approximating transition kernels \citep{o2004bayesian, bhattacharya2010nonparametric, guhaniyogi2018bayesian}, %
so %
effectively scaling MCMC methods for dimensions such as those involved in molecular QTL analyses is still largely %
out of reach. %
The second path therefore investigates deterministic alternatives to sampling-based approaches. These include expectation-maximization algorithms, expectation-propagation inference and variational inference. Effectively implemented, %
these approaches require only reasonable computing resources. %
A legitimate %
concern, however, is whether fast deterministic inference can be sufficiently accurate for variable selection in genome-wide association studies. In the case of variational inference, \citet{carbonetto2012scalable} and \citet{ruffieux2017efficient} provide positive evidence for accurate posterior exploration, %
notably through extensive comparisons with MCMC inference. %
Here, we build on this previous work and develop an efficient variational inference scheme for our global-local modelling framework. We further improve posterior exploration by coupling our algorithm with a simulated annealing procedure \citep{rose1990deterministic, ueda1998deterministic}. %

Let $v$ be the parameter vector of interest. Variational posterior approximations are obtained by considering a tractable analytical approximation, $q(v)$, to the true posterior distribution, $p(v \mid y)$. The \emph{mean-field} approximation \citep{opper2001advanced} assumes that $q(v)$ factorizes over some partition of $v$, $\{v_j\}_{j=1,\ldots, J}$, i.e.,
\begin{equation}\label{eq_mean_field} q(v) = \prod_{j=1}^J q(v_j)\,, \end{equation}
with no assumption on the functional forms of the $q(v_j)$. One then performs inference by maximizing the following lower bound on the marginal log-likelihood, 
\begin{equation}\label{eq_lb}\mathcal{L}(q) = \int q(v) \log\left\{\frac{p(y, v)}{q(v)}\right\}\mathrm{d}v\,,\end{equation}
 which is a tractable alternative to minimizing the Kullback--Leibler divergence \citep[see, e.g.,][]{blei2017variational},
 \begin{equation}\label{eq_kl}\mathrm{KL}\left(q \middle\| p\right) = - \int q(v) \log\left\{ \frac{p(v\mid y)}{q(v)}\right\} \mathrm{d}v\,. \end{equation}
Quantity (\ref{eq_lb}) is often called ELBO, for \emph{evidence lower bound}, in the machine learning literature.

With each $v_j$ modelled as independent \emph{a posteriori} of the other parameters given the observations and the hyperparameters, mean-field variational inferences (\ref{eq_mean_field}) trade off posterior dependence assumptions and computational complexity. %
For our %
model, independence assumptions between $\beta_{st}$ and $\gamma_{st}$ %
would be particularly problematic: they would make $q(\beta_t)$ a unimodal representation of the marginal distribution $p(\beta_t \mid y)$, and thus a poor proxy for the highly multimodal posterior distribution implied by the spike-and-slab prior on $\beta_t$. Considering model reparametrization (\ref{eq_reparam_z}), %
we instead employ a structured factorization, whereby we model $\beta_{st}, \gamma_{st}$ and $z_{st}$ jointly, i.e., for each fixed $t \in \{1, \ldots, q\}$, we seek a variational distribution of the form
\begin{equation}\label{eq_struct}\prod_{s=1}^p q(\beta_{st}, \gamma_{st}, z_{st}).\end{equation}
This structured factorization %
induces point mass mixture factors and hence retains the multimodal behaviour of the spike-and-slab distribution. 
It is also a faithful representation of the true posterior distribution when predictors are only weakly dependent, 
since the latter factorizes %
as (\ref{eq_struct}) when using an orthogonal design matrix, as pointed out by \citet{carbonetto2012scalable}. %
 
Other fast deterministic inference procedures based on expectation-maximization algorithms have been proposed for spike-and-slab models \citep{rovckova2014emvs}. While these are limited to producing point estimates, variational procedures infer full approximating distributions, and thus also estimate posterior parameter uncertainty. Variational inference is frequently decried %
 as underestimating posterior variances, however, as a result of both the mean-field independence assumptions \citep{wainwright2008graphical} 
 and the optimization of a \emph{reverse} %
Kullback--Leibler divergence; $p(\cdot)$ and $q(\cdot)$ are swapped in $(\ref{eq_kl})$ compared to the standard Kullback--Leibler divergence. As the variational objective function (\ref{eq_lb}) is not concave, underestimated variances may also affect the ability of the algorithm to retain relevant variables. Indeed, in highly multimodal settings such as those induced by strongly correlated predictors, %
the approximation tends to concentrate mass on a local mode corresponding to a single configuration of variables in groups of correlated predictors.  Figure \ref{fig_ill_ann} considers a problem with highly correlated predictors, where the vanilla variational algorithm completely misses one active SNP and instead picks one of its correlated neighbours with high confidence. %
Figure \ref{fig_ill_ann}  also suggests that averaging posterior probabilities across multiple runs %
with different starting values can produce ``diluted'' posterior summaries %
which better reflect the uncertainty of SNP selection in regions of high linkage disequilibrium (i.e., with strong local correlation). But although averaging mitigates the problem, %
it results in increased computational costs, becoming quite substantial for typical molecular QTL problems. 

Simulated annealing directly %
targets the improved exploration of multimodal posteriors. It introduces a so-called \emph{temperature} parameter which indexes a series of \emph{heated} distributions and controls the degree of separation of their modes. The procedure starts with large temperatures that flatten the distribution of interest, thereby sweeping most of its local modes away and facilitating the search for the global optimum. Temperatures are then progressively decreased until the \emph{cold} distribution, corresponding to the original multimodal distribution, is reached.

 Optimization via simulated annealing was first described in \citet{metropolis1953equation} and \citet{kirkpatrick1983optimization} for Metropolis algorithms and %
was then adapted for expectation-maximization by %
\citet{ueda1998deterministic} and for variational inference by \citet{katahira2008deterministic}. %
Variational inference %
lends itself to %
simulated annealing principles. Indeed, the objective function (\ref{eq_lb}) can be rewritten as the sum of two terms: the expected value of the log joint distribution, which encourages the approximation to put mass on configurations of the variables that best explain the data, and an entropy term, which prefers the approximation to be more dispersed. Annealing inflates the entropy by multiplying it by the temperature parameter, 
$$\mathcal{L}_T(q) = \int q(v) \log p(y, v) \mathrm{d}v -   T \int q(v) \log q(v) \mathrm{d}v\,, \qquad T \geq 1;$$
it penalizes the first term (when $T>1$) and gradually relaxes this penalty until the original variational algorithm is obtained (when $T = 1$). 

There is no consensus on the type of temperature schedule to use.  We follow \citet{kirkpatrick1983optimization} in their choice of geometric schedule and use the specific implementation of \citet{gramacy2010importance}, %
$$T_j = (1 + \Delta)^{j-1}, \qquad \Delta = T_J^{1/(J-1)} -1, \qquad j = J, \ldots, 1,$$ 
where $T_J$ is the hottest temperature.  %
Our experiments suggest that initial temperatures between $2$ and $5$, and grids of $10$ to $100$ temperatures, depending on the computational resources at hand, are sufficient for good exploration. A final purpose of Figure \ref{fig_ill_ann} is to illustrate the benefits of annealed variational inference over classical variational inference: %
while selection based on the former may suffer from poorly chosen starting values, selection based on the latter consistently identifies the relevant SNPs across all $1,000$ restarts. %

\begin{figure}[t!]
\centering
\includegraphics[scale=0.46]{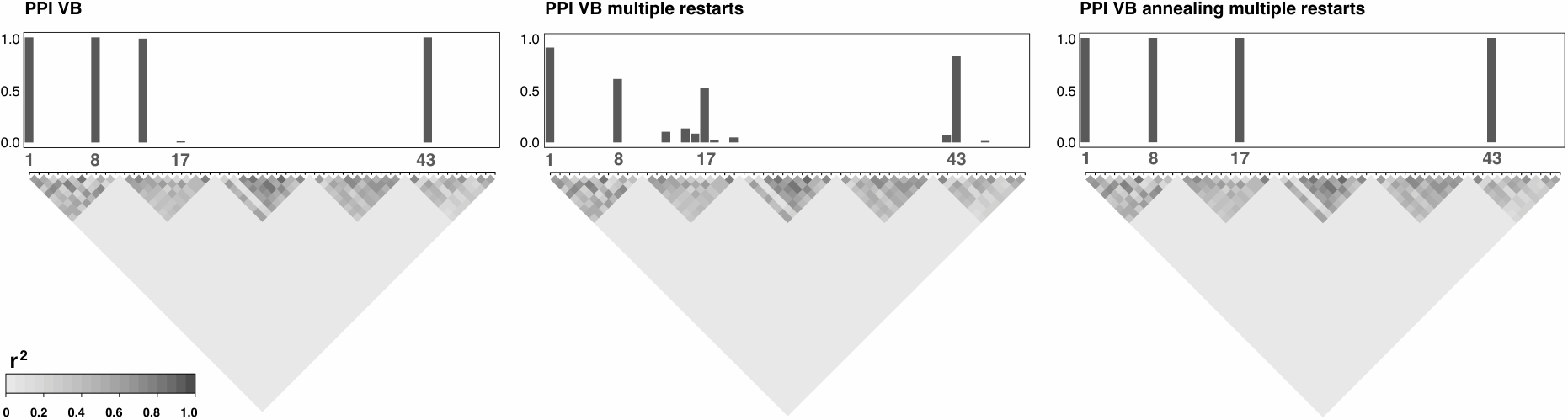}
 \caption{\small Variable selection under high multicollinearity. Problem with a single response and $1,000$ SNPs autocorrelated by blocks as candidate predictors (first $50$ shown). The SNPs simulated as associated with the response explain $30 \%$ of its variance; their positions are marked by the numeric labels.  The bars show the posterior probabilities produced by the variational algorithm, one run (left) and average of $1,000$ runs with different starting values (middle), and by the annealed variational algorithm with initial temperature $T = 5$ and grid of $100$ temperatures, average of $1,000$ runs with different starting values (right).}\label{fig_ill_ann}
\end{figure}

In \emph{large $n$} regimes, the scalability of variational algorithms can often greatly benefit from data-subsampling, which may be implemented  generically in stochastic gradient ascent schemes; this is not the case in high dimensions. %
In this latter regime, we believe that tailored, model-specific, derivations aiming for closed-form updates are important. %
Taking advantage of the conditional conjugacy properties of our model and resorting to suitable reparametrizations, we obtained all the variational updates analytically, albeit using special functions, such as the incomplete gamma and exponential integral functions. In particular, obtaining closed-form updates for the horseshoe's half-Cauchy scale parameters hinged on introducing auxiliary variables \citep[see, e.g.,][]{neville2014mean} to arrive at variational distributions in the Gamma family or involving cheap-to-compute special functions, and this was somewhat complicated by the annealing. The full derivation of the annealed variational updates is in Appendix~\ref{app_vb}. %

We then implemented a block coordinate ascent optimization procedure, where the variational parameters are updated in turn and by blocks for all the responses, exploiting the concavity of $\mathcal{L}(q)$ in each of these blocks. This scheme combined with the rapidly computable %
updates produces a highly effective algorithm. The algorithm returns the variational parameters after convergence, some of which can be directly employed to perform variable selection, e.g., the variational posterior probability of association for each pair, $\mathrm{E}_q\left(\gamma_{st}\right)$, the variational posterior mean of the corresponding regression coefficient, $\mathrm{E}_q\left(\beta_{st}\right)$, and the variational posterior means of the hotspot propensities, $\mathrm{E}_q\left(\theta_s\right)$, where $\mathrm{E}_q\left( \cdot\right)$ is the expectation with respect to the variational posterior distribution $q(\cdot)$ (see Appendix~\ref{app_vb}). %

\section{Simulations}\label{sec_simulations} 

\subsection{Data generation %
for pleiotropic QTL problems}\label{sec_simset}

The numerical experiments presented below are meant to closely reproduce real genetic data scenarios demonstrating pleiotropy, i.e., the control of several outcomes by a single SNP, for which our method is primarily designed. They also broadly illustrate the characteristic features of the method when applied to association studies with a large number of correlated responses.  Simulated data %
mimic molecular QTL data based on general principles of statistical genetics. %
We either extract SNPs from real datasets (for Sections \ref{sec_s2} and \ref{sec_other_appr}) or simulate them as in Hardy--Weinberg equilibrium and autocorrelated by blocks. %
In the latter case, we form the blocks using realisations from multivariate Gaussian latent variables of dimension $50$ and with autocorrelation coefficient drawn uniformly at random in a preselected interval; for simulations of Sections \ref{sec_s1} and \ref{sec_s1_n}, we use $(0.75, 0.95)$. We then use a quantile thresholding rule to code the number of minor alleles as $0$, $1$ or $2$ according to a SNP-specific minor allele frequency drawn from a uniform distribution, $\mathrm{Unif}(0.05, 0.5)$. We also generate block-dependent responses using multivariate normal variables; the blocks consist of $10$ equicorrelated responses, with residual correlation drawn from the interval $(0, 0.25)$ for simulations of Sections \ref{sec_s1} and \ref{sec_s1_n}, and from the interval $(0, 0.5)$ for simulations of Sections \ref{sec_s2} and \ref{sec_other_appr}.

The pleiotropic association pattern is constructed as follows. To model large and functionally inert genomic regions, we partition the SNPs into $N$ chunks of size $200$ and leave $\lfloor N/2 \rfloor$ chunks with no associations. From the remaining chunks, we randomly select labels for the ``active'' predictors, namely, SNPs associated with at least one response. Similarly, we select labels for the ``active'' responses, namely, responses associated with at least one SNP. We then randomly associate each active predictor with one active response, and each active response with one active predictor. For each active SNP $s$, we draw a ``propensity'' parameter $\omega_s$ from a $\mathrm{Beta}(1, 5)$ distribution, and further associate the SNP with other active responses whose labels are sampled with probability $\omega_s$; these SNP-specific propensities $\{\omega_s\}$ therefore create hotspots of different ``sizes''. We effectively generate the associations under an additive dose-effect scheme, whereby each copy of the minor allele results in a uniform and linear increase in risk, and we draw the proportion of a response's variance explained by individual SNPs from a left-skewed Beta distribution to favour the generation of smaller effects. We then rescale these proportions so that the proportion of response variance attributable to genetic variants does not exceed a certain value; the magnitude of SNP effects derives from this value, and the sign of the effects is altered with probability $0.5$. These choices imply an inverse relationship between minor allele frequencies and effect sizes, as expected under natural selection \citep[selection against SNPs with large penetrance is stronger, see, e.g.,][]{park2011distribution}. For a given experiment, we keep the same  association pattern across all replicates, but we regenerate the SNPs (if not real),  effect sizes and responses for each replicate. The remaining settings (e.g., numbers of variables and of samples, proportion of response variance explained by the SNPs) vary, so will be detailed in the text corresponding to each experiment. Data-generation functions are implemented in the R package \texttt{echoseq} available at  \url{https://github.com/hruffieux/echoseq}.

\subsection{Variable selection performance with global-local modelling}\label{sec_s1} %

In this section we evaluate the performance of our proposal for discriminating hotspots and selecting pairs of associated predictor and response variables. 
We simulated a ``reference'' data scenario with hotspots associated with approximately $35$ responses on average and whose cumulated effect sizes are responsible for at most $25\%$ of the variability of each response. We also generated four variants of this scenario: with smaller or larger hotspots (average sizes $\approx 17$ and $85$, respectively), and with weaker or stronger effects (response variance explained by hotspots below $20\%$ and $30\%$, respectively). 
 Each problem involves $p = 1,000$ SNPs and $q = 20,000$ responses (which corresponds the estimated number of protein-coding genes in humans), for $n = 300$ samples. We simulated $20$ hotspots, and, depending on the hotspot size scenario, $100$, $200$ or $500$ responses had at least one association. %

 We benchmarked our global-local model (\ref{eq_all}) against four alternatives. The first three are based on the proposal (\ref{eq_model})--(\ref{eq_beta}) of %
\citet{ruffieux2017efficient}, with three choices of hotspot propensity variance, $\vo$. These choices were made without assuming any prior state of knowledge, as would be faced %
in real data situations: we set the base rate of associated pairs to $\mo = 0.002$, so that two predictors are \emph{a priori} associated with each response, on average. Then, for each model, we set the hotspot propensity scale to a different fraction of this base rate.  The fourth model places a \emph{global} Gamma prior on the hotspot propensity precision and embeds the multiplicity penalty used in our proposal, i.e.,
$$\theta_s \mid \sigma_0 \sim \mathcal{N}\left(0, q^{-1}\sigma_0^2\right), \quad \sigma_0^{-2} \sim \mathrm{Gamma}\left(\frac{1}{2}, \frac{1}{2}\right),$$
which can be reparametrized as %
$$\theta_s \mid \sigma_0 \sim \mathcal{N}\left(0, \sigma_0^2\right), \quad \sigma_0^{-2} \sim \mathrm{Gamma}\left(\frac{1}{2}, \frac{1}{2q}\right).$$
With this choice, the propensity parameter has a Cauchy %
marginal prior distribution, \linebreak $\theta_s~\sim~\mathrm{C}(0, q^{-1/2})$. %
Both the Cauchy and the horseshoe models rely on the base rate level used for the three fixed-variance models to define the prior expected number of predictors associated with each response as $\mathrm{E}_p = \mo \times p = 2$; %
the prior variance for this number is set to $\mathrm{V}_p = 100$, which is large enough to cover a wide range of configurations. 
We use annealed variational inference on all five models; the geometric schedule consists of a grid of $100$ temperatures, with initial temperature $T = 5$.

\begin{figure}[t!]
\begin{center}
\includegraphics[scale=0.5]{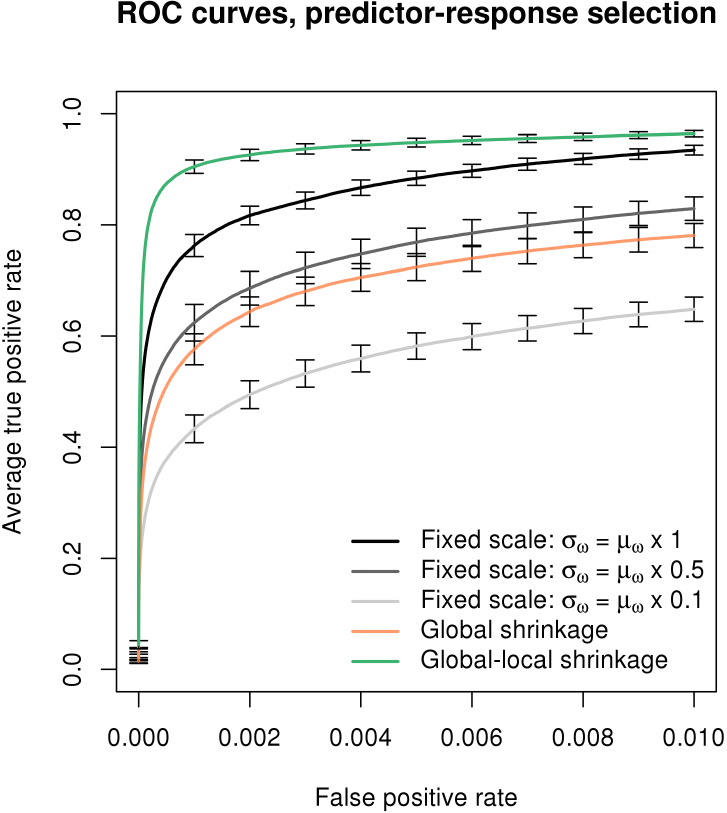}\qquad\qquad
\includegraphics[scale=0.5]{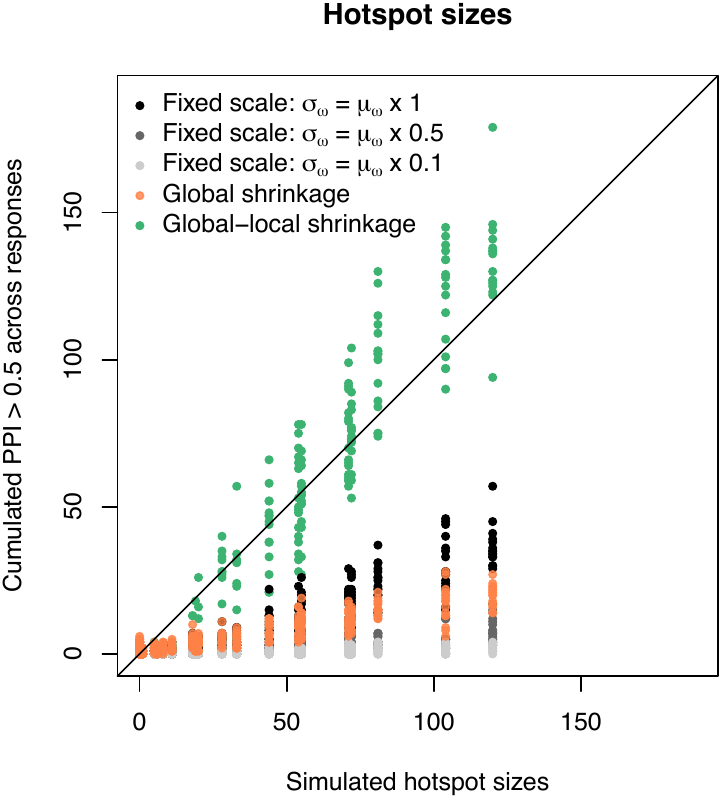}%
\end{center}
\caption{\small Performance of five hotspot modelling approaches, for the ``reference'' data generation case. Left: truncated average ROC curves for predictor-response selection with 95\% confidence intervals obtained from $64$ replicates. %
Right: sizes of recovered hotspots based on the median probability model rule \citep{barbieri2004optimal} applied to the variational posterior probabilities of association, $\mathrm{E}_q\left(\gamma_{st}\right)$; $16$ replicates are superimposed.
The data comprise $p = 1,000$ simulated SNPs with $20$ hotspots, $q = 20,000$ responses, of which $200$ are associated with at least one hotspot, leaving the rest of the responses unassociated. 
The block-autocorrelation coefficients for SNPs were drawn from the interval $(0.75, 0.95)$, and the residual block-equicorrelation coefficients for responses were drawn  from the interval $(0, 0.25)$. At most $25\%$ of each response variance is explained by the hotspots. For the fixed-variance models, we used a base rate $\mo = 0.002$, and scales $\sigma_\omega = \mo \times \{1, 0.5, 0.1\}$, as explained in the text.}\label{fig_s1_ann}
\end{figure}
\begin{table}[ht]
\footnotesize
\centering
\begin{tabular}{rccccc}
  \hline
 & $\sigma_\omega = \mu_\omega \times 0.1$ & $\sigma_\omega = \mu_\omega \times 0.5$ & $\sigma_\omega = \mu_\omega \times 1$ & global & global-local\\ 
  \hline  
  {\bf Pairwise selection\mbox{}} &&&&&\\
Reference %
& 55.5 (1.2) & 74.1 (1.3) & \bf  85.9 (0.7) & 69.5 (1.2) & \bf  93.6 (0.4) \\ 
Smaller hotspots %
 & 56.7 (1.0)& 67.7 (1.3)& \bf  82.6 (0.8) & 74.4 (0.9) & \bf  90.4 (0.6)\\ 
Larger hotspots %
 & 57.7 (1.0) & 84.4 (0.5) &\bf   89.6 (0.3) & 65.3 (0.8) & \bf 96.7 (0.1) \\ 
Weaker hotspots %
& 44.7 (1.0) & 57.5 (1.4) &\bf  77.3 (0.8) & 53.0 (1.2) & \bf 81.5 (1.0) \\
Stronger hotspots %
& 64.0 (1.1) & 82.6 (0.7) & \bf  90.2 (0.4) & 78.6 (0.7) & \bf  96.2 (0.1) \\ 
  {\bf Predictor selection\mbox{}}  &&&&&\\
Reference & 65.8 (2.7)& 68.2 (2.6) & 68.2 (2.2)& \bf 71.7 (2.1)& \bf 74.6 (1.6)\\ 
Smaller hotspots &  53.1 (3.4)& 54.7 (3.3)& 54.9 (3.2) & \bf 61.0 (3.1) & \bf 64.0 (3.0) \\ 
Larger hotspots  & 80.2 (2.9) & \bf  84.3 (2.6) & 84.0 (2.7) & 83.7 (2.4) & \bf 87.1 (2.1) \\
Weaker hotspots  &  52.9 (3.1) & 54.6 (3.2) & 56.0 (2.8) & \bf 59.2 (2.9) & \bf 63.2 (2.5)\\ 
Stronger hotspots  & 75.6 (3.1) & 78.3 (2.8) & 77.4 (2.7) & \bf 81.3 (2.4) & \bf 84.7 (2.1)\\ 
   \hline
\end{tabular}
\caption{\small  Average standardized partial areas under the curve $\times 100$ with false positive threshold $0.01$ for predictor-response selection performance and predictor (hotspot) selection performance. Different hotspot size and effect size scenarios are reported, each based on $64$ replicates; the ``reference'' case is displayed in Figure \ref{fig_s1_ann}. Standard errors are in parentheses and, for each scenario, the best two performances %
are in bold.}\label{tb_s1_range}
\end{table}

Figure \ref{fig_s1_ann} and Table \ref{tb_s1_range}
compare the five models in terms of selection of associated pairs of predictors and responses, selection of predictors (in our case, hotspots) and hotspot size estimation. %
They suggest several comments.

First, they illustrate our motivating statement: %
selection is sensitive to the choice of hotspot propensity variance; the pairwise selection performance of the three models with fixed variances varies greatly. %
The model with small variance strongly shrinks the hotspot sizes, which prevents the detection of many associations. The model with large variance identifies more pairs but %
fails to uncover the smallest hotspots; their estimated signals are overwhelmed by noise as a result of insufficient sparsity being induced (also see Appendix~\ref{app_s1_fig}). %
Moreover, arbitrarily fixing hotspot propensity variances to large values may trigger artifactual ``pile-up'' effects when the data are less informative, as discussed in Section \ref{sec_ps}.

Second, the Cauchy model (global shrinkage only) is often able to discriminate the small hotspot signals from the noise, thanks to its global scale inferred from the data, but is not as good for %
pairwise selection and %
estimation of hotspot sizes; %
because it is mostly informed by SNPs with no simulated associations, the global scale concentrates %
towards zero, which over-penalizes large hotspots, hampering the detection of pairwise associations with these hotspots. This phenomenon is of particular concern when signals are extremely sparse, as is thought to be the case in molecular QTL problems. Degeneracy issues can also arise in empirical Bayes settings, see, e.g., \citet{scott2010bayes, van2018learning} for a discussion, and  \citet{van2016many, van2016conditions, van2017adaptive} for solutions based on prior distributions with support truncated away from zero. One may also attempt to improve the Cauchy specification by acknowledging the presence of genomic regions with diverse degrees of functional plausibility and %
introducing region-specific variance parameters to adapt to these degrees. Although inference may then be %
marginally impacted by the overall signal sparsity, such a formulation raises questions %
on the sensitivity to the chosen genome partition. 

Our proposal performs well for selection of both response-predictor pairs and hotspots. Unlike the fixed-scale models, it can clearly separate the small hotspots from the noise. Moreover, the hotspot sizes are well inferred overall: there is some variability depending on the simulated effects (re-drawn for each replicate), with the very small hotspots often underestimated, but the estimated sizes are much closer to the truth than those of the other models, which all strongly overshrink. %
We obtained the hotspot sizes by thresholding the variational posterior probabilities of association at $0.5$, a threshold which corresponds to the \emph{median probability model} rule described by \citet{barbieri2004optimal} as having optimal prediction performance. %
Hence, the flexibility offered by the horseshoe's heavy-tailed local scale parameters improves on %
global scale parameter formulations, whether the parameter values are fixed or inferred.

\subsection{Null model scenario}\label{sec_s1_n}
We examine the behaviour of our approach on data with neither hotspots nor individual associations. We took the data simulated for the first replicate of the ``reference'' scenario discussed in Section \ref{sec_s1}, but randomly shuffled the response sample labels, thus leaving the response correlation structure untouched. We ran the method on eight such permuted datasets and observed no hotspot using the $0.5$-thresholding rule on the variational posterior probabilities of association: there were at most four associated responses per predictor. %
The %
average proportion of false positive pairwise associations was  $2 \times 10^{-5}$.

\subsection{Annealed variational inference in presence of strong multicollinearity}\label{sec_s2} %

The present numerical experiment focuses on data exhibiting strong predictor and response multicollinearity. %
To best reproduce conditions encountered in molecular QTL studies, %
we used real SNP data from the eQTL study described in Section \ref{sec_mot}. We considered a $1.7$ megabase (Mb) region located $\approx 1$ Mb upstream of the MHC region and comprising $200$ variants for which $n =413$ observations were available. %
We distributed five active SNPs across the blocks and simulated $500$ ``active'' responses. Effects were small, with each response having at most $10\%$ of its variability explained by genetic variation. We added another $19,500$ inactive responses, drawn from Gaussian noise. The residual correlation of the responses spanned larger values than in Section \ref{sec_s1}, with block-correlation coefficients $\rho \in (0, 0.5)$.

\begin{figure}[t!]
\begin{center}
\includegraphics[scale=0.396]{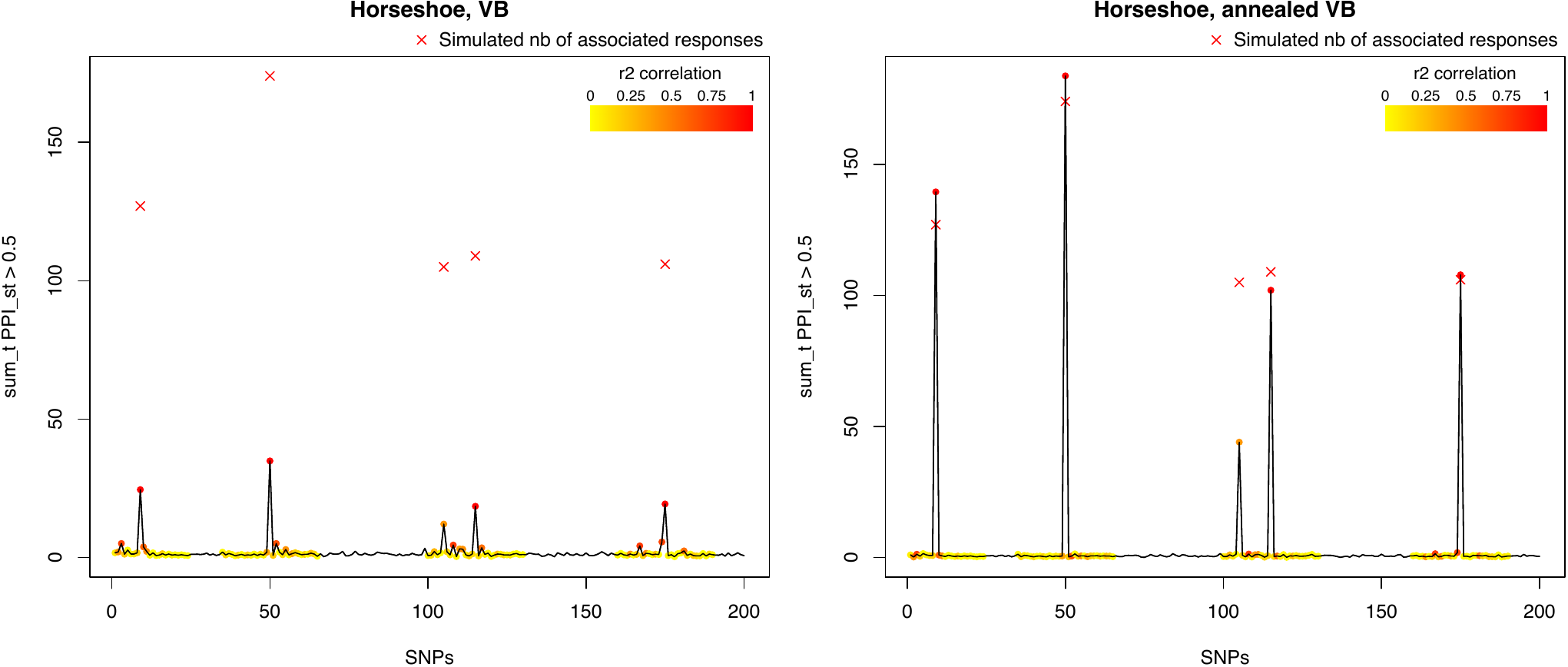}\\
\includegraphics[scale=0.405]{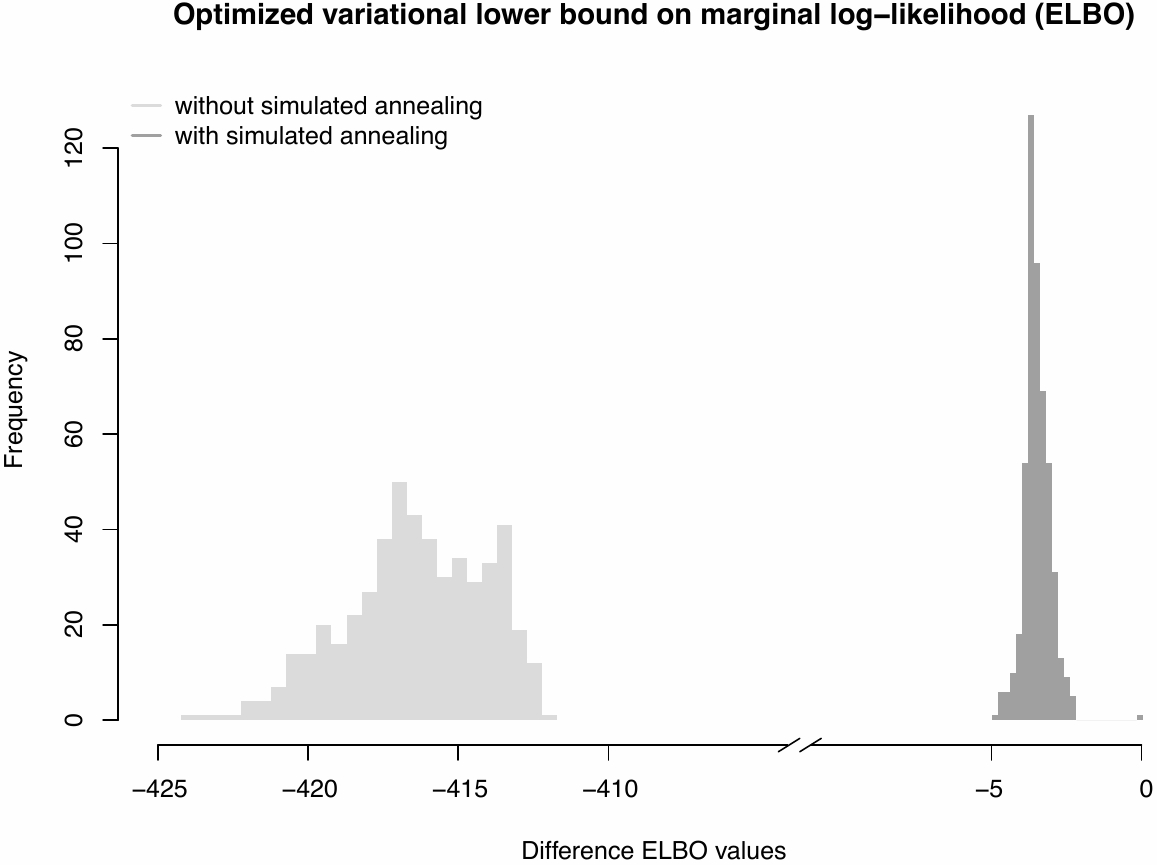}\quad\quad
\includegraphics[scale=0.415]{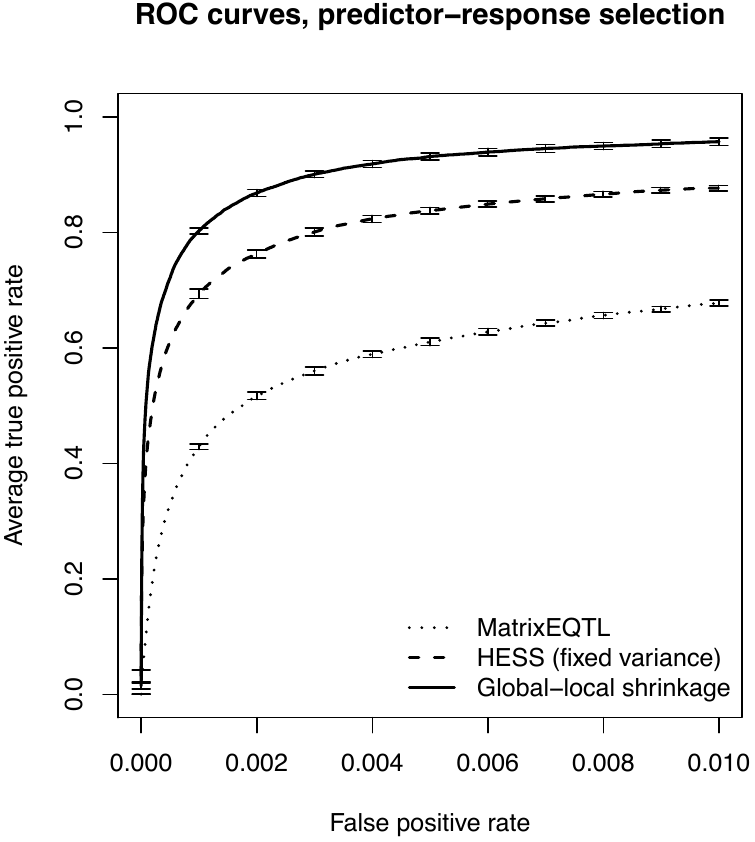}
\end{center}
\caption{\small Performance comparisons between classical and annealed variational inferences, and with competing methods. Problem with responses equicorrelated by blocks with residual correlation $\rho \in (0, 0.5)$; $500$ of them are under genetic control. Candidate predictors are $p = 200$ SNPs from a cohort of European ancestry, for $n = 413$ individuals. Top: Hotspot discrimination achieved by classical (left) and annealed (right) variational inferences, for a problem with $q = 20,000$ responses. The plots show the cumulated number of responses associated per SNP, after thresholding the variational posterior probabilities of association at $0.5$, and averaging over $16$ replicates. The crosses show the simulated sizes of the five hotspots (whose cumulated effects account for at most $10\%$ of the variability of a response). The highlighted %
regions quantify the linkage disequilibrium structure in $r^2$ computed with respect to hotspots $9$, $50$, $115$ and $175$, respectively;  %
hotspot $105$ is correlated with hotspot $115$. Bottom, left: Histograms of optimized lower bound on the marginal log-likelihood (ELBO) with classical and annealed variational inferences, $500$ replicates; the $x$-axis shows the maximum ELBO value subtracted from all other values. Bottom, right: Truncated average ROC curves with 95\% confidence intervals for the MatrixEQTL and HESS methods, and our proposal. The same settings as above are used, except for the number of responses, limited to $q = 10,000$, and the number of hotspots, $15$, whose cumulated effects account for at most $20\%$ of the variability of a response. Both HESS and our proposal have prior expectation $\mathrm{E}_p = 1$ and variance $\mathrm{V}_p = 10$ for the number of SNPs associated with a response;  the value of $\mathrm{E}_p$ is smaller than in Sections \ref{sec_s1} and \ref{sec_s1_n} because there are fewer candidate predictors, and so is the value of $\mathrm{V}_p$, to limit the computational costs of HESS.}\label{fig_s2_ann}
\end{figure}

Figure \ref{fig_s2_ann} %
indicates that the annealed variational algorithm clearly discriminates hotspots. Moreover, when declaring associations using a threshold of $0.5$ on the marginal posterior probabilities, the hotspot sizes were well estimated, except for SNP id $105$. %
In contrast, the non-annealed version of the algorithm struggled to single out the relevant SNPs from their correlated neighbours, especially around SNP id $110$. %
Hence, the behaviour observed in the small experiment of Section \ref{sec_inference} also arises in multiple-response settings.  
We also applied the algorithm with and without annealing on the data from the first replicate, performing $500$ runs each using different starting values. We found that the optimal value reached by the objective function (\ref{eq_lb}) was consistently higher and less variable in the annealed case (Figure \ref{fig_s2_ann}). This was expected, %
as (\ref{eq_lb}) is a lower bound on the marginal log-likelihood, but further suggests that this bound may indeed represent a good proxy for the marginal log-likelihood. %

\subsection{Comparison with other approaches}\label{sec_other_appr}

We conclude this series of simulation experiments by comparing the method with existing approaches. %
We choose %
two competing methods, MatrixEQTL \citep{shabalin2012matrix} and HESS \citep{richardson2010bayesian, bottolo2011bayesian} as representative of two types of approaches: a univariate screening algorithm that tests the SNP-response pairs one by one, and joint hierarchical modelling coupled with parallel chain MCMC inference. 
We restrict the number of simulated responses to $10,000$ in order to ensure a reasonable convergence time for the HESS MCMC run, %
and involve $15$ SNPs in associations. %

We rely on the default settings proposed in the MatrixEQTL and HESS implementations: these correspond, for the former, to using an additive linear model for the genotype effects and  \emph{t}-statistics for significance tests, and for the latter, to running three parallel chains for $22,000$ iterations, discarding the first $2,000$ as burn-in.  Our annealed variational inference procedure was about $30$ times faster than the MCMC inference implemented in HESS, with an average runtime for one replicate of 4 hours and 17 minutes for the former and 5 days and 10 hours for the latter on an Intel Xeon CPU, 2.60 GHz.

As expected, the ROC curves of Figure \ref{fig_s2_ann} indicate that MatrixEQTL performs worse than the two joint approaches. %
It correctly identifies the strong associations but also declares many spurious associations involving SNPs in high linkage disequilibrium. This agrees with the motivating example in Section \ref{sec_mot}; marginal screening often %
provides satisfactory answers when the aim is to highlight \emph{cis} associations at the level of loci, but, because of the multiplicity burden, it often fails to declare weaker effects such as those involved in \emph{trans} associations. 
By borrowing information across all SNPs and responses, HESS achieves much better association recovery. The HESS run is based on a specific choice of hotspot propensity variance, which is hard-coded and not accessible to the user; we expect the performance to vary with other choices of variances, similarly to what was shown in Figure \ref{fig_s1_ann} %
for the approach of \citet{ruffieux2017efficient}. With its global and local variances inferred from the data, our proposal performs best. %
Confronting this performance with MCMC inference further suggests that the independence assumptions underlying the variational mean-field formulation %
do not degrade the quality of variable selection, as shown by \citet{ruffieux2017efficient}. The coupling with simulated annealing results in an excellent selection in our experiments, and in a fraction of the time required by MCMC techniques; this is particularly remarkable in highly multimodal settings. %

\section{A targeted study of hotspot activity with stimulated monocyte expression}\label{sec_application} %

We return to the eQTL data presented in Section \ref{sec_mot}. As discussed there, stimulation of monocytes may boost \emph{trans}-regulatory activity, so %
the analysis of stimulated eQTL data should benefit from a method tailored to the detection of hotspots. In this section, we analyse three genomic regions comprising genes thought to play a central role in the pathogenesis of immune disorders \citep{fairfax2012genetics, fairfax2014innate}:  \emph{NFE2L3} on chromosome $7$,  \emph{IFNB1} on chromosome $9$, and \emph{LYZ} on chromosome~$12$. Each region involves $1,500$ SNPs and spans from $7.5$ to $12$ Mb.  %

The following quality control steps were performed prior to the analyses. For the genotyping, we applied standard filters that exclude SNPs with call rate $< 95\%$, violate the Hardy--Weinberg equilibrium assumption (at nominal $p$-value level $10^{-4}$), or have minor allele frequency $< 0.05$. %
For the transcripts, we considered the top $30\%$ quantile of the interquartile range distributions in each (un)stimulated condition. In order to work with a common set of transcripts across conditions, we then retained the intersection of the transcripts selected in each condition, and checked that no highly varying transcript was dropped in this process. Finally, we discarded samples with unusual transcript values, separately for each condition; the numbers of individuals thus retained were $413$ for unstimulated monocytes, $366$ for IFN-$\gamma$, $260$ for LPS 2h, $321$ for LPS 24h and $275$ for B-cells, and the number of transcripts was $24,461$.

\begin{figure}[t!]
\small
\centering
\includegraphics[scale=0.405]{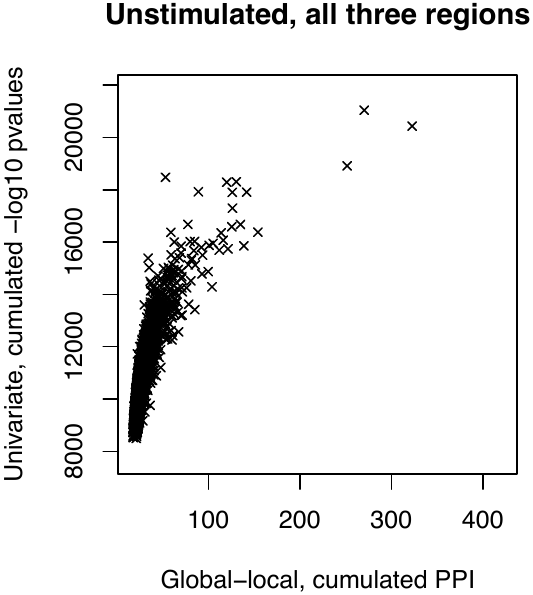}
\includegraphics[scale=0.405]{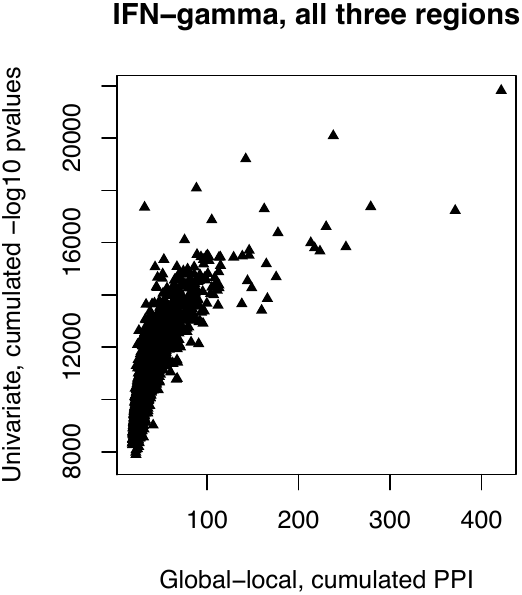}
\includegraphics[scale=0.405]{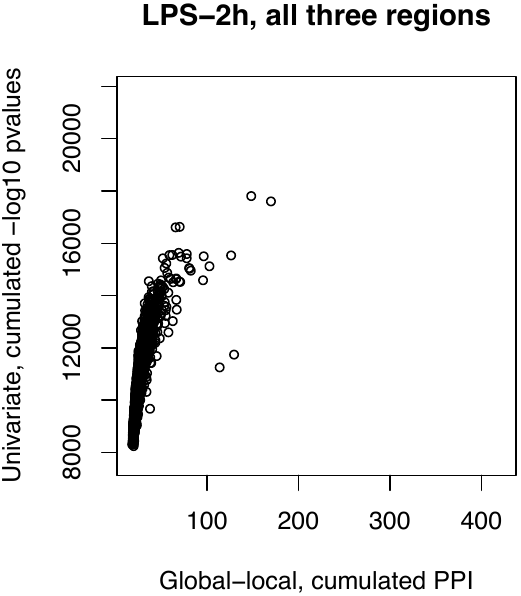}
\includegraphics[scale=0.405]{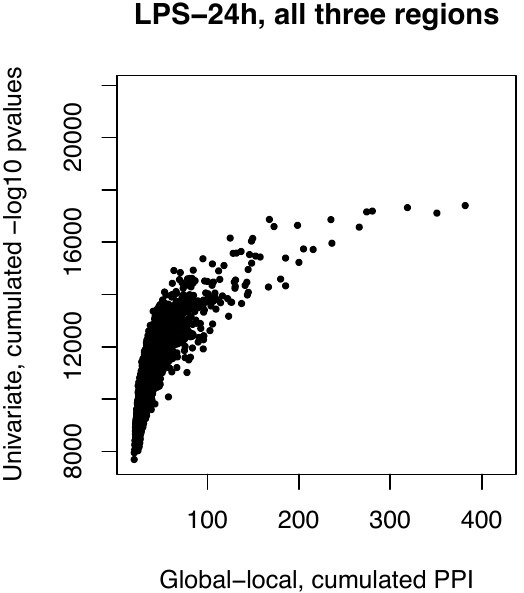}
\includegraphics[scale=0.55]{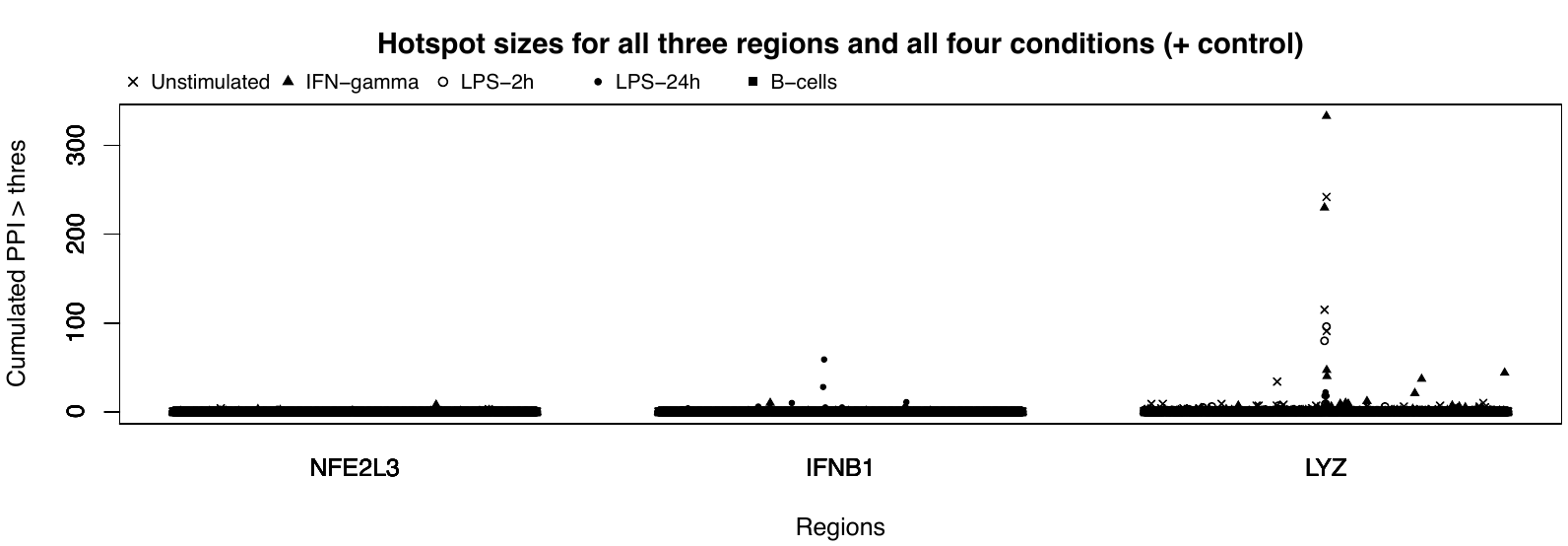}
\caption{\small Hotspots from stimulated eQTL analyses, for the  \emph{NFEL2L3}, \emph{IFNB1} and \emph{LYZ} genomic regions with the four monocyte conditions and the B-cell negative controls. Top: for each condition, raw hotspot evidence for all three regions comprising \emph{NFEL2L3}, \emph{IFNB1} and \emph{LYZ}. Scatterplots with $-\log_{10}$ $p$-values of univariate screening, summed across responses, versus variational posterior probabilities of association obtained by our proposal, summed across responses. Bottom: Hotspot sizes, as declared using a permutation-based FDR of $20\%$.}\label{fig_app_1} %
\end{figure}

We ran our method on each of the three regions, and for all four monocyte conditions, as well as for the B cells, %
resulting in $15$ separate analyses. We employed the same prior base rate of associated pairs as in the simulation of Sections \ref{sec_s1} and \ref{sec_s1_n}, giving a prior average of $\mathrm{E}_p = 0.002 \times p = 3$ SNPs associated with each transcript, and used a variance of $\mathrm{V}_p = 25$. Figure \ref{fig_app_1} compares the evidence for hotspots produced by our proposal and plain univariate screening. It shows the nominal $- \log_{10}$ $p$-values of a univariate screening against the raw posterior probabilities, both summed across responses, and suggests that the two approaches agree on the small or moderate evidence but also that our proposal appears to boost and better distinguish hotspot effects.

In order to derive empirical false discovery rates, we ran a permutation analysis with $30$ replicates for each region and condition, by shuffling the sample labels of the expression matrix; this was computationally feasible thanks to the efficiency of our variational procedure. We then obtained Bayesian false discovery rates %
for a fine grid of thresholds on the variational posterior probabilities of association, and fitted a spline %
in order to derive thresholds corresponding to a false discovery rate of $20\%$.

Figure \ref{fig_app_1} indicates increased \emph{trans}-regulatory activity under stimulation with IFN-$\gamma$ and LPS 24h. This activity was endorsed by the absence of hotspots in the B-cell analysis; indeed, previous studies comparing B cells and monocytes on the three regions suggested that QTL activity was specific to the latter \citep{fairfax2012genetics}, 
so the former may be used as negative controls in our analyses. The degree of activity also varies greatly across the three regions: the  \emph{NFE2L3} region is essentially inactive; its largest hotspot is of size $8$ and appears under IFN-$\gamma$ stimulation, in line with previous observations \citep{fairfax2014innate}. The \emph{IFNB1} region shows more activity under LPS 24h stimulation; this confirms existing work  \citep{fairfax2014innate}, but also reveals more associations with transcripts. The top LPS 24h hotspot in the  \emph{IFNB1} region, rs$3898946$, is an eQTL reported in the GTEx database, for genes \emph{FOCAD} and \emph{MLLT3} in the tibial artery, and for gene \emph{PTPLAD2} in skin tissues; this provides further support for a mechanistic role of this hotspot (to be confirmed in further work). %
The \emph{LYZ} region is known for its high degree of pleiotropy \citep{rotival2011integrating} and is indeed very active in our analyses. %

Although \citet{fairfax2014innate} mostly report stimuli-specific \emph{trans}-regulatory activities, %
our top hotspot hit, rs$6581889$, located only $9$ Kb downstream of the \emph{LYZ} gene, is persistent across all four conditions: it is the largest hotspot in the unstimulated condition with size $242$, in the IFN-$\gamma$ condition with size $333$, and in the LPS 2h condition with size $96$, and it is the second largest hotspot in the LPS 24h condition with size $18$; a Venn diagram showing the transcript overlap across conditions is given in Appendix~\ref{app_real}. %
Hence, the SNP activity was triggered by the  IFN-$\gamma$ stimulation, but was also substantial after $2$ hours and $24$ hours of LPS stimulation. The B-cell data provide a good negative control as they show no activity in the \emph{LYZ} region; the largest number of responses associated with a given SNP is three, and the signal does not co-localize with any hotspot uncovered in monocytes. %
Finally, rs$6581889$ is a 
known \emph{cis} eQTL for \emph{LYZ} and \emph{YEATS4} in multiple tissues, two associations which our analyses confirmed. %

\section{Conclusion}\label{sec_conclusion}

We have introduced a new approach for the efficient detection hotspots in regression problems with tens of thousands of response variables. 
 Our proposal makes novel contributions to both modelling and inference: it introduces a flexible fully Bayesian model for hotspots, and implements an efficient variational inference procedure coupled with simulated annealing. %
It accommodates three essential characteristics of molecular QTL: extreme sparseness of association patterns,  strong multimodality induced by locally correlated genetic variants, and very high dimensions of both the predictor and the response vectors. %

Our simulations indicated that severe sparsity renders ineffective models based only on a global variance for the hotspot propensity. %
 Our global-local model provides sufficient %
 refinement to properly identify the locations and sizes of individual hotspots; it %
 is free of ad-hoc variance choices and automatically adapts to different signal sparsity degrees. %

Collinearity exacerbates posterior multimodality and often causes unstable estimates when obtained by joint inference. As accurate inference is critical to the effective use of the hotspot model in high dimensions, we developed a simulated annealing scheme to improve the exploration of multimodal posterior spaces. In our numerical experiments, the resulting inferences were robust to different algorithm initializations, even on data with marked correlation structures. %
It yielded satisfactory estimates of hotspot sizes in situations where classical variational inference would strongly overshrink.

Our formulation of the first-level model %
 involves two-group mixture priors of the spike-and-slab form for the regression coefficients $ \beta_{st}$ rather than one-group continuous priors, such as the Laplace or the horseshoe
prior, reserving the latter for the second-level modelling of the probability parameters $\omega_{st}$. The relative merits of one-group and two-group shrinkage priors for testing purposes have been a subject of considerable discussion over the past few years %
  \citep[see, e.g., ][]{li2017variable, piironen2017sparsity}. Testing associations for each predictor and response pair can be effectively achieved in a variety of fashions (while permitting some borrowing of information across the responses), including with one-group priors %
on $\beta_{st}$. 
Indeed, good theoretical guarantees for testing and uncertainty quantification are now available for %
the one-group prior framework. 
In particular, \citet{van2014horseshoe, van2016many, van2017adaptive} studied %
posterior concentration %
under the horseshoe prior. \citet{datta2013asymptotic} established optimal asymptotic error rates for a multiple testing decision rule %
based on the horseshoe shrinkage factor and \citet{ghosh2016asymptotic} extended their results to a rich class of one-group priors. %
Interestingly, \citeauthor{datta2013asymptotic} pointed out %
 that %
two-group priors are conceptually very natural for testing tasks, thanks to their noise-signal components, and that the horseshoe decision rule is built on analogies with the two-group model. %
 At the same time, %
 the computational advantages %
 of one-group continuous priors are also often contrasted to the burden caused by the large discrete search space induced by two-group mixture priors. 
 
In our model, we used both a two-group prior and a one-group prior, at different levels. We bypassed the computational burden of the %
two-group formulation by developing a %
variational approach that permits fast inference, even under the spike-and-slab prior of the first level. Crucially, the spike-and-slab formulation directly serves the primary aim of our method: the detection of hotspot effects via a dedicated parameter, $\theta_s$, that allows further borrowing of information across the responses. %
We saw how this leads to a natural %
representation of the hotspot propensities as predictor-specific modulations of the %
probability of association. 
We then made use of the adaptive properties of the global-local horseshoe formulation to 
embed a penalty that adjusts for the response dimension and 
prevents the manifestation of artifactual hotspots when the likelihood is relatively flat; %
 we provided two complementary justifications for its choice.

Several extensions may be considered. First, the illustration on stimulated-monocyte eQTL data suggests extending the model to jointly account for the multiple stimulated states. 
Other types of conditions may benefit from such joint modelling: for instance, molecular QTL data are nowadays often collected for multiple tissues or time points. See \citet{petretto2010new} and \citet{lewin2015mt}, for examples based on the model of \citet{richardson2010bayesian}. Second, 
on the algorithmic side, a natural enhancement would be to embed the annealing temperature as an auxiliary parameter to be inferred.  This would permit adaptive and dynamic control of the temperature schedule and may help to balance the number of temperatures used, and hence the use of resources, with the level of entropy needed for good exploration. \citet{mandt2016variational} have a procedure based on this idea, %
but their proposal requires precomputing an approximation of the joint distribution normalizing constant. Sensible cheap estimates may be envisioned for \emph{large $n$} cases, but are unrealistic for high-dimensional regression.  Obtaining theoretical guarantees for our algorithm would also be beneficial; several recent results on tempered variational approximation for simpler models suggest that desirable convergence properties may be provable for our annealed variational updates \citep{alquier2016properties, alquier2017concentration, yang2017alpha}.

We do not claim that our approach can provide direct conclusive evidence on the functional consequences of the identified hotspots, as this always requires follow-up studies at the level of individual loci. We do argue, however, that it is well suited to highlight promising candidate variants for functional studies, which %
may save substantial investment in prospective research. Our method applies to any type of molecular QTL problem. In particular, it may be used with proteomic and lipidomic expression data, which are gaining in popularity because they may be more closely linked with clinical phenotypes.

\section*{Software}

The software \texttt{atlasqtl} is written in R with C++ subroutines. It is available at \url{https://github.com/hruffieux/atlasqtl}.%

\section*{Acknowledgements}
We are grateful to the editor and the two anonymous referees for their valuable comments that improved the presentation of the paper. 
We thank Armand Valsesia for his helpful comments, and also thank Colin Star, Bruce O’Neel and Jaroslaw Szymczak for giving us access to computing resources.

\section*{Funding}

This research was funded by Nestlé Research (H.R., J.H.), the Alan Turing Institute under the Engineering and Physical Sciences Research Council grant EP/N510129/1 (L.B.),  the MRC grant MR/M013138/1 ``Methods and tools for structural models integrating multiple high-throughput omics data sets in genetic epidemiology'' (L.B.), the UK Medical Research Council programme MRC MC UU 00002/10 (S.R.) and the Alan Turing Institute Fellowship number TU/B/000092 (S.R.).

\newpage
\appendix
\renewcommand\thefigure{\thesection.\arabic{figure}}    
\setcounter{figure}{0}

\section{Complements to motivating example}\label{app_mot}

Figure \ref{fig1_app_mot} illustrates the observations made in Section \ref{sec_mot}  about the drawbacks of univariate screening approaches applied to eQTL data. It indicates that the uncovered distal \emph{trans} effects tend to be smaller than the proximal \emph{cis} effects, and suggests that the identification of \emph{trans} effects is hampered by the redundant \emph{cis} effects detected by the univariate screening at a single locus (here for the transcript \emph{B3GALT6}). %

\begin{figure}[h!]
\small
\includegraphics[scale=0.46]{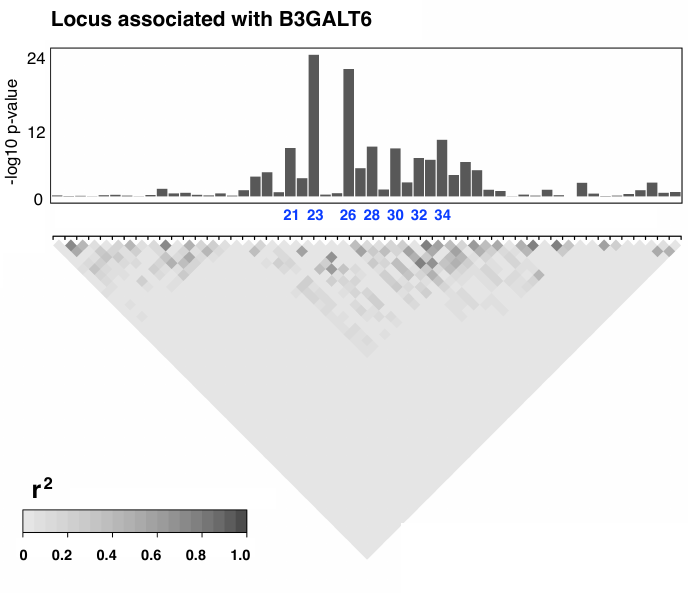}\quad
\includegraphics[scale=0.35]{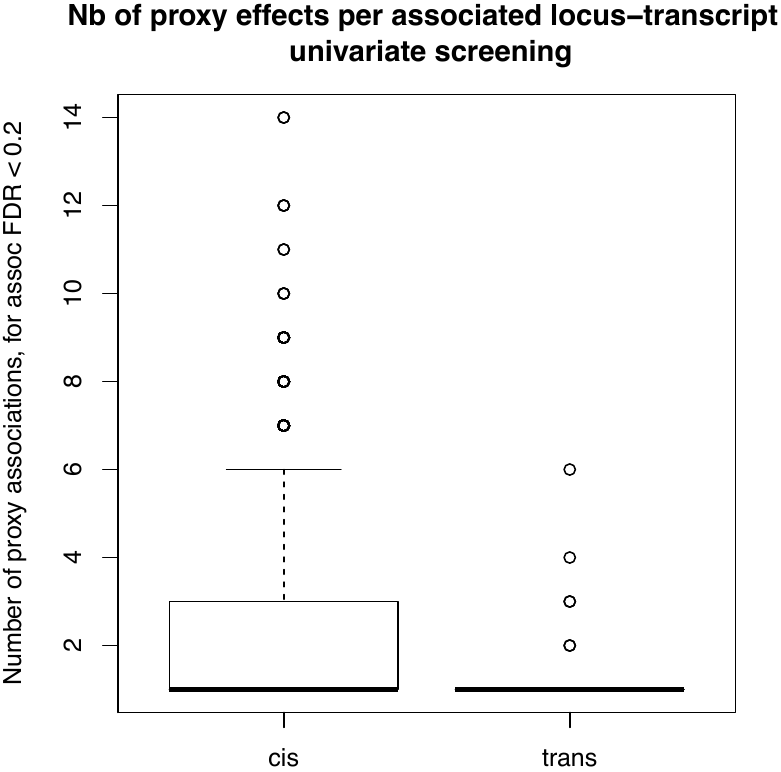}\quad
\includegraphics[scale=0.35]{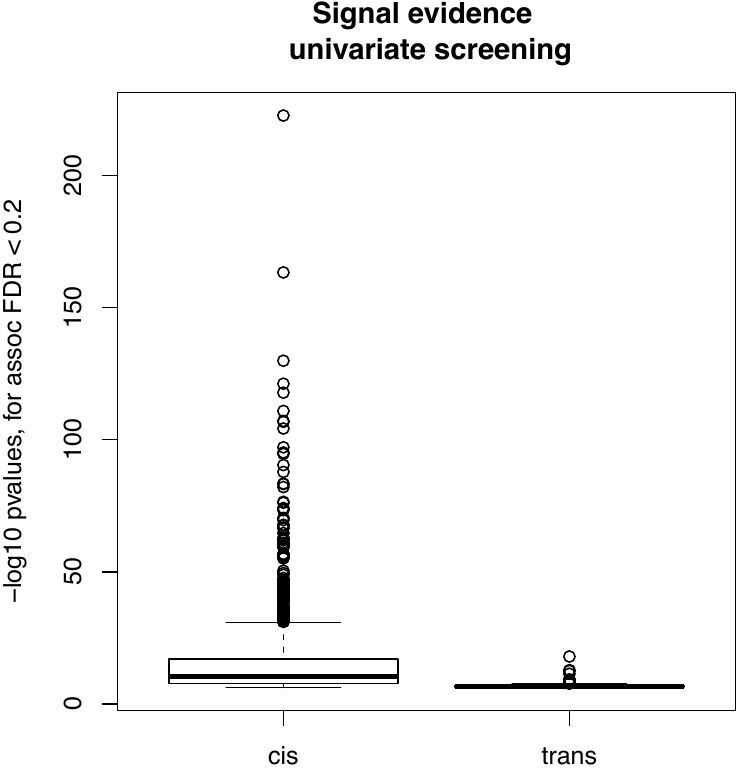}
\caption{\small Detection of \emph{cis} and \emph{trans} associations by a univariate screening approach using a Benjamini--Hochberg FDR of $20\%$. Left: example of linkage disequilibrium plot and Manhattan plot, here for associations with transcript \emph{B3GALT6}. The blue labels indicate seven SNPs \emph{cis}-acting on \emph{B3GALT6} at FDR $20\%$; these effects are likely to be proxies for a single signal in the locus and arise because of the failure of univariate approaches to handle local correlation structures. %
Middle: number of such ``proxy'' associations for \emph{cis} and \emph{trans} effects, based on a linkage disequilibrium threshold of $0.5$ ($r^2$ correlation) and window size $2$ Mb. Right: $-\log_{10} p$-values for the declared effects.}\label{fig1_app_mot} 
\end{figure}

 \section{Hyperparameter specification for top-level priors}\label{app_hyper} %

We describe the hyperparameter settings for the prior distribution of $\zeta_t$. We borrow ideas from \citet{bottolo2011bayesian}, in that we let the response-specific parameter $\zeta_t\overset{\mathrm{iid}}{\sim} \mathcal{N}(n_0, t_0^2)$ control the sparsity level, i.e., the number of predictors associated with each response, and use parameter %
$\theta_s \mid \sigma_0 \overset{\mathrm{iid}}{\sim} \mathrm{Horseshoe}(0, \sigma_0)$, with $\sigma_0 \sim \mathrm{C}^{+}(0, q^{-1/2})$, as a predictor-specific modulator of this level.

We will rely on the following results: for $X \sim \mathcal{N}\left(\mu, \sigma^2\right)$,
\begin{eqnarray}\label{formula1}
\mathrm{E}\lc \Phi\lb X\rb \rc&=&\Phi\lb \frac{\mu}{ \sqrt{1 + \sigma^2}} \rb, \\ 
\label{formula2}\mathrm{E}\lc \Phi\lb X\rb^2 \rc &=& \Phi\lb \frac{\mu}{ \sqrt{1 + \sigma^2}}\rb - 2 \mathrm{T}\lb  \frac{\mu}{ \sqrt{1 + \sigma^2}},  \frac{1}{ \sqrt{1 + 2\sigma^2}} \rb,\end{eqnarray}
where $$\mathrm{T}\lb h, a \rb = \varphi(h) \int_0^a \frac{\varphi(hx)}{1 + x^2}\mathrm{d}x, \qquad a, h \in \R,$$ is Owen's T function \citep{owen1956tables}, with $\varphi(\cdot)$ the standard normal density function.

Equality (\ref{formula1}) can be obtained as follows. Let $Z_1~\sim~\mathcal{N}\lb - \sigma^{-1} \mu, \sigma^{-2}\rb$ and $Z_2 \sim \mathcal{N}\lb 0, 1\rb$ be independent and observe that  
$$\pr\lb Z_1 \leq Z_2 \mid Z_2 = z\rb = \pr\lb Z_1 \leq z\rb = \Phi(\sigma z + \mu), \quad z \in \R,$$
so that
$$\pr\lb Z_1 \leq Z_2 \rb = \int \Phi(\sigma z + \mu) \varphi(z)\mathrm{d}z,$$
which corresponds to the left hand-side of (\ref{formula1}).
But since $Z_1 -Z_2 \sim \mathcal{N}\lb -\sigma^{-1} \mu, \sigma^{-2} + 1\rb$, we also have
$$\pr\lb Z_1 \leq Z_2 \rb = \pr \lb Z_1 - Z_2 \leq 0\rb = \Phi\lb \frac{\mu}{ \sqrt{1 + \sigma^2}} \rb,$$
which gives the result. Equality (\ref{formula2}) can be obtained similarly.

Coming back to the hyperparameter setting, we %
make the simplifying assumption that there is %
no predictor-specific modulation ($\theta_s = 0$) so that, given $\zeta_t$, the prior probability of association between predictor $X_s$ and response $y_t$ is 
$$\mathrm{E}\lb \gamma_{st} \mid \theta_s = 0, \zeta_t \rb = \Phi(\zeta_t).$$ 
We then set $n_0$ and $t_0^2$ by specifying a prior expectation and a prior variance for the number of predictors associated with each response, $p_{\gamma, t} = \sum_{s = 1}^p\gamma_{st}$, $t = 1, \ldots, q$,
\begin{eqnarray*}\mathrm{E}\lb p_{\gamma, t} \mid \theta = 0\rb &=& \mathrm{E}\lc\mathrm{E}\lb p_{\gamma, t} \mid \theta = 0, \zeta_t\rb\rc = p\, \mathrm{E}\lc \Phi\lb \zeta_t\rb \rc\,,\\ 
\mathrm{Var}\lb p_{\gamma, t} \mid \theta = 0\rb &=& \mathrm{Var}\lc\mathrm{E}\lb p_{\gamma, t} \mid \theta = 0, \zeta_t\rb\rc + \mathrm{E}\lc\mathrm{Var}\lb p_{\gamma, t} \mid \theta = 0, \zeta_t\rb\rc\\
&=& p (p - 1) \mathrm{E}\lc \Phi\lb\zeta_t\rb^2\rc + p \mathrm{E}\lc \Phi\lb\zeta_t\rb\rc \ls 1 - p\mathrm{E}\lc \Phi\lb\zeta_t\rb\rc \rs, \end{eqnarray*} %
in which we use (\ref{formula1}) and (\ref{formula2}) with $\mu = n_0$ and $\sigma^2 = t_0^2$. We then solve this system numerically to obtain $n_0$ and $t_0^2$.

We next show that our approximation is asymptotically exact as the total number of responses $q$ tends to infinity, i.e., we have the following lemma:
\begin{lemma}\label{sm_lemma_asy}
Let $p_{\gamma, t} = \sum_{s = 1}^p\gamma_{st}$ the number of predictors associated with response $t = 1, \ldots, q$, we have
\begin{eqnarray}\lim_{q \to \infty} \mathrm{E}\lb p_{\gamma, t}\rb&=&\mathrm{E}\lb p_{\gamma, t} \mid \theta = 0\rb,\label{e_asy}\\ \lim_{q \to \infty}\mathrm{Var}\lb p_{\gamma, t}\rb&=& \mathrm{Var}\lb p_{\gamma, t} \mid \theta = 0\rb .\label{var_asy}\end{eqnarray}
\end{lemma}
\begin{proof}
Since
\begin{eqnarray*}
\mathrm{E}(p_{\gamma, t}) %
&=&  p\mathrm{E} \left\{\Phi\left(\theta_s + \zeta_t\right)\right\}, \\%
\mathrm{Var}(p_{\gamma, t}) &=& p(p-1) \mathrm{E} \left\{\Phi\left(\theta_s + \zeta_t\right)^2\right\}  + p\mathrm{E} \left\{\Phi\left(\theta_s + \zeta_t\right)\right\} \left[1-p\mathrm{E} \left\{\Phi\left(\theta_s + \zeta_t\right) \right\} \right], 
\end{eqnarray*}
the lemma is proved %
by showing that 
\begin{eqnarray}\lim_{q \to \infty} \mathrm{E} \left\{\Phi\left(\theta_s + \zeta_t\right)\right\} &=& \mathrm{E}\lc\Phi\lb\zeta_t\rb\rc,\label{phi_asy}\\ 
\lim_{q \to \infty}\mathrm{E} \left\{\Phi\left(\theta_s + \zeta_t\right)^2\right\}&=& \mathrm{E} \left\{\Phi\left(\zeta_t\right)^2\right\},\label{phi2_asy}\end{eqnarray}
for $\theta_s \mid \sigma_0 \overset{\mathrm{iid}}{\sim} \mathrm{Horseshoe}(0, \sigma_0)$, with $\sigma_0 \sim \mathrm{C}^{+}(0, q^{-1/2})$, i.e., 
$$\theta_s \mid \lambda_s, \sigma_0 \sim \mathcal{N}\left(0, \lambda_s^2 \sigma_0^2 \right), \qquad \lambda_s \overset{\mathrm{iid}}{\sim} \mathrm{C}^+(0, 1)\,, \qquad \sigma_0 \sim  \mathrm{C}^+(0, q^{-1/2}), $$
where $\mathrm{C}^+(\cdot, \cdot)$ is the half-Cauchy distribution.
For (\ref{phi_asy}), we have
\begin{eqnarray} \mathrm{E} \left\{\Phi\left(\theta_s + \zeta_t\right)\right\} &=&  \mathrm{E}\left[\mathrm{E} \left\{\Phi\left(\theta_s + \zeta_t\right)\mid \sigma_0, \lambda_s\right\}\right] \nonumber\\ &=&  \mathrm{E}\left\{\Phi\left(\frac{n_0}{\sqrt{1 + t_0^2 + \sigma_0^2\lambda_s^2}}\right)\right\}\nonumber\\
&=& \frac{4}{\pi^2}\sqrt{q}\int_0^\infty \int_0^\infty \Phi\left(\frac{n_0}{\sqrt{1 + t_0^2 + \sigma_0^2\lambda_s^2}}\right) \frac{1}{1 + q\sigma_0^2}\, \frac{1}{1 + \lambda_s^2} \mathrm{d}\sigma_0 \mathrm{d}\lambda_s\nonumber\\
&=& \frac{4}{\pi^2}\int_0^\infty \int_0^\infty \Phi\left(\frac{n_0}{\sqrt{1 + t_0^2 + q^{-1}\tilde{\sigma}_0^2\lambda_s^2}}\right) \frac{1}{1 + \tilde{\sigma}_0^2}\, \frac{1}{1 + \lambda_s^2} \mathrm{d}\tilde{\sigma}_0 \mathrm{d}\lambda_s\nonumber\\
&\overset{q\to \infty}{\longrightarrow}& \Phi\left(\frac{n_0}{\sqrt{1 + t_0^2}}\right)  = \mathrm{E} \left\{\Phi\left(\zeta_t\right)\right\} \label{show_phi_asy}
\end{eqnarray}
where the second equality uses (\ref{formula1}) noting that $\theta_s + \zeta_t \mid \sigma_0, \lambda_s \sim \mathcal{N}(n_0, t_0^2 + \sigma_0^2\lambda_s^2)$, and where the limit is obtained by the dominated convergence theorem, 
as
$$\frac{1}{1 + \tilde{\sigma}_0^2}\, \frac{1}{1 + \lambda_s^2},\quad \tilde{\sigma}_0 > 0, \lambda_s > 0,$$
is an integrable bound of the integrand. %
Similarly, for (\ref{phi2_asy}), we have
\begin{eqnarray*} \mathrm{E} \left\{\Phi\left(\theta_s + \zeta_t\right)^2\right\} &=&  \mathrm{E}\left[\mathrm{E} \left\{\Phi\left(\theta_s + \zeta_t\right)^2\mid \sigma_0, \lambda_s\right\}\right] \\ &=&  \mathrm{E}\left\{\Phi\left(\frac{n_0}{\sqrt{1 + t_0^2 + \sigma_0^2\lambda_s^2}}\right)\right\} - 2 \mathrm{E} \left\{\mathrm{T}\left(\frac{n_0}{\sqrt{1 + t_0^2 + \sigma_0^2\lambda_s^2}}, \frac{1}{\sqrt{1 + 2t_0^2 + 2\sigma_0^2\lambda_s^2}}\right)\right\},
\end{eqnarray*}
using (\ref{formula2}). We have seen in (\ref{show_phi_asy}) that the first term of the sum converges to $\Phi\left(n_0/\sqrt{1 + t_0^2}\right)$, so it remains to show that 
$$\mathrm{E}\left\{\mathrm{T}\left(\frac{n_0}{\sqrt{1 + t_0^2 + \sigma_0^2\lambda_s^2}}, \frac{1}{\sqrt{1 + 2t_0^2 + 2\sigma_0^2\lambda_s^2}}\right)\right\} \overset{q\to \infty}{\longrightarrow} \mathrm{T}\left(\frac{n_0}{\sqrt{1 + t_0^2}}, \frac{1}{\sqrt{1 + 2t_0^2}}\right)$$
to obtain (\ref{phi2_asy}). The left hand side is
$$
\frac{4}{\pi^2}\int_0^{\infty}\int_0^{\infty}\varphi\left(\frac{n_0}{\sqrt{1 + t_0^2 + q^{-1}\tilde{\sigma}_0^2\lambda_s^2}}\right) \int_0^{\frac{1}{\sqrt{1 + 2t_0^2 + 2q^{-1}\tilde{\sigma}_0^2\lambda_s^2}}} \frac{\varphi\left(\frac{n_0}{\sqrt{1 + t_0^2 + q^{-1}\tilde{\sigma}_0^2\lambda_s^2}}x\right)}{1 + x^2}\mathrm{d}x \frac{1}{1 +\tilde{\sigma}_0^2}\, \frac{1}{1 + \lambda_s^2}  
\mathrm{d}\tilde{\sigma}_0\mathrm{d}\lambda_s,
$$
which converges to 
$$
\frac{4}{\pi^2}\int_0^{\infty}\int_0^{\infty}\varphi\left(\frac{n_0}{\sqrt{1 + t_0^2 }}\right) \int_0^{\frac{1}{\sqrt{1 + 2t_0^2}}} \frac{\varphi\left(\frac{n_0}{\sqrt{1 + t_0^2 }}x\right)}{1 + x^2}\mathrm{d}x \frac{1}{1 +\tilde{\sigma}_0^2}\, \frac{1}{1 + \lambda_s^2}
\mathrm{d}\tilde{\sigma}_0\mathrm{d}\lambda_s = \mathrm{T}\left(\frac{n_0}{\sqrt{1 + t_0^2 }}, \frac{1}{\sqrt{1 +2t_0^2}}\right),
$$
since $$\varphi(0)^2 \int_0^1 \frac{1}{1+x^2}\mathrm{d}x\,  \frac{1}{1 +\tilde{\sigma}_0^2}\, \frac{1}{1 + \lambda_s^2},\quad \tilde{\sigma}_0 > 0, \lambda_s > 0,$$ is an integrable bound on the integrand (with respect to $\tilde{\sigma}_0$ and $\lambda_s$).
\end{proof}

\begin{figure}[t!]
\includegraphics[scale=0.43]{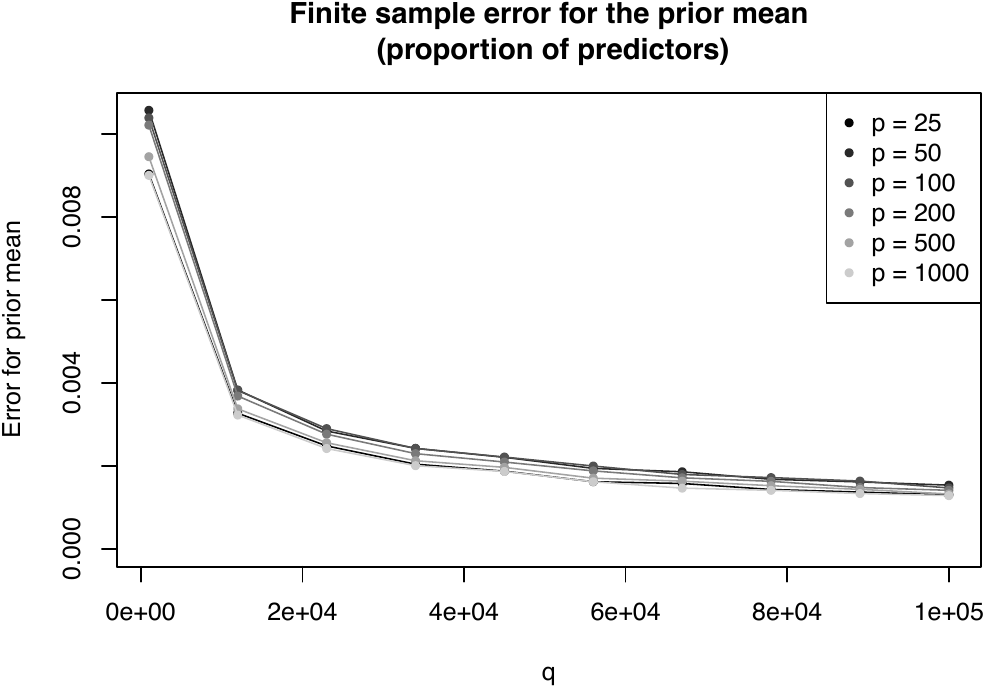}\;\;
\includegraphics[scale=0.43]{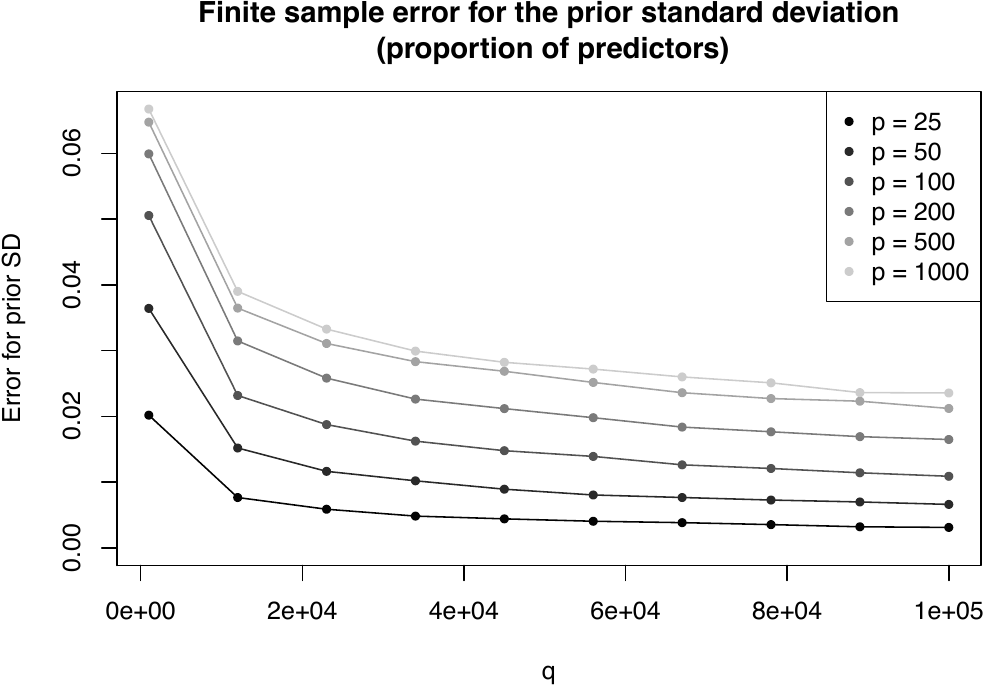}
\caption{\small Finite sample approximation errors
$\vert \mathrm{E}\lb p_{\gamma, t}\rb - \mathrm{E}\lb p_{\gamma, t} \mid \theta = 0\rb \vert$ (left)  and $\vert \mathrm{sd}\lb p_{\gamma, t}\rb -  \mathrm{sd}\lb p_{\gamma, t} \mid \theta = 0\rb\vert$ (right) in terms of proportion of predictors against the response dimension~$q$, for a grid of predictor dimensions~$p$. The hyperparameter settings used are those of the simulation presented in Section \ref{sec_s2}.}\label{sm_diff_phi} 
\end{figure}
\begin{figure}[t!]
\includegraphics[scale=0.46]{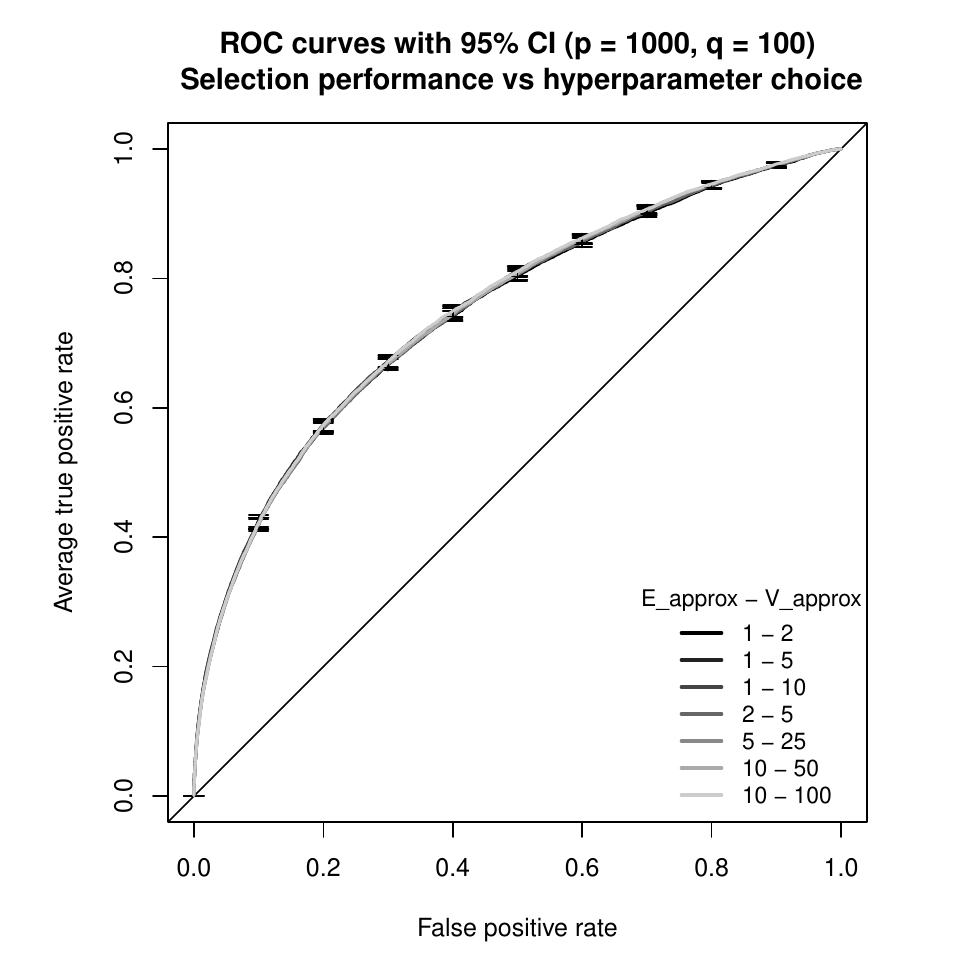}
\includegraphics[scale=0.46]{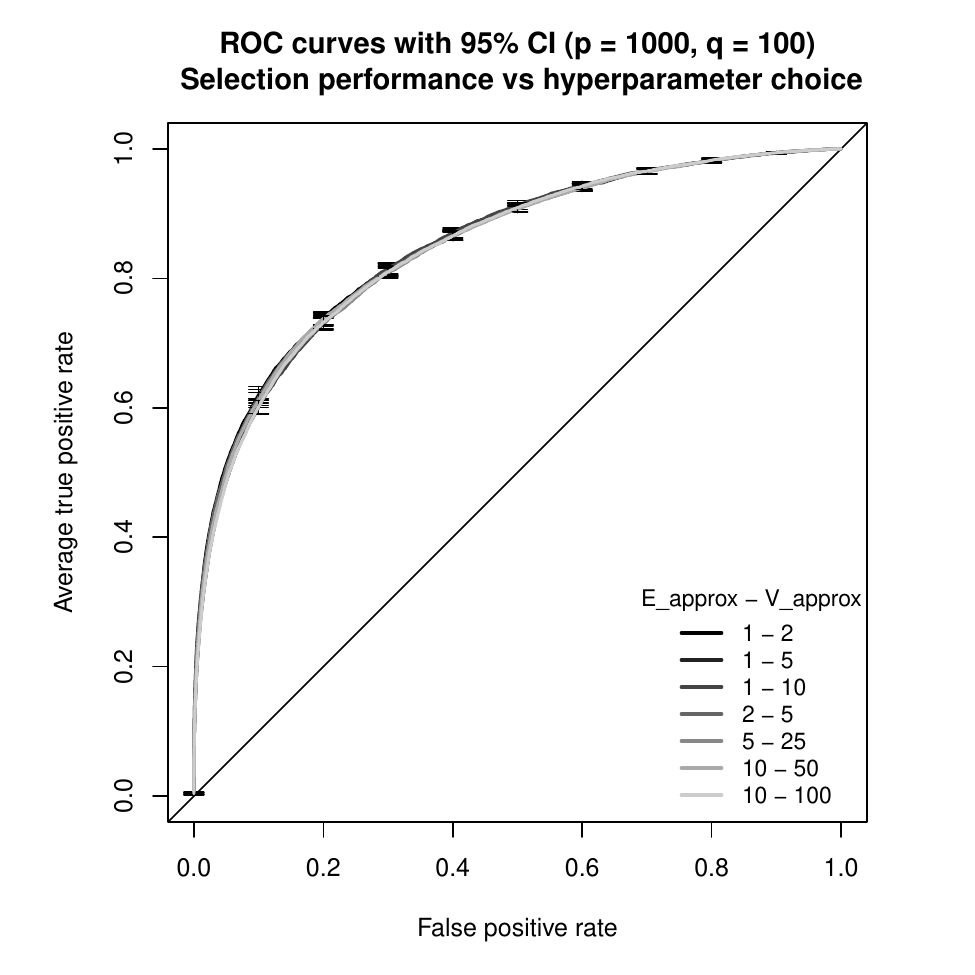}
\caption{\small Hyperparameter sensitivity analysis for a grid of pairs  $\mathrm{E_{approx}} = \mathrm{E}\lb p_{\gamma, t} \mid \theta = 0 \rb$, $\mathrm{V_{approx}} = \mathrm{Var}\lb p_{\gamma, t} \mid \theta = 0\rb$ in a setting where the approximation errors $\vert \mathrm{E}\lb p_{\gamma, t}\rb - \mathrm{E}\lb p_{\gamma, t} \mid \theta = 0\rb \vert$ and $\vert \mathrm{sd}\lb p_{\gamma, t}\rb -  \mathrm{sd}\lb p_{\gamma, t} \mid \theta = 0\rb\vert$ can be substantial, i.e., using $p = 1,000$ and $q = 100$. The response variance explained by the predictors does not exceed $5\%$ (left) and $10\%$ (right). The remaining settings are those of the simulation presented in Section \ref{sec_s2}. All the curves overlap.}\label{sm_roc_sensitivity} 
\end{figure}

Figure \ref{sm_diff_phi} shows the corresponding finite sample approximation errors obtained by Monte Carlo simulation for a grid of predictor and response dimensions. The error on the prior mean proportion of predictors associated with each response is reasonably small; it is roughly constant for all $p$ and increases as $q$ decreases, as expected. The error in the prior standard deviation is larger and leads to inflated variability specifications. Because of this approximation, we warn the reader that specifications for the prior number of predictors associated with each response should not be interpreted strictly. To give a sense of the quality of approximation for user-specific hyperprior elicitations and guide these choices before running the algorithm, we implemented a routine which evaluates the error by Monte Carlo simulation; 
it is provided as the R function \texttt{map\_hyperprior\_elicitation()} in the method's package \texttt{atlasqtl}. %

Importantly however,  we checked by simulation that the impact of these hyperparameter choices on variable selection performance is limited, %
thereby diminishing the practical relevance of the approximation error. %
For instance, Figure \ref{sm_roc_sensitivity} presents a sensitivity analysis for a regime where the approximation can be poor ($p$ large, $q$ small) and indicates that the performance is not affected by different hyperparameter choices with our simulation settings. %
The code for this sensitivity analysis is available at \url{https://github.com/hruffieux/atlasqtl_addendum} and can be readily modified to accommodate a variety of problem parameters (different $p$, $q$, number of samples $n$, proportion of variance explained, number of hotspots, correlation structures) and check %
 the robustness to hyperparameter choices in a problem-specific manner.

\section{Computational details for hotspot propensity shrinkage profile}\label{app_mult}

Some insight into the shrinkage enforced by the horseshoe prior on the hotspot propensities can be gained by examining the effect of the prior specification in terms of the total number of responses $q$  in the second-stage model. 

For a given predictor $X_s$, we introduce the auxiliary variable $z_{s}=(z_{s1}, \ldots, z_{sq})$ to reparametrize the probit link formulation, 
$$\gamma_{st}\mid \theta_s,  \zeta_t \sim \mathrm{Bernoulli}\left\{\Phi(\theta_s + \zeta_t)\right\}\,, \qquad t = 1, \ldots, q\,,$$
as 
$$\gamma_{st} = \one\{z_{st} > 0\}\,,\qquad z_{st} \mid \theta_s, \zeta_t \sim \mathcal{N}\left(\theta_{s} + \zeta_t, 1\right)\,, \qquad t = 1, \ldots, q\,.$$
We next provide the %
proof of the following lemma used in the development of Section \ref{sec_multadj}:
\begin{lemma}\label{sm_lemma_mean}
For $\theta_s \mid \sigma_0, \lambda_s \sim \mathcal{N}\left(0, \sigma_0^2\lambda_s^2\right)$ and $\zeta_t \sim \mathcal{N}(n_0, t_0^2)$, we have 
$$\mathrm{E}\left(\theta_s \mid z_s, \sigma_0, \lambda_s\right) = (1-\kappa_s)  \frac{1}{q} \sum_{t = 1}^q (z_{st} - n_0) + \kappa_s \times 0 = (1-\kappa_s)\bar{z}'_{s},$$
where $\bar{z}'_{s} =  \bar{z}_s - n_0$ and
$$\kappa_s = \frac{1}{1 + \alpha(\sigma_0) \lambda_s^2}$$
is the \emph{shrinkage factor} for hotspot propensities, with $\alpha(\sigma_0) = q (1 + t_0^2)^{-1} \sigma_0^2$.
\end{lemma}

\begin{proof}
Since $\zeta_t \sim \mathcal{N}(n_0, t_0^2)$, we have
\begin{eqnarray*}
p(z_{st} \mid \theta_s) &=& \int p(z_{st} \mid \theta_s, \zeta_t)\,  p(\zeta_{t}) \mathrm{d}\zeta_t\\
&\propto& \int \exp\left\{-\frac{1}{2} \left(z_{st} - \theta_s - \zeta_t\right)^2 -\frac{1}{2t_0^2} \left(\zeta_t - n_0\right)^2\right\} \mathrm{d}\zeta_t\\
&\propto& \exp\left\{-\frac{1}{2} \left(z_{st} - \theta_s \right)^2\right\} \int \exp\left\{-\frac{t_0^2 + 1}{2t_0^2} \left(\zeta_{t}^2 - 2 \frac{t_0^2 (z_{st}-\theta_s)+n_0}{1 + t_0^2 }\zeta_t\right)\right\}\mathrm{d}\zeta_t\\
&\propto& \exp\left[-\frac{1}{2} \left(z_{st} - \theta_s \right)^2 + \frac{1+t_0^2}{2t_0^2} \frac{\left\{t_0^2(z_{st}-\theta_s) + n_0\right\}^2}{\left(1 + t_0^2\right)^2}\right]\\
&\propto& \exp\left\{-\frac{1}{2 \left(1 + t_0^2\right)} \left(z_{st} - \theta_s -n_0\right)^2\right\},
\end{eqnarray*}
i.e., $z_{st} \mid \theta_s \sim \mathcal{N}\left(\theta_s + n_0, 1 + t_0^2\right)$. 
Then, since $\theta_s\mid \sigma_0, \lambda_s \sim \mathcal{N}\left(m_0, \sigma_0^2 \lambda_s^2\right)$, we have
\begin{eqnarray*}
p(\theta_s \mid z_s) &\propto& p(\theta_s)\prod_{t = 1}^q p(z_{st}\mid \theta_s) \\
&\propto& \exp\left\{-\frac{1}{2 \left(1+t_0^2\right)} \sum_{t = 1}^q \left(z_{st}-\theta_s - n_0\right)^2 - \frac{1}{2 \sigma^2_0\lambda_s^2}\left(\theta_s-m_0\right)^2\right\}\\
&\propto& \exp\left[-\frac{1}{2} \left(\frac{q}{1+t_0^2}+\frac{1}{\sigma_0^2\lambda_s^2}\right) \left\{ \theta_s^2 - 2 \,m_0(z_s)\,\theta_s\right\}\right],
\end{eqnarray*}
so $\theta_s \mid z_{s} \sim \mathcal{N}\left(m_0(z_{s}), \left\{q (1+t_0^2)^{-1}+ \sigma_0^{-2}\lambda_s^{-2}\right\}^{-1}\right),$
where 
\begin{eqnarray*}
m_0(z_s) &=&\left(\frac{q}{1+t_0^2}+\frac{1}{\sigma_0^2\lambda_s^2}\right)^{-1}\left\{\frac{1}{1 + t_0^2}\sum_{t = 1}^q \left(z_{st}-n_0\right) + \frac{1}{\sigma_0^2\lambda_s^2} m_0\right\}\\
&=& \frac{\sigma_0^2\lambda_s^2}{q\sigma_0^2\lambda_s^2 + 1 + t_0^2}\sum_{t = 1}^q \left(z_{st}-n_0\right) + \frac{1+t_0^2}{q\sigma_0^2\lambda_s^2 + 1 + t_0^2} m_0\\
&=& \frac{\alpha(\sigma_0)\lambda_s^2}{\alpha(\sigma_0)\lambda_s^2 + 1} \frac{1}{q}\sum_{t = 1}^q \left(z_{st} - n_0\right) +  \frac{1}{\alpha(\sigma_0)\lambda_s^2 + 1} m_0\\
&=&\left(1-\kappa_s\right)\, \frac{1}{q}\sum_{t = 1}^q \left(z_{st} - n_0\right) +  \kappa_s m_0,\\
\end{eqnarray*}
with $$\kappa_s = \frac{1}{\alpha(\sigma_0) \lambda_s^2 + 1}, \qquad \alpha(\sigma_0) = q (1 + t_0^2)^{-1} \sigma_0^2.$$
Parameter $\kappa_s$ is a ``shrinkage factor'' which represents the amount of shrinkage applied on $\theta_s$ towards $m_0$ from  $\frac{1}{q}\sum_{t = 1}^q \left(z_{st} - n_0\right)$; in our horseshoe prior specification, we have $m_0 = 0$, hence the result.
\end{proof}

\begin{figure}[t!]
\centering
\includegraphics[scale=0.4]{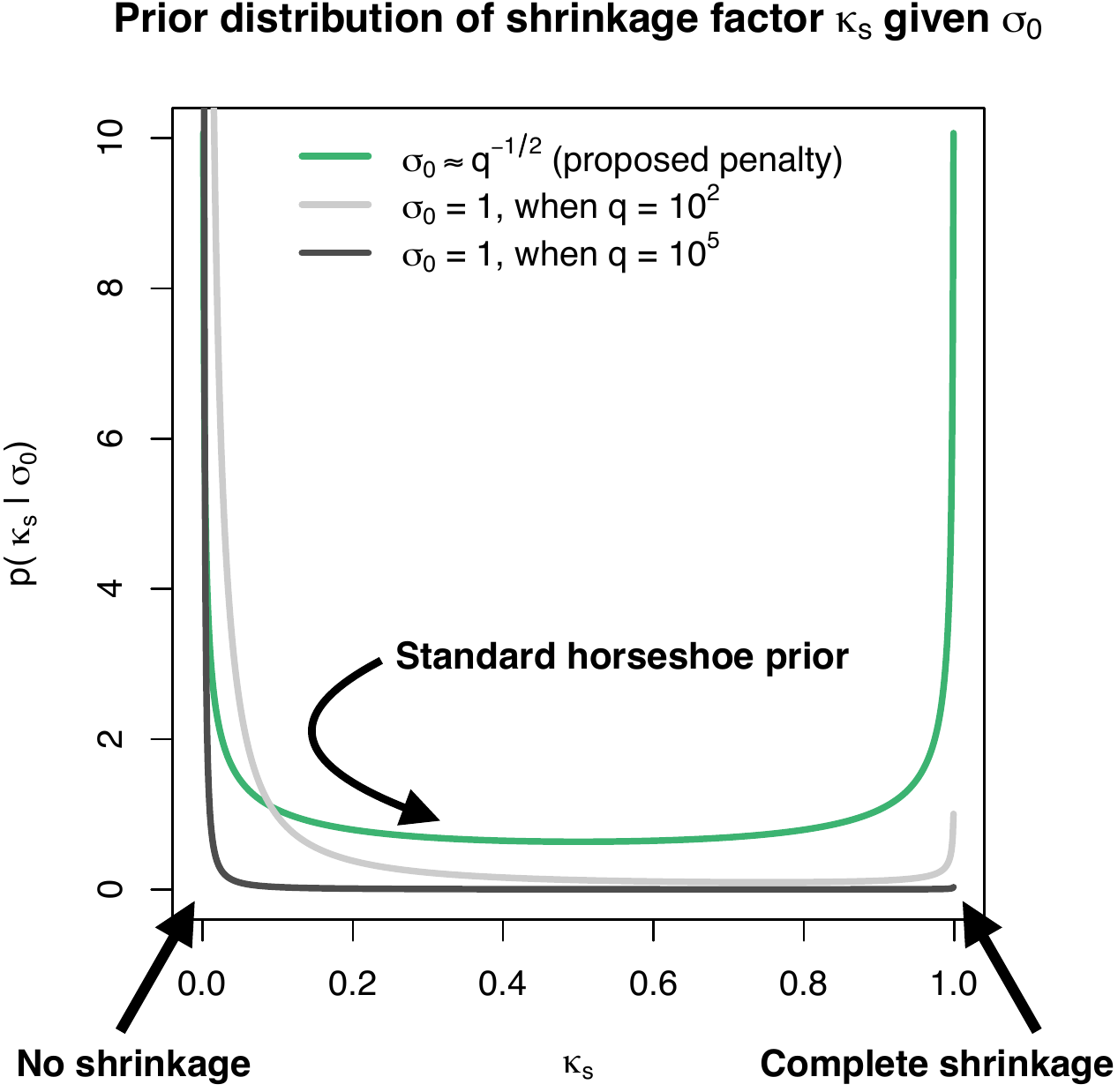}
\caption{\small Prior distribution for the global-local shrinkage factor parameter. Cases without response-multiplicity adjustment (grey): weaker shrinkage enforced as $q$ gets large. Case with the response-multiplicity adjustment (green): standard horseshoe $\mathrm{Beta}(1/2, 1/2)$ shrinkage factor recovered.}\label{sm_fig_shr} %
\end{figure}

\begin{lemma}\label{sm_lemma_prior_kappa}
For $\lambda_s \sim C^{+}(0,1)$, the prior distribution of the shrinkage factor $\kappa_s$ given $\sigma_0$ is
\begin{equation}\label{sm_prior_kappa}p(\kappa_s \mid \sigma_0) =\pi^{-1} \alpha(\sigma_0)^{1/2}\, \kappa_s^{-1/2} (1 - \kappa_s)^{-1/2} \ls 1 + \kappa_s \lc \alpha(\sigma_0) - 1\rc\rs^{-1}, \qquad 0 < \kappa_s < 1\,,\end{equation}
where $\alpha(\sigma_0) = q (1 + t_0^2)^{-1} \sigma_0^2$. 
\end{lemma}
\begin{proof}
Since $\kappa_s  = 1/\left\{\alpha(\sigma_0) \lambda_s^2 + 1\right\}$, the result is obtained by a simple change of variable. Since $\lambda_s$ has a standard half-Cauchy prior, $\lambda_s^2$ has a beta prime distribution with shape parameters both equal to $1/2$,
$$p_{\lambda^2_s}(x) = \pi^{-1} x^{-1/2} (1+x)^{-1}, \qquad x > 0\,.$$
Introducing $g(\kappa_s) := \alpha(\sigma_0)^{-1} \left(\kappa_s^{-1}-1\right) = \lambda_s^2$, we then have
\begin{eqnarray*}
p_{\kappa_s}(\kappa_s) &=& p_{\lambda^2_s}\left(g\left(\kappa_s\right)\right) \vert g'\left(\kappa_s\right)\vert\\
&=& \pi^{-1} \alpha(\sigma_0)^{1/2} \left(\kappa_s^{-1}-1\right)^{-1/2} \left\{1+\alpha(\sigma_0)^{-1}\left(\kappa_s^{-1}-1\right)\right\}^{-1} \alpha(\sigma_0)^{-1} \kappa_s^{-2}\\
&=& \pi^{-1} \alpha(\sigma_0)^{1/2}\, \kappa_s^{-1/2} (1 - \kappa_s)^{-1/2} \ls 1 + \kappa_s \lc \alpha(\sigma_0) - 1\rc\rs^{-1},
\end{eqnarray*}
for $0 < \kappa_s < 1$.
\end{proof}
Figure \ref{sm_fig_shr} shows (\ref{sm_prior_kappa}) for different choices of $\sigma_0$ and response dimensions; in particular, taking $\sigma_0 = q^{-1/2}$ gives $\alpha(\sigma_0) \approx 1$ as $t_0 \ll 1$ in sparse settings, and (\ref{sm_prior_kappa}) approaches the standard horseshoe shrinkage factor prior, $\mathrm{Beta}(1/2, 1/2)$.

\section{Derivation of the annealed variational algorithm}\label{app_vb}

Recall the full model specification. For $q$ centered responses, $y = \lb y_1, \ldots, y_q\rb$, and $p$ centered predictors, $X =\lb x_1, \ldots, x_p\rb$, for $n$ samples ($n \ll p$),
\begin{eqnarray}\label{app_eq_model}
y_t &\mid& \beta_t, \tau_t \sim \mathcal{N}_n\lb X\beta_t, \tau_t^{-1} I_n\rb,\hspace{3.73cm} \tau_t \overset{\mathrm{ind}}{\sim} \mathrm{Gamma}(\eta_t, \kappa_t)\,,\hspace{0.3cm}  t = 1, \ldots, q\,, \nonumber\\
\beta_{st} &\mid& \gamma_{st}, \tau_t, \sigma \sim \gamma_{st}\,\mathcal{N}\lb 0, \sigma^2\,\tau_{t}^{-1}\rb  + (1-\gamma_{st})\,\delta_0\,, \hspace{0.65cm} \sigma^{-2} \sim \mathrm{Gamma}(\nu, \rho)\,,  \hspace{0.6cm} s = 1, \ldots, p\,,\nonumber\\
\gamma_{st}&\mid& \theta_s,  \zeta_t \sim \mathrm{Bernoulli}\left\{\Phi(\theta_s + \zeta_t)\right\},\hspace{2.9cm} \zeta_t \overset{\mathrm{iid}}{\sim}  \mathcal{N}(n_0, t_0^2)\,,\\
\theta_s &\mid& \lambda_s, \sigma_0 \sim \mathcal{N}\left(0, q^{-1}\sigma_0^2 \lambda_s^2 \right)\,,\hspace{0.6cm} \sigma_0 \sim \mathrm{C}^+(0, 1)\,, \hspace{0.65cm}  \lambda_s \overset{\mathrm{iid}}{\sim} \mathrm{C}^+(0, 1)\,, \nonumber 
\end{eqnarray}
where $\delta_0$ is the Dirac distribution, $\Phi(\cdot)$ is the standard normal cumulative distribution function, and $\mathrm{C}^+(\cdot, \cdot)$ is a half-Cauchy distribution.

In order to obtain closed-form updates for our variational algorithm, we apply two reparametrizations. We first rewrite the probit-link level using the classical representation
$$\gamma_{st} = \one\{z_{st} > 0\}\,,\qquad z_{st} \mid \theta_s, \zeta_t \sim \mathcal{N}\left(\theta_{s} + \zeta_t, 1\right),$$
where $z_{st}$ is an auxiliary variable.
We then consider the following formulation for the scale parameters of the horseshoe prior,
$$\sigma_0^{-2} \mid \xi \sim \mathrm{Gamma}\left(\frac{1}{2}, \xi^{-1}\right), \hspace{0.6cm} \xi^{-1} \sim \mathrm{Gamma}\left(\frac{1}{2}, 1\right), \hspace{0.6cm} p\left(\lambda_s^{-2}\right) = \pi^{-1}\left(1 + \lambda_s^2\right)^{-1} \lambda_s^3\,.$$
This parametrization introduces the auxiliary variable, $\xi$; it was first proposed by \citet{neville2014mean}. 
For completeness, we reformulate two lemmas which 
establish the equivalence with the original formulation in (\ref{app_eq_model}). %

\begin{lemma}
If $a$ is a random variable such that
$$a \mid \xi\sim \mathrm{Gamma}\left(\frac{1}{2}, \xi^{-1}\right), \hspace{0.6cm} \xi^{-1} \sim \mathrm{Gamma}\left(\frac{1}{2}, A^{-2}\right), \quad A>0,$$
then 
$a^{-1/2} \sim \mathrm{C}^+(0, A)$.
\end{lemma}

\begin{lemma}
If $a$ is a random variable such that
$$p\left(a\right) = \pi^{-1}\left(1 + a\right)^{-1} a^{-1/2},\quad a>0,$$
then 
$a^{-1/2} \sim \mathrm{C}^+(0, 1)$.
\end{lemma}

\begin{proof}
Both results are straightforward.
\end{proof}

We now provide the updates for all the heated variational parameters. %
Let $T> 1$ be the current temperature from the annealing schedule, let $v$ be the parameter vector of interest, $p(v \mid y)$, true posterior distribution, and $q_T(v)$, the heated mean-field variational approximation.  We maximize the lower bound on the marginal log-likelihood,
$$\mathcal{L}_T(q) = \int q_T(v) \log p(v, y) \mathrm{d}v -   T \int q_T(v) \log q_T(v) \mathrm{d}v\,.$$
We derive the form of the heated variational distribution $q_T(v_j)$ by observing that 
\begin{eqnarray}
\mathcal{L}_T(q) &=& \E_j\ls \E_{-j} \lc \log p(v, y)\rc - T \log q_T(v_j)\rs + \mathrm{cst}\nonumber\\\nonumber\\
& = & \E_j\ls \log \lc \frac{\exp\lc \E_{-j} \log p(v, y)\rc}{q_T(v_j)^T} \rc\rs + \mathrm{cst}\nonumber\\\nonumber\\
& = & T\, \E_j\ls \log \lc \frac{p_{T, -j}(v_j, y) }{q_T(v_j)} \rc\rs + \mathrm{cst},\label{sm_fin}
\end{eqnarray}
where we introduced the distribution $p_{T, -j}(v_j, y) \propto \exp\lc T^{-1}\E_{-j} \log p(v, y)\rc$, and where $\mathrm{E}_{j}(\cdot)$ denotes the expectation with respect to the distribution $q_T(v_j)$, $\mathrm{E}_{-j}(\cdot)$,  the expectation with respect to the distributions $q_T(v_k)$, for all the variables $v_k$ ($k\neq j$), and $\mathrm{cst}$ is constant with respect to $v_j$. The expectation in (\ref{sm_fin}) corresponds to the negative Kullback--Leibler divergence between $q_T(v_j)$ and the $p_{T, -j}(v_j, y)$; $\mathcal{L}_T(q)$ is therefore maximal when $q_T(v_j) = p_{T, -j}(v_j, y)$, i.e., when
\begin{equation}\label{sm_upd}\log q_T(v_j) = T^{-1}\mathrm{E}_{-j} \{\log p(y, v)\} + \mathrm{cst}\,,\qquad\qquad j=1, \ldots, J\,.\end{equation}
For ease of reading, we hereafter drop the subscript $T$ in $q_T(\cdot)$, and write $c = T^{-1}$ and $v^{(r)}_j$ for the $r^{th}$ moment with respect to the approximate posterior distribution $q(v_j)$. We find that, 
$$
q(\beta_{st}, \gamma_{st}, z_{st}) = q(\beta_{st} \mid z_{st}) q(z_{st} \mid \gamma_{st}) q(\gamma_{st})\,,\qquad s = 1, \ldots, p,\; t = 1, \ldots, q
$$ 
with 
\begin{eqnarray*}\label{guess}
\beta_{st} \mid z_{st} > 0, y\; &\sim& \mathcal{N}\left(\mu_{\beta, st}, \sigma^2_{\beta, st}\right)\,,\quad
\beta_{st} \mid z_{st} \leq 0, y\; \sim \delta_0\,,\quad\\
z_{st} \mid \gamma_{st} = \delta, y\; &\sim& \mathcal{TN}\lb \theta_s^{(1)} + \zeta_t^{(1)}, c^{-1}; \lc 0 < (-1)^{1-\delta}z_{st}\rc\rb, \\
\gamma_{st} \mid y &\sim& \mathrm{Bernoulli}\left(\gamma^{(1)}_{st}\right)\,,
\end{eqnarray*}
where   $X \sim \mathcal{TN}\lb \mu, \sigma^2; \lc a < x < b\rc\rb$ denotes a truncated normal variable, 
$$\sigma^{-2}_{\beta, st} = c\, \tau_t^{(1)}\lc \Vert X_{s}\Vert^2 + \left(\sigma^{-2}\right)^{(1)} \rc\,,\qquad
\mu_{\beta, st} = c\, \sigma_{\beta, st}^2\tau_t^{(1)}X_{s}^T\left(y_t-\sum_{j=1, j\neq s}^p\gamma_{jt}^{(1)}\mu_{\beta, jt}X_{j}\right)\,,$$
and 
\begin{eqnarray*}
\frac{1}{\gamma^{(1)}_{st}}&=& 1 +  \exp\ls -c \lc\frac{1}{2} \left(\log \sigma^{-2} \right)^{(1)} + \frac{1}{2} \left(\log \tau_t \right)^{(1)} + \frac{1}{2}\mu^2_{\beta, st}\sigma^{-2}_{\beta, st} + \log \sigma_{\beta, st} \right.\right. \\
&&\left.\left.\hspace{2.2cm}- \log\lc 1- \Phi\lb \theta_s^{(1)} + \zeta_t^{(1)}\rb\rc + \log \Phi\lb \theta_s^{(1)} + \zeta_t^{(1)}\rb\rc \rs.
\end{eqnarray*}
Writing $\alpha_{st} = \theta_s + \zeta_t$, the first moment of $z_{st}$ given $\gamma_{st}$ is
\begin{align*}
\mathrm{E}_q \lb z_{st} \mid \gamma_{st} \rb = \alpha_{st}^{(1)} + c^{-1/2}M\lb c^{1/2}\alpha_{st}^{(1)}, \gamma_{st}\rb ,
\end{align*}
where 
$$M\lb  u, \gamma\rb = \lb-1\rb^{1-\gamma}\frac{\varphi\lb u\rb}{\Phi\lb u\rb^{\gamma}\ls 1-\Phi\lb u\rb\rs^{1-\gamma}}, \qquad u \in \R, \; \gamma = 0, 1, $$ is the inverse Mills ratio and $\mathrm{E}_q\left( \cdot\right)$ is the expectation with respect to the variational distribution $q(\cdot)$. We therefore have
\begin{eqnarray*}z_{st}^{(1)} &=& \gamma_{st}^{(1)} \lb \alpha_{st}^{(1)}+ c^{-1/2}M( c^{1/2}\alpha_{st}^{(1)},1)\rb + (1-\gamma_{st}^{(1)}) \lb \alpha_{st}^{(1)} + c^{-1/2} M(c^{1/2} \alpha_{st}^{(1)},0)\rb \\
&=& c^{-1/2} \gamma_{st}^{(1)} \lc M(c^{1/2} \alpha_{st}^{(1)},1) - M(c^{1/2} \alpha_{st}^{(1)},0)\rc + \alpha_{st}^{(1)} + c^{-1/2} M(c^{1/2} \alpha_{st}^{(1)},0).
\end{eqnarray*}
The second moment of $z_{st}$ given $\gamma_{st}$ is
\begin{align*}
\mathrm{E}_q \lb z_{st}^2 \mid \gamma_{st} \rb  &= c^{-1}  + \lb\alpha_{st}^{(1)} \rb^2 - c^{-1/2} \alpha_{st}^{(1)} M\lb c^{1/2}  \alpha_{st}^{(1)},\gamma_{st}\rb + 2 c^{-1/2} \alpha_{st}^{(1)}M\lb c^{1/2} \alpha_{st}^{(1)},\gamma_{st}\rb \\ &= c^{-1}  + \alpha_{st}^{(1)}\mathrm{E}_q \lb z_{st} \mid \gamma_{st} \rb , 
\end{align*}
which implies that
\begin{align*}z_{st}^{(2)} &= c^{-1} \gamma_{st}^{(1)} +  \gamma_{st}^{(1)} \alpha_{st}^{(1)}\mathrm{E}_q \lb z_{st} \mid \gamma_{st} = 1\rb + c^{-1} (1-  \gamma_{st}^{(1)}) +  (1-\gamma_{st}^{(1)}) \alpha_{st}^{(1)} \mathrm{E}_q \lb z_{st} \mid \gamma_{st} =0\rb\\
&= c^{-1} + \alpha_{st}^{(1)} z_{st}^{(1)},
\end{align*}
and finally its entropy is
\begin{align}\label{entr}
H(z_{st} \mid \gamma_{st}) &= \log\ls \sqrt{\frac{2 \pi e}{c}}\;\Phi\lb c^{1/2}  \alpha_{st}^{(1)}\rb^{\gamma_{st}}\lc 1-\Phi\lb c^{1/2}  \alpha_{st}^{(1)}\rb\rc^{1-\gamma_{st}}\rs  -\frac{1}{2} c^{1/2} \alpha_{st}^{(1)} M\lb c^{1/2}  \alpha_{st}^{(1)}, \gamma_{st}\rb. 
\end{align}

Then, we find
\begin{eqnarray*}
\sigma^{-2} \mid y \sim \mathrm{Gamma}\lb \nu_\sigma,\rho_\sigma\rb,\qquad \lb \sigma^{-2}\rb^{(1)} = \nu_\sigma / \rho_\sigma,
\end{eqnarray*}
with
\begin{align*}
\nu_\sigma=c \lb \nu+\frac{1}{2}\sum_{t=1}^{q}\sum_{s=1}^{p}\gamma_{st}^{(1)}\rb - c + 1, \qquad \rho_\sigma=c \lc\rho+\frac{1}{2}\sum_{t=1}^{q}\sum_{s=1}^{p}\gamma_{st}^{(1)}\lb\mu_{\beta,st}^{2}+\sigma_{\beta,st}^{2} \rb\tau_t^{(1)}\rc.
\end{align*}
The residual precision parameters have
$$\tau_t \mid y \sim \mathrm{Gamma}\left(\eta_{\tau,t}, \kappa_{\tau,t}\right),\qquad
\tau_t^{(1)}= \eta_{\tau,t} /\kappa_{\tau,t}\,,
$$
where
\begin{eqnarray*}
\eta_{\tau,t} &=& c\lb \eta_t+\frac{n}{2}+\frac{1}{2}\sum_{s=1}^p\gamma_{st}^{(1)}\rb - c + 1\,,\\
\kappa_{\tau,t} &=& c \left[ \kappa_t+\frac{1}{2} \Vert y_{t}\Vert ^2- y_{t}^T\sum_{s=1}^p \mu_{\beta, st} \gamma^{(1)}_{st} X_{s} + \sum_{s=1}^{p-1}\mu_{\beta, st}\gamma^{(1)}_{st} X_{s}^T\sum_{j=s+1}^p\mu_{\beta, jt}\gamma^{(1)}_{jt}X_{j}\right.\\
&&\left.\quad  + \frac{1}{2}\sum_{s=1}^p  \gamma_{st}^{(1)}\left(\sigma^2_{\beta,st} + \mu^2_{\beta, st} \right) \left\{ \Vert X_{s}\Vert^2+\left(\sigma^{-2}\right)^{(1)} \right\}\right]\,.
\end{eqnarray*}
We then have
$$\zeta_t \mid y\sim \mathcal{N}\left(\mu_{\zeta, t}, \sigma^2_{\zeta,t}\right),$$
with
$$\sigma^{-2}_{\zeta, t} = c\lb p + t_0^{-2}\rb, \qquad \mu_{\zeta, t} = c\,\sigma^2_{\zeta, t}\left( \sum_{s = 1}^p z_{st}^{(1)} - \sum_{s = 1}^p \theta_s^{(1)} + t_0^{-2} n_0 \right).$$
Similarly, we have
$$ \theta_s \mid y\sim \mathcal{N}\left(\mu_{\theta,s}, \sigma^2_{\theta,s}\right),$$
with
$$\sigma^{-2}_{\theta,s} = c\, q \lc 1 + \left(\sigma_0^{-2}\right)^{(1)} \left(\lambda_s^{-2}\right)^{(1)}\rc, \qquad \mu_{\theta,s} = c\,\sigma^2_{\theta,s}\left(\sum_{t = 1}^q z_{st}^{(1)} - \sum_{t = 1}^q \zeta_{t}^{(1)}\right).$$
The global precision parameters have variational distributions
$$ \sigma_0^{-2} \mid y\sim \mathrm{Gamma}\left(\nu_{\sigma_0}, \rho_{\sigma_0}\right),\qquad \lb \sigma_0^{-2}\rb^{(1)} = \nu_{\sigma_0} / \rho_{\sigma_0},$$
with
$$\nu_{\sigma_0} = \frac{c}{2} \left(p - 1\right) +1, %
\qquad \rho_{\sigma_0} = c \lc \left(\xi^{-1}\right)^{(1)} + \frac{q}{2} \sum_{s = 1}^p\left(\lambda_s^{-2}\right)^{(1)} \left(\mu_{\theta,s}^2 + \sigma_{\theta,s}^2\right)\rc,$$
and
$$ \xi^{-1} \mid y\sim \mathrm{Gamma}\left(\nu_{\xi}, \rho_{\xi}\right),\qquad \lb \xi^{-1}\rb^{(1)} = \nu_{\xi} / \rho_{\xi},$$
with
$$\nu_{\xi} = 1, \qquad \rho_{\xi} = c \lc 1 + \left(\sigma_0^{-2}\right)^{(1)}\rc.$$
Finally, the updates for the local precision parameters are given by the following lemma. 
\begin{lemma}
Let $0 < c \leq 1$. Then 
\begin{equation}\label{sm_lam_c}\lb\lambda_s^{-2}\rb^{(1)} = \frac{\Gamma(-c+2, L_s)}{ L_s\Gamma(-c+1, L_s)} -1\,,\end{equation}
where
$$L_s  = \frac{c\, q}{2} \left(\sigma_0^{-2}\right)^{(1)} \left(\mu_{\theta,s}^2 + \sigma_{\theta, s}^2\right),$$
and $\Gamma(\cdot, \cdot)$ is the incomplete Gamma function.
For $c = 1$, (\ref{sm_lam_c}) reduces to
$$\lb\lambda_s^{-2}\rb^{(1)} = \frac{1}{L_s \exp(L_s)\mathrm{E}_1(L_s)} -1\,,$$
where $E_1(\cdot)$ is the exponential integral function of order 1.
\end{lemma}

\begin{proof}
Write $a_s = \lambda_s^{-2}$ for simplicity. Using (\ref{sm_upd}), one finds
$$q\lb a_s \rb %
\propto (1 + a_s)^{-c} \exp\lb - L_s a_s\rb, \quad a_s >0\,.$$
One then needs to compute
$$ a_s^{(1)} = \frac{\displaystyle\int_0^\infty a_s(1 + a_s)^{-c} \exp\lb - L_s a_s\rb \mathrm{d}a_s}{\displaystyle\int_0^\infty (1 + a_s)^{-c} \exp\lb - L_s a_s\rb \mathrm{d}a_s}\,.$$
The denominator is obtained as 
\begin{eqnarray*}
\int_0^\infty (1 + a_s)^{-c} \exp\lb - L_s a_s\rb \mathrm{d}a_s &=& \exp\lb L_s\rb \int_0^\infty (1 + a_s)^{-c} \exp\lc - L_s (1 + a_s)\rc \mathrm{d}a_s \\
&=&  \exp\lb L_s\rb L_s^{c-1} \,\Gamma(-c+1, L_s)
\end{eqnarray*}
with $\Gamma(s, x) =  \int_{x}^\infty t^{s-1} \mathrm{e}^{-t}\mathrm{d}t$.
The numerator can be decomposed as
$$
\int_0^\infty a_s(1 + a_s)^{-c} \exp\lb - L_s a_s\rb \mathrm{d}a_s = \int_0^\infty (1 + a_s)^{1-c} \exp\lc - L_s a_s\rc \mathrm{d}a_s - \int_0^\infty (1 + a_s)^{-c} \exp\lc - L_s a_s\rc \mathrm{d}a_s $$
The second term is the denominator computed above (changing the sign). The first term can be computed in a similar fashion as
$$ \int_0^\infty (1 + a_s)^{1-c} \exp\lc - L_s a_s\rc \mathrm{d}a_s  =  \exp\lb L_s\rb L_s^{c-2} \,\Gamma(-c+2, L_s)\,.$$
The first part of the lemma follows immediately. The second part is trivially obtained by noting that 
$\Gamma(1, L_s) = e^{-L_s}$ and $\Gamma(0, L_s) = \displaystyle \int_{L_s}^{\infty} t^{-1}e^{-t} \mathrm{d}t= \mathrm{E}_1\left(L_s\right)$.
\end{proof}
To avoid overflow/underflow issues and ensure numerical stability, we implemented these updates using the \emph{log-sum-exp} formulation \citep{calafiore2014optimization} where appropriate and we used an iterative scheme based on continued fractions for evaluating $\exp(x)\mathrm{E}_1(x)$, $x>0$, similarly as described in \citet{neville2014mean}. 

We now provide the details for the variational lower bound, $\mathcal{L}(q)$, of the marginal log-likelihood, $\log p(y)$; $\mathcal{L}(q)$ is evaluated to monitor convergence at each iteration, once the final temperature $T = 1$ has been reached:
\begin{eqnarray*}
\mathcal{L}(q)&=& \int q(v)\log\left\{\frac{p(y, v)}{q(v)}\right\}\mathrm{d}v \\
&=& \sum_{t = 1}^q \mathcal{L}_y \lb y_{t}\mid \beta_{t},\gamma_{t}, \tau_t\rb + 
\sum_{s=1}^{p}\sum_{t=1}^{q}\mathcal{L}_{\beta, \gamma}\lb \beta_{st},\gamma_{st}, z_{st}\mid \sigma^{-2}, \tau_t, \theta_s, \zeta_t\rb + \sum_{t=1}^{q}\mathcal{L}_\zeta\lb \zeta_{t}\rb  \\ 
&& + \sum_{s=1}^{p}\mathcal{L}_\theta\lb \theta_{s}\rb  + \mathcal{L}_{\sigma_0}\lb \sigma_0^{-2} \rb  + \mathcal{L}_\xi\lb \xi^{-1} \rb + \sum_{s = 1}^p \mathcal{L}_\lambda\lb \lambda_s^{-2} \rb
+ \mathcal{L}_\sigma\lb \sigma^{-2}\rb +\sum_{t=1}^{q}\mathcal{L}_\tau\lb \tau_t \rb, 
\end{eqnarray*}
where
\begin{eqnarray*}
\mathcal{L}_y\left(y_t \mid \beta_t, \tau_t\right) &=&\mathrm{E}_q\left\{\log p(y_t \mid \beta_t, \tau_t)\right\} \\
&=&-\frac{n}{2} \log(2\pi)+\frac{n}{2}\mathrm{E}\left(\log \tau_t \right)-\tau_t^{(1)}\left\{\kappa_{\tau,t}-\frac{1}{2}\sum_{s=1}^p\gamma_{st}^{(1)} \left(\sigma_{\beta, st}^2+\mu_{\beta, st}^2\right)\left(\sigma^{-2}\right)^{(1)}-\kappa_t\right\}\,,\hspace{0.1cm}
\end{eqnarray*}
\begin{eqnarray*}
\mathcal{L}_{\beta, \gamma}\lb \beta_{st},\gamma_{st}, z_{st}\mid \sigma^{-2}, \tau_t, \theta_s, \zeta_t\rb  &=& \mathrm{E}_{q}\log p\lb \beta_{st}\mid \gamma_{st}, \sigma^{-2}, \tau_t\rb +\mathrm{E}_{q}\log p\lb \gamma_{st}\mid z_{st}\rb \\
&& +\, \mathrm{E}_{q}\log p\lb z_{st}\mid \theta_s, \zeta_t\rb -\mathrm{E}_{q}\log q\lb \beta_{st},\gamma_{st}, z_{st}\rb \\
&=&\mathcal{L}_{\beta}\lb \beta_{st} \mid \gamma_{st}, z_{st}, \sigma^{-2}, \tau_t \rb + \mathcal{L}_{\gamma}\lb \gamma_{st}, z_{st}\mid \theta_s, \zeta_t\rb,\hspace{2cm}
\end{eqnarray*}
with
\begin{eqnarray*}
\mathcal{L}_\beta\lb \beta_{st} \mid \gamma_{st}, z_{st}, \sigma^{-2}, \tau_t \rb  &=&  \mathrm{E}_{q}\log p\lb \beta_{st}\mid \gamma_{st}, \sigma^{-2}, \tau_t\rb -\mathrm{E}_{q}\log q\lb \beta_{st}\mid z_{st}\rb\\
&=& \frac{1}{2}\gamma_{st}^{(1)}\lc \mathrm{E}_{q}\lb\log \sigma^{-2}\rb +\mathrm{E}_{q}\lb\log \tau_t\rb-\lb \mu_{\beta,st}^{2}+\sigma_{\beta,st}^{2}\rb\lb \sigma^{-2}\rb ^{(1)} \tau_t^{(1)}\rc \\ 
&&+\, \frac{1}{2} \gamma_{st}^{(1)}\lb \log\sigma_{\beta, st}^{2} + 1\rb,\hspace{0.8cm}
\end{eqnarray*}
\begin{eqnarray*}
\mathcal{L}_{\gamma}\lb \gamma_{st}, z_{st}\mid \theta_s, \zeta_t\rb &=&  \mathrm{E}_{q}\log p\lb \gamma_{st}\mid z_{st}\rb + \mathrm{E}_{q}\log p\lb z_{st} \mid \theta_s, \zeta_t\rb -\mathrm{E}_{q}\log q\lb z_{st}\mid \gamma_{st}\rb-\mathrm{E}_{q}\log q\lb \gamma_{st}\rb\\
& = & \mathrm{E}_{q}\ls \gamma_{st} \log \one\{z_{st}>0\} +\lb1-\gamma_{st}\rb\log\one\{z_{st}\leq0\} -\frac{1}{2}\log(2\pi) - \frac{1}{2} (z_{st} - \theta_s - \zeta_t)^2 \rs\\
&&\,+ H \lb z_{st} \mid \gamma_{st} = 1\rb \gamma_{st}^{(1)} + H \lb z_{st} \mid \gamma_{st} = 0\rb \lb 1 - \gamma_{st}^{(1)}\rb -\gamma_{st}^{(1)} \log \gamma_{st}^{(1)}  \\
&&\, -\lb 1- \gamma_{st}^{(1)}\rb \log\lb 1-  \gamma_{st}^{(1)}\rb\\
&=&  \lb 1- \gamma_{st}^{(1)}\rb\log\lc 1- \Phi\lb \alpha_{st}^{(1)}\rb\rc + \gamma_{st}^{(1)} \log \Phi\lb \alpha_{st}^{(1)}\rb - \frac{1}{2} \sigma_{\theta, s}^2 - \frac{1}{2} \sigma_{\zeta, t}^2-\gamma_{st}^{(1)} \log \gamma_{st}^{(1)} \\
&&\,  -\lb 1- \gamma_{st}^{(1)}\rb \log\lb 1-  \gamma_{st}^{(1)}\rb,
\end{eqnarray*}
where %
$H(\cdot)$ is the entropy (\ref{entr}) (which cancels out).
We then find
\begin{eqnarray*}
\mathcal{L}_\zeta\lb \zeta_t \rb &=& \mathrm{E}_q\left\{\log p(\zeta_t)\right\} - \mathrm{E}_q\left\{\log q(\zeta_t)\right\} \\
 &=&\frac{1}{2} \lc - \log t_0^2 +  \log\left(\sigma^2_{\zeta,t}\right) - t_0^{-2}\left(\mu_{\zeta,t} - n_0\right)^2 -  t_0^{-2}\sigma^2_{\zeta,t} + 1\rc,\hspace{4cm}
\end{eqnarray*}
\begin{eqnarray*}
\mathcal{L}_\theta\lb \theta_s \rb &=& \mathrm{E}_q\left\{\log p(\theta_s)\right\} - \mathrm{E}_q\left\{\log q(\theta_s)\right\} \\
&=&\frac{1}{2} \lc \lb\log \sigma_0^{-2}\rb^{(1)} + \log(q)+  \lb\log \lambda_s^{-2}\rb^{(1)} +  \log\left(\sigma^2_{\theta, s}\right)    - q\lb\sigma_0^{-2}\rb^{(1)} \lb\lambda_s^{-2}\rb^{(1)}\lb \mu_{\theta,s}^2  + \sigma^2_{\theta,s} \rb  + 1\rc,
\end{eqnarray*}
\begin{eqnarray*}
\mathcal{L}_{\sigma_0}\lb \sigma_0^{-2} \rb &=& \mathrm{E}_q\left\{\log p(\sigma_0^{-2})\right\} - \mathrm{E}_q\left\{\log q(\sigma_0^{-2})\right\} \\
&=& \lb \frac{1}{2}-\nu_{\sigma_0}\rb\lb\log\sigma_0^{-2}\rb^{(1)}-\lc \lb\xi^{-1}\rb^{(1)} -\rho_{\sigma_0}\rc\lb\sigma_0^{-2}\rb^{(1)}+ \frac{1}{2}\lc\log\xi^{-1}\rb^{(1)} -\nu_{\sigma_0}\log\rho_{\sigma_0}\\&&-\frac{1}{2}\log\pi +\log\Gamma\lb \nu_{\sigma_0}\rb ,\hspace{1cm}
\end{eqnarray*}
\begin{eqnarray*}
\mathcal{L}_\xi\lb \xi^{-1} \rb &=& \mathrm{E}_q\left\{\log p( \xi^{-1} )\right\} - \mathrm{E}_q\left\{\log q( \xi^{-1})\right\} \\
&=& \lb \frac{1}{2}-\nu_{\xi}\rb\lb\log\xi^{-1}\rb^{(1)}-\lb 1 -\rho_{\xi}\rb\lb\xi^{-1}\rb^{(1)} -\nu_{\xi}\log\rho_{\xi} -\frac{1}{2}\log\pi +\log\Gamma\lb \nu_{\xi}\rb, \hspace{1cm}%
\end{eqnarray*}
\begin{eqnarray*}
\mathcal{L}_\lambda\lb \lambda_s^{-2} \rb &=& \mathrm{E}_q\left\{\log p(\lambda_s^{-2} )\right\} - \mathrm{E}_q\left\{\log q(\lambda_s^{-2})\right\} \\
&=& -\log \pi - \frac{1}{2} \lb\log\lambda_s^{-2}\rb^{(1)} + L_s \lc \lb\lambda_s^{-2}\rb^{(1)} + 1 \rc  + \log \mathrm{E}_1 \lb L_s \rb,\hspace{4cm}
\end{eqnarray*}
\begin{eqnarray*}
\mathcal{L}_\sigma\lb \sigma^{-2}\rb  &=& \mathrm{E}_{q}\log p\lb\sigma^{-2}\rb-\mathrm{E}_{q}\log q\lb\sigma^{-2}\rb \\ 
&=&\lb \nu-\nu_\sigma\rb\lb\log\sigma^{-2}\rb^{(1)}-\lb \rho-\rho_\sigma\rb\lb\sigma^{-2}\rb^{(1)}+\nu\log\rho-\nu_\sigma\log\rho_\sigma-\log\Gamma\lb \nu\rb +\log\Gamma\lb \nu_\sigma\rb ,
\end{eqnarray*}
\begin{eqnarray*}
\mathcal{L}_\tau\lb \tau_t \rb &=& \mathrm{E}_q\left\{\log p(\tau_{t})\right\} - \mathrm{E}_q\left\{\log q(\tau_{t})\right\} \\
 &=&\left(\eta_t-\eta_{\tau, t}\right)\left(\log \tau_t \right)^{(1)}-\left(\kappa_t-\kappa_{\tau, t}\right)\tau_t^{(1)}+\eta_t\log \kappa_t - \eta_{\tau, t}\log \kappa_{\tau, t} -\log \Gamma(\eta_t) +\log \Gamma(\eta_{\tau, t}) \,.
\end{eqnarray*}
All variational updates and terms composing the variational objective function can be computed in closed form, albeit using special functions. In particular, for $X~\sim~\mathrm{Gamma}(a, b)$, we have 
$$\mathrm{E}\lb \log X\rb =\Psi\lb a\rb -\log\left( b\right),$$
where $\Psi(\cdot)$ stands for the digamma function. To optimize computational efficiency, updates are made by blocks, in a vectorized fashion, for all responses; convergence under such a scheme is guaranteed by the concavity of the objective function in each of the subvectors composing the blocks. Moreover, the variational objective function, $\mathcal{L}(q)$, is guaranteed to increase monotonically at every iteration, which provides a useful check against mistakes in the computations or the implementation. 
The algorithm returns all variational parameters after convergence; in particular, the variational posterior quantities
$$\mathrm{E}_q\left(\gamma_{st}\right) = \gamma_{st}^{(1)}, \qquad \mathrm{E}_q\left(\beta_{st}\right) = \gamma_{st}^{(1)} \mu_{\beta, st}, \qquad \mathrm{E}_q\left(\theta_s\right) = \mu_{\theta,s}, \qquad s = 1, \ldots, p, \; t = 1, \ldots, q,$$ can be directly employed to perform predictor-response and hotspot selection.
\newpage
\section{Complements to simulation experiments}
We provide additional performance illustrations for the simulation experiments presented in Section \ref{sec_simulations}. 

\subsection{\,Simulation study 1: performance with global-local modelling}\label{app_s1_fig}

Figure \ref{app_fig_s1_ann} compares hotspot selection performance for the five models discussed in Section \ref{sec_s1} on the reference scenario. It also compares the cumulated posterior probabilities for the small simulated hotspots for the fixed-variance model with largest variance and our global-local proposal.   
\begin{figure}[h!]
\begin{center}
\includegraphics[scale=0.5]{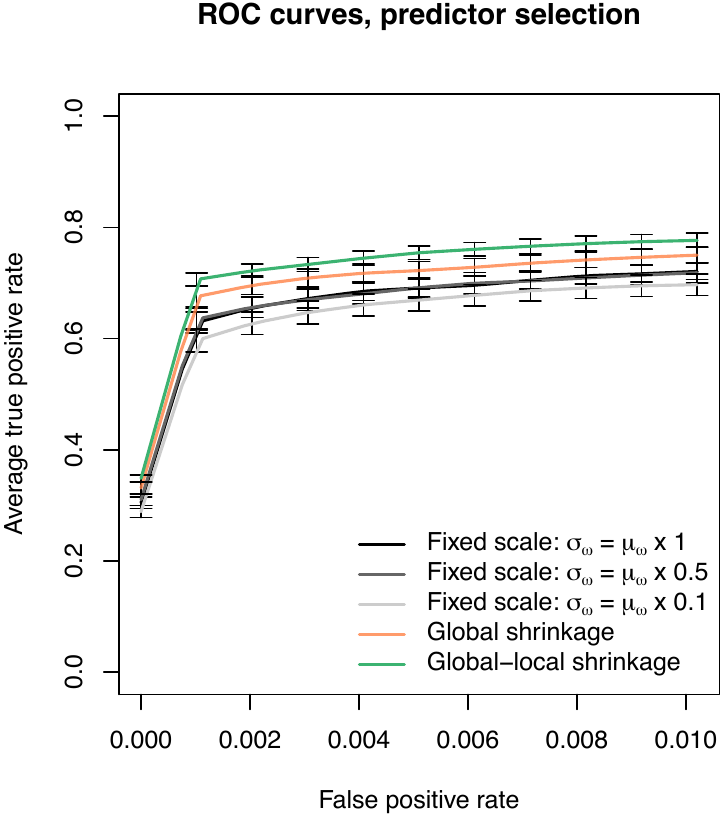}%
\qquad \qquad
\includegraphics[scale=0.5]{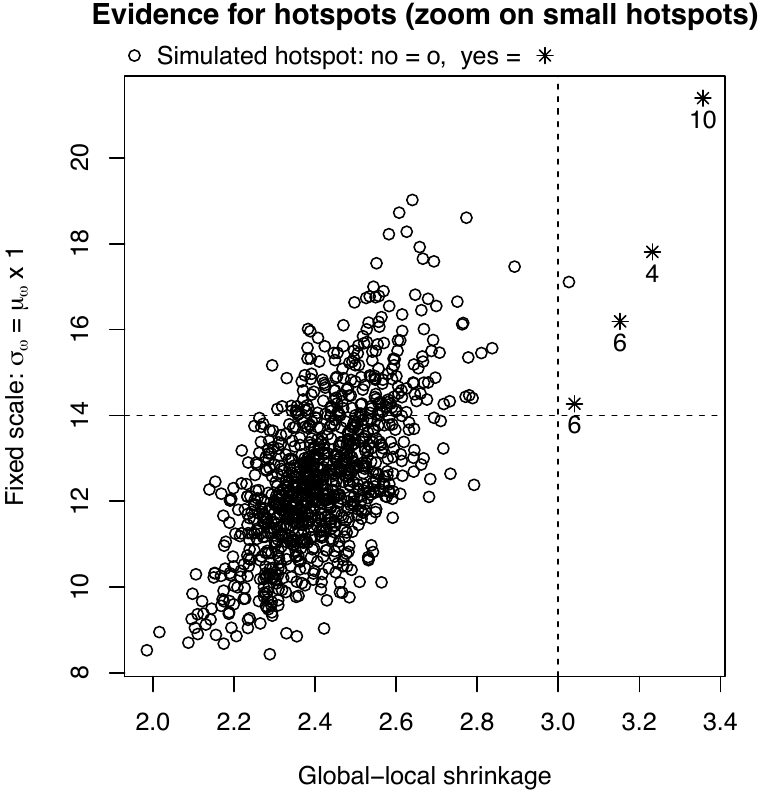}
\end{center}
\caption{\small Performance of five hotspot modelling approaches. %
Left: truncated average receiver operating characteristic curves for hotspot selection with 95\% confidence intervals obtained from $64$ replications. 
Right: evidence for hotspots computed as, for each candidate predictor, the sum of its posterior probabilities of associations with all responses; average over $16$ replications. Zoom on the noise level and four smallest hotspots, with simulated sizes $4$, $6$, $6$ and $10$. The dashed lines highlight how the global-local shrinkage proposal is better than the fixed-scale model with $\sigma_\omega = \mo$ in discriminating weak hotspot signals from the noise. The data comprise $p = 1,000$ simulated SNPs with $20$ hotspots, and $q = 20,000$ responses, of which $200$ are associated with at least one hotspot, leaving the rest of the responses unassociated. 
The block-autocorrelation coefficients for SNPs were drawn from the interval $(0.75, 0.95)$, and the residual block-equicorrelation coefficients for responses, from the interval $(0, 0.25)$. At most $25\%$ of each response variance is explained by the hotspots. For the fixed-variance models, we used a base-rate of $\mo = 0.002$, and scales of $\sigma_\omega = \mo \times \{1, 0.5, 0.1\}$.}\label{app_fig_s1_ann}
\end{figure}

\subsection{\,Simulation study 1: performance for a grid of correlation levels}\label{app_s1_corr}

Tables \ref{sm_tb_corr_grid} and \ref{sm_tb_corr_grid2} compare the pairwise and hotspot selection performance for the five models in Section \ref{sec_s1} by varying the correlation levels of the predictors and responses from the reference scenario. The performance decreases slightly as the correlation among variables increases, and more so for hotspot selection and when the response residual correlation is high. The global-local model continues to outperform the others, except for the setting $\rho_x \in (0.4, 0.6)$ - $\rho_y \in (0.6, 0.8)$, where it has an average hotspot selection performance slightly inferior to that of the ``global-scale-only'' model; yet the former is still largely superior to the latter in terms of pairwise selection.

We expect the dependence simulated for most scenarios reported in Tables  \ref{sm_tb_corr_grid} and \ref{sm_tb_corr_grid2} to be much higher than in real data. For instance, as shown in Figure \ref{sm_fig_corr_real_fairfax}, dependence in the raw transcript data of Section \ref{sec_application} is mostly negligible even if it can be substantial within transcript modules  (the $0.1\%-$ and $99.9\%-$quantiles are %
$-0.57$ and $0.66$). %
Likewise, for most SNPs sufficiently far appart will be uncorrelated, and the local correlation due to linkage disequilibrium seems mostly moderate (the $0.1\%$ and $99.9\%$ quantiles are %
$-0.42$ and $0.65$). 

\begin{table}[h!]
\begin{center}
\footnotesize
\begin{tabular}{rrlllllllll}
  \hline
 Model &$\mu_\omega \times 0.1$ & $\mu_\omega \times 0.5$ & $\mu_\omega \times 1$ &G & GL\\
  Predictor/response auto/equi-correlation&&&&\\
  \hline
$\rho_x \in (0, 0.2)$, $\rho_y \in (0.2, 0.4)$ & $44.4$ $(0.5)$ & $62.8$ $(0.9)$ & $78.4$ $(0.5)$ & $54.5$ $(0.8)$ & $\mathbf{90.3}$ $\mathbf{(0.4)}$\\ 
 $\rho_y \in (0.4, 0.6)$ & $44.4$ $(0.5)$ & $62.8$ $(0.8)$ & $76.7$ $(0.6)$ & $54.3$ $(0.8)$ & $\mathbf{90.0}$ $\mathbf{(0.4)}$ \\ 
 $\rho_y \in (0.6, 0.8)$ & $44.4$ $(0.5)$ & $62.7$ $(0.8)$ & $72.7$ $(0.7)$ & $54.1$ $(0.9)$ & $\mathbf{89.1}$ $\mathbf{(0.6)}$\\ 
$\rho_x \in (0.4, 0.6)$, $\rho_y \in (0.2, 0.4)$ & $44.2$ $(0.4)$ & $62.6$ $(0.6)$ & $77.8$ $(0.4)$ & $54.9$ $(1.0)$ & $\mathbf{90.0}$ $\mathbf{(0.5)}$\\ 
 $\rho_y \in (0.4, 0.6)$ & $44.3$ $(0.4)$ & $62.5$ $(0.6)$ & $76.4$ $(0.4)$ & $55.0$ $(0.9)$ & $\mathbf{89.5}$ $\mathbf{(0.4)}$\\ 
 $\rho_y \in (0.6, 0.8)$& $44.3$ $(0.4)$& $62.4$ $(0.6)$ & $72.6$ $(0.5)$ & $55.2$ $(0.9)$ & $\mathbf{88.3}$ $\mathbf{(0.6)}$ \\ 
$\rho_x \in (0.8, 1)$, $\rho_y \in (0.2, 0.4)$ & $44.3$ $(0.4)$ & $58.7$ $(0.6)$ & $77.9$ $(0.6)$ & $59.2$ $(0.5)$ & $\mathbf{88.3}$ $\mathbf{(0.6)}$ \\
 $\rho_y \in (0.4, 0.6)$ & $44.4$ $(0.4)$& $58.9$ $(0.6)$ & $76.6$  $(0.6)$& $59.0$ $(0.7)$ & $\mathbf{88.1}$ $\mathbf{(0.7)}$ \\
 $\rho_y \in (0.6, 0.8)$ & $44.4$ $(0.4)$ & $58.9$ $(0.6)$ & $73.7$ $(0.7)$ & $59.1$ $(0.6)$ & $\mathbf{87.6}$ $\mathbf{(0.7)}$ \\ 
   \hline
\end{tabular}
\caption{\small Predictor-response selection performance for a grid of correlation settings: average standardized partial areas under the curve $\times 100$ with false positive threshold $0.01$, each based on $16$ replicates; the remaining settings correspond to the ``reference'' case is displayed in Figure \ref{fig_s1_ann} of the paper. Standard errors are in parentheses and, for each scenario, the best average performance is in bold.}\label{sm_tb_corr_grid}
\end{center}
\end{table}

\begin{table}[h!]
\begin{center}
\small
\begin{tabular}{rrlllllllll}
  \hline
 Model & $ \mu_\omega \times 0.1$ & $ \mu_\omega \times 0.5$ & $\mu_\omega \times 1$& G & GL\\
 Predictor/response auto/equi-correlation&&&&\\
  \hline
$\rho_x \in (0, 0.2)$, $\rho_y \in (0.2, 0.4)$ & $65.5$ $(3.9)$ & $65.4$ $(3.8)$ & $62.9$ $(3.4)$ & $73.7$ $(2.3)$ & $\mathbf{80.5}$ $\mathbf{(1.9)}$\\ 
 $\rho_y \in (0.4, 0.6)$ & $55.2$ $(3.5)$& $55.2$ $(3.3)$ & $48.9$ $(3.1)$ & $70.5$ $(2.1)$ & $\mathbf{76.4}$ $\mathbf{(2.4)}$ \\ 
$\rho_y \in (0.6, 0.8)$ & $42.3$ $(3.3)$ & $41.1$ $(3.5)$ & $31.9$ $(4.1)$ & $\mathbf{57.6}$ $\mathbf{(3.3)}$ & $56.7$ $(4.3)$\\ 
$\rho_x \in (0.4, 0.6)$, $\rho_y \in (0.2, 0.4)$ & $63.3$ $(4.3)$ & $63.8$ $(4.0)$ & $61.6$ $(3.7)$& $74.7$ $(2.6)$ & $\mathbf{78.7}$ $\mathbf{(2.4)}$ 
 \\ 
 $\rho_y \in (0.4, 0.6)$ & $57.2$ $(4.2)$ & $57.2$ $(3.6)$ & $52.7$ $(3.1)$ & $70.9$ $(3.7)$ & $\mathbf{76.9}$ $\mathbf{(4.2)}$\\ 
 $\rho_y \in (0.6, 0.8)$& $45.1$ $(2.8)$ & $43.5$ $(3.1)$ & $33.6$ $(3.9)$ & $58.0$ $(3.4)$ & $\mathbf{59.9}$ $\mathbf{(3.1)}$\\
$\rho_x \in (0.8, 1)$, $\rho_y \in (0.2, 0.4)$ & $59.8$ $(4.0)$ & $63.5$ $(3.9)$& $64.6$ $(2.5)$& $73.7$ $(2.8)$ & $\mathbf{76.4}$ $\mathbf{(2.8)}$\\   
 $\rho_y \in (0.4, 0.6)$ & $51.9$ $(4.8)$ & $55.2$ $(5.1)$ & $53.6$ $(3.4)$& $71.2$ $(3.3)$ & $\mathbf{73.8}$ $\mathbf{(2.6)}$ \\ 
 $\rho_y \in (0.6, 0.8)$ & $41.8$ $(4.3)$ & $42.7$ $(4.0)$& $36.6$ $(3.2)$ & $56.7$ $(4.2)$& $\mathbf{60.0}$ $\mathbf{(3.3)}$ \\ 
   \hline
\end{tabular}
\caption{\small Predictor (hotspot) selection for a grid of correlation settings; see caption of Table \ref{sm_tb_corr_grid} for the remaining settings.}\label{sm_tb_corr_grid2}
\end{center}
\end{table}

\begin{figure}[h!]
\small
\centering
\includegraphics[scale=0.38]{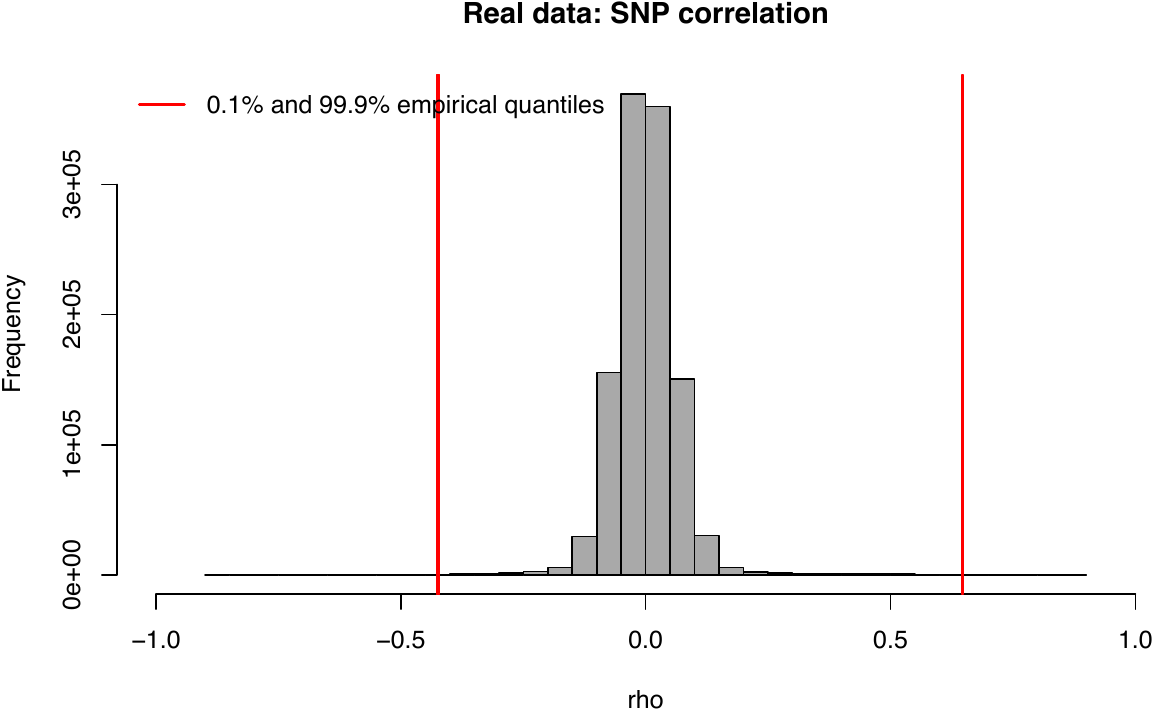}
\includegraphics[scale=0.38]{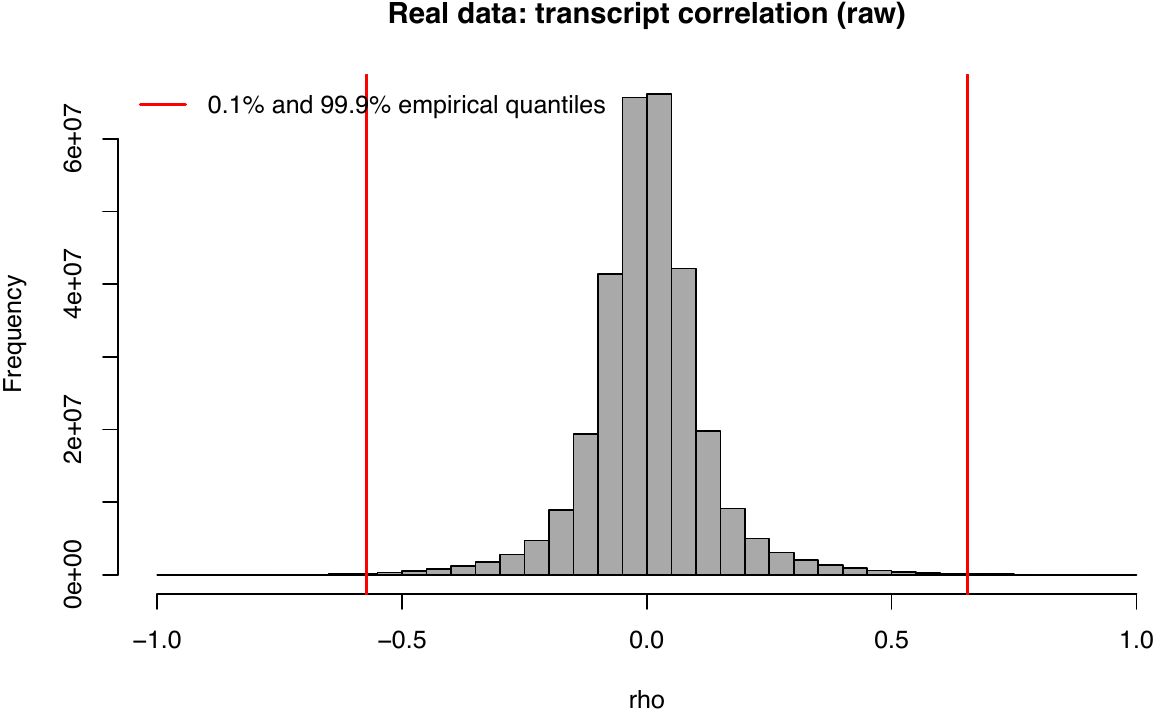}
\caption{\small Empirical correlation for the SNPs (left) and transcript expression levels (right) of the data used in Section \ref{sec_application}.
}\label{sm_fig_corr_real_fairfax}
\end{figure}

\newpage
\subsection{\,Simulation study 2: performance with and without simulated annealing}

Figure \ref{sm_fig_s2_ann} compares the performance of classical and annealed variational inferences on the data of Section~\ref{sec_s2}. 
\begin{figure}[h!]
\small
\centering
\includegraphics[scale=0.55]{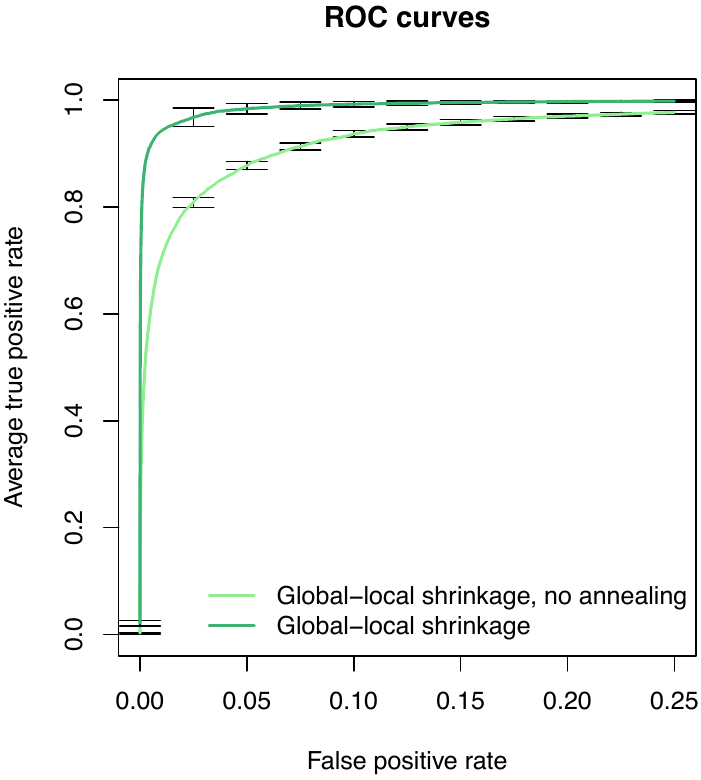}%
\caption{\small  Pairwise selection performance by classical variational algorithm and annealed variational algorithm. Truncated average receiver operating characteristic curves for hotspot selection with 95\% confidence intervals obtained from $16$ replications.}\label{sm_fig_s2_ann}
\end{figure}

\section{\,Stimulated eQTL analysis: overlap of transcripts associated with hotspot 
rs$6581889$ across conditions
}\label{app_real}

Figure \ref{app_venn} displays the overlap of transcripts found associated with hotspot rs$6581889$ in the application to monocyte eQTL data presented in Section \ref{sec_application}. Most associations are shared between the unstimulated and the IFN-gamma conditions, although $107$ associations are specific to the latter condition. Stimulations by LPS of 2 hours and 24 hours have no triggering effect for rs$6581889$. 
\begin{figure}[h!]
\small
\centering
\includegraphics[scale=0.35]{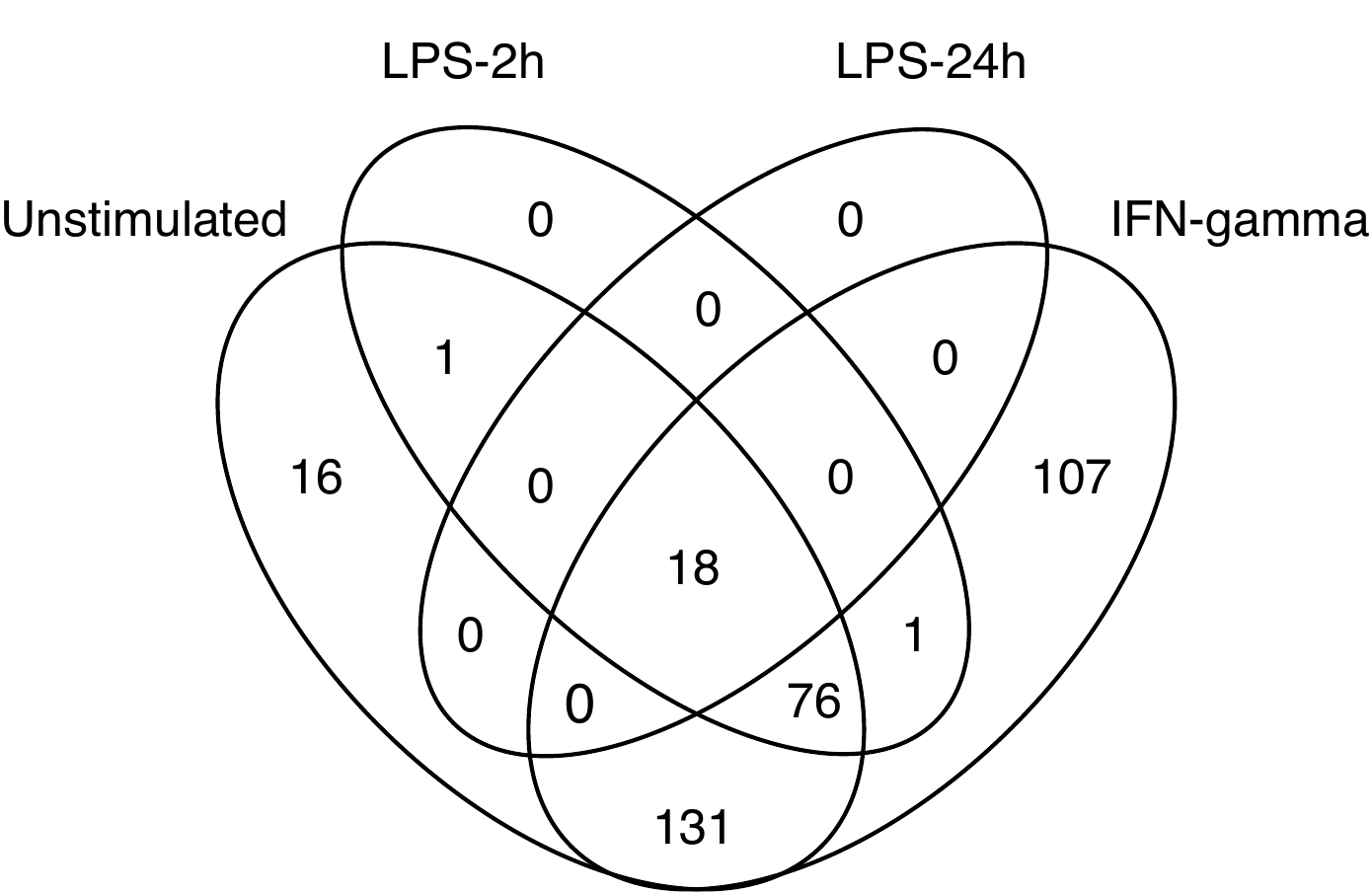}
\caption{\small Venn diagram for transcripts associated with hotspot rs$6581889$ across conditions.}\label{app_venn}
\end{figure}

\section{\,Glossary}\label{sm_gloss}
\begin{itemize}
\item {\bf Bayesian prior odds ratio}: for two hypotheses $H_1$ and $H_2$, the Bayesian prior odds ratio is $\pr(H_1) / \pr(H_2)$.\\
\item {\bf \emph{cis}- and \emph{trans}-acting}: genetic variants that act locally, affecting the levels of a nearby gene product, are said to be \emph{cis}-acting; genetic variants altering the levels of remote gene products are said to be \emph{trans}-acting.\\
\item {\bf global-local shrinkage factor}: parameter $\kappa_s$ which represents how much the hotspot parameter $\theta_s$ is shrunk towards the prior mean from the data mean, through the action of the horseshoe prior; see Lemma \ref{sm_lemma_mean}.\\
\item {\bf hotspot}: predictor associated with several response variables; in a molecular QTL setting, genetic variants regulating several molecular/clinical traits, which may be responsible for important functional mechanisms underlying complex traits.\\
\item {\bf hotspot propensity}: degree of pleiotropy of a hotspot, i.e., its propensity to influence multiple traits at once; hotspot propensity parameter: $\theta_s$, see model specification (\ref{eq_all}).\\
\item {\bf hotspot size}: number of responses/traits associated with a given hotspot.\\
\item {\bf linkage disequilibrium}: block dependence structures among variants along the genome resulting from the nonrandom assortment of alleles at two or more polymorphisms on a chromosome.\\
\item {\bf molecular QTL study  (eQTL, pQTL, mQTL)}: molecular quantitative trait locus study aims to uncover associations between genetic variants and molecular levels such as gene expression (eQTL), protein expression (pQTL) or metabolite levels (mQTL).\\
\item {\bf monocyte conditions}: the different types of stimulations performed on the monocytes (see Section \ref{sec_application}), i.e.,  they were exposed to the inflammation proxies interferon-$\gamma$ (IFN-$\gamma$), to differing durations of lipopolysaccharide (LPS 2h or LPS 24h), or they were left unstimulated.\\
\item {\bf pile-up}: artifactual hotspot effect caused by the lack of adjustment for the response dimension in flat likelihood cases; see Section \ref{sec_ps}. \\
\item {\bf pleiotropy}: effect of a hotspot genetic variant, i.e., regulation of several molecular/clinical traits by a single variant. \\
\item {\bf SNP}: single nucleotide polymorphism, variation in the nucleotide, A, T, G, or C, that is present to some appreciable extent in a population (the minor allele of this variation has frequency $> 1\%$ or $>5\%$).\\
\end{itemize}

\bibliographystyle{plainnat}
\bibliography{literature_ref}  

\end{document}